\documentclass[11pt]{article}
\usepackage[latin9]{inputenc}
\usepackage{geometry}
\usepackage{array}
\usepackage{float}
\usepackage{multirow}
\usepackage{amsmath}
\usepackage{amssymb}
\usepackage{graphicx}
\usepackage{setspace}
\usepackage{esint}
\usepackage[authoryear]{natbib}
\usepackage[unicode=true,pdfusetitle,
 bookmarks=true,bookmarksnumbered=false,bookmarksopen=false,
 breaklinks=false,pdfborder={0 0 1},backref=section,colorlinks=false]
 {hyperref}
\hypersetup{
 colorlinks,citecolor=blue,pdftex}

\makeatletter

\providecommand{\tabularnewline}{\\}

\usepackage{geometry}
\usepackage{amsmath}
\usepackage{amssymb}
\usepackage{amsfonts}
\usepackage{amsmath, amssymb,bm, amsthm, stmaryrd}
\usepackage{lscape}
\usepackage{appendix}
\usepackage{graphicx}
\usepackage{setspace}
\usepackage{dcolumn}
\usepackage[authoryear]{natbib}
\usepackage{rotating}
\usepackage{comment}
\usepackage{color}
\usepackage{breakurl}

\usepackage{xr}

\usepackage{lscape}
\newcounter{Acnt}

\newcounter{Acntp}
\newtheorem{Ap-item}[Acntp]{ }

\newtheorem{assumption}[Acnt]{Assumption}

\newtheorem{theorem}{Theorem}[section]
\newtheorem{proposition}{Proposition}[section]

\newtheorem{definition}[theorem]{Definition}

\newtheorem{lemma}{Lemma}[section]

\newtheorem*{as1-2}{Assumption 1$^{\prime}$}

\numberwithin{equation}{section}
\numberwithin{theorem}{section}

\usepackage{mathtools}

\makeatother

\begin{document}
\title{Estimation in a Generalization of Bivariate Probit Models with Dummy
Endogenous Regressors\thanks{The authors thank Jason Abrevaya, Xiaohong Chen, Stephen Donald, Brendan
Kline, Ed Vytlacil, and Haiqing Xu for valuable discussions. An earlier
version of this paper has been circulated under the title ``Sensitivity
Analysis in Triangular Systems of Equations with Binary Endogenous
Variables.''}}
\author{Sukjin Han\\
 Department of Economics\\
 University of Texas at Austin \\
 \href{mailto:sukjin.han@austin.utexas.edu}{sukjin.han@austin.utexas.edu}\and
Sungwon Lee\\
 Global Asia Institute \\
National University of Singapore \\
 \href{mailto:gails@nus.edu.sg}{gails@nus.edu.sg}}
\date{First Draft: March 19, 2015 \\
 This Draft: March 4, 2019}
\maketitle
\begin{abstract}
The purpose of this paper is to provide guidelines for empirical researchers
who use a class of bivariate threshold crossing models with dummy
endogenous variables. A common practice employed by the researchers
is the specification of the joint distribution of the unobservables
as a bivariate normal distribution, which results in a \textit{bivariate
probit model}. To address the problem of misspecification in this
practice, we propose an easy-to-implement semiparametric estimation
framework with parametric copula and nonparametric marginal distributions.
We establish asymptotic theory, including root-$n$ normality, for
the sieve maximum likelihood estimators that can be used to conduct
inference on the individual structural parameters and the average
treatment effect (ATE). In order to show the practical relevance of
the proposed framework, we conduct a sensitivity analysis via extensive
Monte Carlo simulation exercises. The results suggest that the estimates
of the parameters, especially the ATE, are sensitive to parametric
specification, while semiparametric estimation exhibits robustness
to underlying data generating processes. We then provide an empirical
illustration where we estimate the effect of health insurance on doctor
visits. In this paper, we also show that the absence of excluded instruments
may result in identification failure, in contrast to what some practitioners
believe. \\


\end{abstract}

\section{\label{sec:Introduction}Introduction}

The purpose of this paper is to provide guidelines for empirical researchers
who use a class of bivariate threshold crossing models with dummy
endogenous variables. This class of models is typically written as
follows. With the binary outcome $Y$ and the observed binary endogenous
treatment $D$, we consider
\begin{equation}
\begin{array}{r}
Y=\mathbf{1}[X^{\prime}\beta+\delta_{1}D-\varepsilon\geq0],\\
D=\mathbf{1}[X^{\prime}\alpha+Z^{\prime}\gamma-\nu\geq0],
\end{array}\label{eq:model}
\end{equation}
where $X$ denotes a vector of exogenous regressors that determine
both $Y$ and $D$, and $Z$ denotes a vector of exogenous regressors
that directly affect $D$, but not $Y$ (i.e., instruments for $D$).
Since $Y$ does not appear in the equation for $D$, this model forms
a triangular model, as a special case of a simultaneous equations
model, with the binary endogenous variables. In this paper, we investigate
the consequences of the common practices employed by empirical researchers
who use this class of models. As an important part of this investigation,
we conduct a sensitivity analysis on the specification of the joint
distribution of the unobservables $(\varepsilon,\nu)$. This is the
component of the model that practitioners have the least knowledge
about, and thus typically impose a parametric assumption. To address
the problem of misspecification, we propose a semiparametric estimation
framework with parametric copula and nonparametric marginal distributions.
The semiparametric specification is an attempt to ensure robustness
while achieving point identification and efficient estimation.

The parametric class of models \eqref{eq:model} includes the \textit{bivariate
probit model,} in which the joint distribution of $(\varepsilon,\nu)$
is assumed to be a bivariate normal distribution. This model has been
widely used in empirical research, including the works of \citet{Evans_Schwab_1995},
\citet{Neal_1997_catholic}, \citet{goldman2001eim}, \citet{altonji2005evaluation},
\citet{bhattacharya2006estimating}, \citet{rhine2006importance}
and \citet{marra2011estimation} to name a just few. The distributional
assumption in this model, however, is made out of convenience or convention,
and is hardly justified by underlying economic theory and thus susceptible
to misspecification. With binary endogenous regressors, the objects
of interest in model \eqref{eq:model} are the mean treatment parameters,
in addition to the individual structural parameters. Because the outcome
variable is also binary, the mean treatment parameters such as the
average treatment effect (ATE) are expressed as the differential between
the marginal distributions of $\varepsilon$. Therefore, the problem
of misspecification when estimating these treatment parameters can
be even more severe than that when estimating individual parameters.

To one extreme, a nonparametric joint distribution of $(\varepsilon,\nu)$
can be used in a bivariate threshold crossing model, as in \citet{shaikh2011partial}.
Their results, however, suggest that the ATE is only partially identified
in this fully flexible setting. Instead of sacrificing point identification,
we impose a parametric assumption on the dependence structure between
the unobservables using copula functions that are known up to a scalar
parameter. At the same time, in order to ensure robustness, we allow
the marginal distribution of $\varepsilon$ (and $\nu$), which is
involved in the calculation of the ATE, to be unspecified. Our class
of models encompasses both parametric and semiparametric models with
parametric copula and either parametric or nonparametric marginal
distributions. This broad range of models allows us to conduct a sensitivity
analysis on the specification of the joint distribution of $(\varepsilon,\nu)$.

The identification of the individual parameters and the ATE in this
class of models is established in \citet[hereafter, HV17]{HV17}.
They show that when the copula function for $(\varepsilon,\nu)$ satisfies
a certain stochastic ordering, identification is achieved in both
parametric and semiparametric models under an exclusion restriction
and mild support conditions. Building on these results, we consider
estimation and inference in the same setting. For the semiparametric
class of models \eqref{eq:model} with parametric copula and nonparametric
marginal distributions, the likelihood contains infinite-dimensional
parameters (i.e., the unknown marginal distributions). To estimate
this model, we consider the sieve maximum likelihood (ML) estimation
method for the finite- and infinite-dimensional parameters of the
model, as well as their functionals. The estimation of the parametric
model, on the other hand, is within the standard ML framework.

The contributions of this paper can be summarized as follows. Through
these contributions, this paper is intended to provide a guideline
to empirical researchers. First, we establish the asymptotic theory
for the sieve ML estimators in a class of semiparametric copula-based
models. This result can be used to conduct inference on the functionals
of the finite- and infinite-dimensional parameters, such as inference
on the individual structural parameters and the ATE. We show that
the sieve ML estimators are consistent and that their smooth functionals
are root-$n$ asymptotically normal.

Second, in order to show the practical relevance of the theoretical
results for empirical researchers, we conduct a sensitivity analysis
via extensive Monte Carlo simulation exercises. We find that the parametric
ML estimates, especially those for the ATE, can be highly sensitive
to the misspecification of the marginal distributions of the unobservables.
On the other hand, the sieve ML estimates perform well in terms of
the mean squared error (MSE) as they are robust to the underlying
data generating process. Moreover, their performance is comparable
to that of the parametric estimates under a correct specification.
We also show that copula misspecification does not have a substantial
effect in estimation, as long as the true copula is within the stochastic
ordering class of the identification. As copula misspecification is
a problem common to both parametric and semiparametric models considered
in this paper, our sensitivity analysis suggests that a semiparametric
consideration may be more preferable in estimation and inference.

Third, we provide an empirical illustration of the sieve estimation
and  the sensitivity analysis of this paper. We estimate the effect
of health insurance on decisions to visit doctors using the the Medical
Expenditure Panel Survey data combined with the National Compensation
Survey data by matching industry types. We compare the estimates of
parametric and semiparametric bivariate threshold crossing models
with the Gaussian copula. We show that the estimates differ, especially
so for the estimated ATE's, which suggest the misspecification of
the marginal distribution of the unobservables, consistent with the
simulation results. In other words, the estimates of the bivariate
probit model can be misleading in this example.

Fourth, we formally show that identification may fail without the
exclusion restriction, in contrast to the findings of \citet{wilde2000identification}.
The bivariate probit model is sometimes used in applied work without
instruments (e.g., \citet{white_wolaver2003} and \citet{rhine2006importance}).
We show, however, that this restriction is not only sufficient but
also necessary for identification in parametric and semiparametric
models when there is a single binary exogenous variable common to
both equations. We also show that under joint normality of the unobservables,
the parameters are, at best, weakly identified when there are common
(and possibly continuous) exogenous variables. \footnote{HV17 only show the sufficiency of this restriction for identification.
\citet{mourifie2014note} show the necessity of the restriction, but
their argument does not exploit all information available in the model;
see Section 2.2 of the present paper for further details.} We also note that another source of identification failure is the
absence of restrictions on the dependence structure of the unobservables,
as mentioned above.

The sieve estimation method is a useful nonparametric estimation framework
that allows for a flexible specification, while guaranteeing the tractability
of the estimation problem; see \citet{chen2007large} for a survey
of sieve estimation in semi-nonparametric models. The estimation method
is also easy to implement in practice. The sieve ML estimation has
been used in various contexts: \citet[hereafter, CFT06]{cft06} consider
the sieve estimation of semiparametric multivariate distributions
that are modeled using parametric copulas; \citet{bierens2008semi}
applies the estimation method to the mixed proportional hazard model;
and \citet{hu_schennach2008} and \citet{chen2009nonparametric} use
the method to estimate nonparametric models with non-classical measurement
errors. The asymptotic theory developed in this paper is based on
the results established in the sieve extremum estimation literature
(e.g., CFT06; \citet{chen2007large}; \citet{bierens2014consistency}).
A semiparametric version of bivariate threshold crossing models is
also considered in \citet{marra2011estimation} and \citet{ieva2014semiparametric}.
In contrast to our setting, however, they introduce flexibility for
the index function of the threshold, and not for the distribution
of the unobservables.

The remainder of this paper is organized as follows. The next section
reviews the identification results of HV17, and then discusses the
lack of identification in the absence of exclusion restrictions and
in the absence of restrictions on the dependence structure of the
unobservables. Section \ref{sec:Estimation} introduces the sieve
ML estimation framework for the semiparametric class of models defined
in \eqref{eq:model}, and Section \ref{sec:Asymptotic} establishes
the large sample theory for the sieve ML estimators. The sensitivity
analysis is conducted in Section \ref{sec:Monte_Carlo} by investigating
the finite sample performance of the parametric ML and sieve ML estimates
under various specifications. Section \ref{sec:Empirical-Example}
presents the empirical example, and Section \ref{sec:Conclusions}
concludes.

\section{\label{sec:Identification}Identification and Failure of Identification}

\subsection{Identification Results in Han and Vytlacil (2017)}

We first summarize the identification results in HV17. In model \eqref{eq:model},
let $\underset{(k+1)\times1}{X}\equiv(1,X_{1},...,X_{k})^{\prime}$
and $\underset{l\times1}{Z}\equiv(Z_{1},...,Z_{l})^{\prime}$, and
conformably, let $\alpha\equiv(\alpha_{0},\alpha_{1},...,\alpha_{k})'$,
$\beta\equiv(\beta_{0},\beta_{1},...,\beta_{k})'$, and $\gamma\equiv(\gamma_{1},\gamma_{2},...,\gamma_{l})'$.

\begin{assumption}\label{as:independence} $X$ and $Z$ satisfy
that $(X,Z)\perp(\varepsilon,\nu)$, where ``$\perp$'' denotes
statistical independence. \end{assumption}

\begin{assumption}\label{assumption:xz}$(X^{\prime},Z^{\prime})$
does not lie in a proper linear subspace of $\mathbb{R}^{k+l}$ a.s.\footnote{A proper linear subspace of $\mathbb{R}^{k+l}$ is a linear subspace
with a dimension strictly less than $k+l$. The assumption is that
if $M$ is a proper linear subspace of $\mathbb{R}^{k+l}$, then $\Pr[(X^{\prime},Z^{\prime})\in M]<1$.}\end{assumption}

\begin{assumption}\label{as:joint_dist} There exists a copula function
$C:(0,1)^{2}\rightarrow(0,1)$ such that the joint distribution $F_{\varepsilon\nu}$
of $(\varepsilon,\nu)$ satisfies $F_{\varepsilon\nu}(\varepsilon,\nu)=C(F_{\varepsilon}(\varepsilon),F_{\nu}(\nu))$,
where $F_{\varepsilon}$ and $F_{\nu}$ are the marginal distributions
of $\varepsilon$ and $\nu$, respectively, that are strictly increasing
and absolutely continuous with respect to Lebesgue measure.\footnote{ Sklar's theorem (e.g., \citet{nelsen1999introduction}) guarantees
the existence of such a copula, which is, in fact, unique because
$F_{\varepsilon}$ and $F_{\nu}$ are continuous.}\end{assumption}

\begin{assumption}\label{as:normalize} As scale and location normalizations,
$\alpha_{1}=\beta_{1}=1$ and $\alpha_{0}=\beta_{0}=0$. \end{assumption}

A model with alternative scale and location normalizations, $Var(\varepsilon)=Var(\nu)=1$
and $E[\varepsilon]=E[\nu]=0$, can be viewed as a reparametrized
version of the model with the normalizations given in Assumption \ref{as:normalize};
see, for example, the reparametrization \eqref{eq:repara} below.
For $x\in\text{supp}(X)$ and $z\in\text{supp}(Z)$, write a one-to-one
map (by Assumption \ref{as:joint_dist}) as
\begin{align}
s_{xz} & \equiv F_{\nu}(x^{\prime}\alpha+z^{\prime}\gamma),\quad r_{0,x}\equiv F_{\varepsilon}(x^{\prime}\beta),\quad r_{1,x}\equiv F_{\varepsilon}(x^{\prime}\beta+\delta_{1}).\label{eq:repara}
\end{align}
Take $(x,z)$ and $(x,\tilde{z})$, for some $x\in\text{supp}(X|Z=z)\cap\text{supp}(X|Z=\tilde{z})$,
where $\mbox{supp}(X|Z)$ is the conditional support of $X$, given
$Z$. Then, by Assumption \ref{as:independence}, model \eqref{eq:model}
implies that the fitted probabilities are written as 
\begin{equation}
\begin{array}{cc}
p_{11,xz}=C(r_{1,x},s_{xz}), & \quad p_{11,x\tilde{z}}=C(r_{1,x},s_{x\tilde{z}}),\\
p_{10,xz}=r_{0,x}-C(r_{0,x},s_{xz}), & \quad p_{10,x\tilde{z}}=r_{0,x}-C(r_{0,x},s_{x\tilde{z}}),\\
p_{01,xz}=s_{xz}-C(r_{1,x},s_{xz}), & \quad p_{01,x\tilde{z}}=s_{x\tilde{z}}-C(r_{1,x},s_{x\tilde{z}}),
\end{array}\label{eq:fitted}
\end{equation}
where $p_{yd,xz}\equiv\Pr[Y=y,D=d|X=x,Z=z]$ for $(y,d)\in\{0,1\}^{2}$.
The equation \eqref{eq:fitted} serves as the basis for the identification
and estimation of the model. Depending upon whether one is willing
to impose an additional assumption on the dependence structure of
the unobservables $(\varepsilon,\nu)$ via $C(\cdot,\cdot)$, the
underlying parameters of the model are either point identified or
partially identified.

We first consider point identification. The results for point identification
can be found in HV17, which we adapt here given Assumption \ref{as:normalize}.
The additional dependence structure can be characterized in terms
of the stochastic ordering of the copula parametrized with a scalar
parameter.

\begin{definition}[Strictly More SI or Less SD]\label{def_moreSI}Let
$C(u_{2}|u_{1})$ and $\tilde{C}(u_{2}|u_{1})$ be conditional copulas,
for which $1-C(u_{2}|u_{1})$ and $1-\tilde{C}(u_{2}|u_{1})$ are
either increasing or decreasing in $u_{1}$ for all $u_{2}$. Such
copulas are referred to as stochastically increasing (SI) or stochastically
decreasing (SD), respectively. Then, $\tilde{C}$ is strictly more
SI (or less SD) than $C$ if $\psi(u_{1},u_{2})\equiv\tilde{C}^{-1}(C(u_{2}|u_{1})|u_{1})$
is strictly increasing in $u_{1}$,\footnote{Note that $\psi(u_{1},u_{2})$ is increasing in $u_{2}$ by definition.}
which is denoted as $C\prec_{S}\tilde{C}$.\end{definition}

This ordering is equivalent to having a ranking in terms of the first
order stochastic dominance. Let $(U_{1},U_{2})\sim C$ and $(\tilde{U}_{1},\tilde{U}_{2})\sim\tilde{C}$.
When $\tilde{C}$ is strictly more SI (less SD) than $C$ is, then
$\Pr[\tilde{U}_{2}>u_{2}|\tilde{U}_{1}=u_{1}]$ increases even more
than $\Pr[U_{2}>u_{2}|U_{1}=u_{1}]$ does as $u_{1}$ increases.\footnote{In the statistics literature, the SI dependence ordering is also referred
to as the (strictly) ``more regression dependent'' or ``more monotone
regression dependent'' ordering; see \citet{Joe1997} for details.}

\begin{assumption}\label{as_copula}The copula in Assumption \ref{as:joint_dist}
satisfies $C(\cdot,\cdot)=C(\cdot,\cdot;\rho)$ with a scalar dependence
parameter $\rho\in\Omega$, is twice differentiable in $u_{1}$, $u_{2}$
and $\rho$, and satisfies 
\begin{equation}
C(u_{1}|u_{2};\rho_{1})\prec_{S}C(u_{1}|u_{2};\rho_{2})\mbox{ for any }\rho_{1}<\rho_{2}.\label{ssi}
\end{equation}
\end{assumption}

The meaning of the last part of this assumption is that the copula
is ordered in $\rho$ in the sense of the stochastic ordering defined
above. This requirement defines the class of copulas that we allow
for identification. Many well-known copulas satisfy \eqref{ssi}:
the normal copula, Plackett copula, Frank copula, Clayton copula,
among many others; see HV17 for the full list of copulas and their
expressions. Under these assumptions, we first discuss the identification
in a fully parametric model. 

\begin{assumption}\label{as:para_dist}$F_{\varepsilon}$ and $F_{\nu}$
are \textit{known} up to means $\mu\equiv(\mu_{\varepsilon},\mu_{\nu})$
and variances $\sigma^{2}\equiv(\sigma_{\varepsilon}^{2},\sigma_{\nu}^{2})$.\end{assumption}

Given this assumption, $F_{\nu}(\nu)=F_{\tilde{\nu}}(\tilde{\nu})$
and $F_{\varepsilon}(\varepsilon)=F_{\tilde{\varepsilon}}(\tilde{\varepsilon})$,
where $F_{\tilde{\nu}}$ and $F_{\tilde{\varepsilon}}$ are the distributions
of $\tilde{\nu}\equiv(\nu-\mu_{\nu})/\sigma_{\nu}$ and $\tilde{\varepsilon}\equiv(\varepsilon-\mu_{\varepsilon})/\sigma_{\varepsilon}$,
respectively. Define
\[
\mathcal{X}\equiv\bigcup_{\substack{z^{\prime}\gamma\neq\tilde{z}^{\prime}\gamma\\
z,\tilde{z}\in\text{supp}(Z)
}
}\text{supp}(X|Z=z)\cap\text{supp}(X|Z=\tilde{z}).
\]

\begin{theorem}\label{thm_full} In model \eqref{eq:model}, suppose
Assumptions \ref{as:independence}\textendash \ref{as:para_dist}
hold. Then, $(\alpha{}^{\prime},\beta{}^{\prime},\delta_{1},\gamma,\rho,\mu,\sigma)$
are point identified in an open and convex parameter space if (i)
$\gamma$ is a nonzero vector, and (ii) $\mathcal{X}$ does not lie
in a proper linear subspace of $\mathbb{R}^{k}$ a.s. \end{theorem}

The proof of this theorem is a minor modification of the proof of
Theorem 5.1 in HV17.

Although the parametric structure on the copula is necessary for the
point identification of the parameters, HV17 show that the parametric
assumption for $F_{\varepsilon}$ and $F_{\nu}$ are not necessary.
In addition, if we make a large support assumption, we can also identify
the nonparametric marginal distributions $F_{\varepsilon}$ and $F_{\nu}$.

\begin{assumption}\label{as:add_support}(i) The distributions of
$X_{j}$ (for $1\leq j\leq k$) and $Z_{j}$ (for $1\leq j\leq l$)
are absolutely continuous with respect to Lebesgue measure; (ii) There
exists at least one element $X_{j}$ in $X$ such that its support
conditional on $(X_{1},...,X_{j-1},X_{j+1},...,X_{k})$ is $\mathbb{R}$
and $\alpha_{j}\neq0$ and $\beta_{j}\neq0$, where, without loss
of generality, we let $j=1$. \end{assumption}

\begin{theorem}\label{thm_full_unknown} In model \eqref{eq:model},
suppose Assumptions \ref{as:independence}\textendash \ref{as_copula},
and \ref{as:add_support}(i) hold. Then $(\alpha{}^{\prime},\beta{}^{\prime},\delta_{1},\gamma,\rho)$
are point identified in an open and convex parameter space if (i)
$\gamma$ is a nonzero vector; and (ii) $\mathcal{X}$ does not lie
in a proper linear subspace of $\mathbb{R}^{k}$ a.s. In addition,
if Assumption \ref{as:add_support}(ii) holds, $F_{\varepsilon}(\cdot)$
and $F_{\nu}(\cdot)$ are identified up to additive constants.\end{theorem}

An interesting function of the underlying parameters that are point
identified under the parametric and semiparametric distributional
assumptions is the conditional ATE:
\begin{align}
ATE(x) & =E[Y_{1}-Y_{0}|X=x]=F_{\varepsilon}(x^{\prime}\beta+\delta_{1})-F_{\varepsilon}(x^{\prime}\beta).\label{eq:CATE}
\end{align}

\subsection{Extension of Han and Vytlacil (2017): Identification under Conditional
Independence}

The identification analysis of \citet{HV17} relies on the full independence
assumption (Assumption \ref{as:independence}) for $(X,Z)$. The analysis,
however, can be easily extended to a case where conditional independence
is alternatively assumed. Since this is a more empirically relevant
situation, we explore this case in detail here. In the empirical section
below, we impose the conditional independence. Let $W$ be a vector
of (potentially endogenous) covariates in $\text{supp}(W)$.

\begin{as1-2}\label{indep2}$X$ and $Z$ satisfy that $(X,Z)\perp(\varepsilon,\nu)|W$.
\end{as1-2}

Similarly, we modify Assumptions \ref{assumption:xz}\textendash \ref{as:joint_dist},
\ref{as_copula}\textendash \ref{as:add_support} accordingly. Then
the following theorems immediately hold by applying the same proof
strategies as in Theorems \ref{thm_full} and \ref{thm_full_unknown}.
Let $C_{w}(u_{1},u_{2})\equiv C(u_{1},u_{2}|W=w)$ be the conditional
copula, and $F_{\varepsilon\nu|w}(\varepsilon,\nu)\equiv F_{\varepsilon\nu|W=w}(\varepsilon,\nu)$,
$F_{\varepsilon|w}(\varepsilon)\equiv F_{\varepsilon|W=w}(\varepsilon)$
and $F_{\nu|w}(\nu)\equiv F_{\nu|W=w}(\nu)$ be the conditional distributions.
\begin{theorem}\label{thm_full-1} In model \eqref{eq:model}, suppose
Assumptions \ref{as:independence}$^{\prime}$ and \ref{as:normalize}
hold. Also, suppose Assumption \ref{assumption:xz} holds conditional
on $W$, and Assumptions \ref{as:joint_dist}, \ref{as_copula}\textendash \ref{as:para_dist}
hold with $C_{w}(u_{1},u_{2})$, $F_{\varepsilon\nu|w}(\varepsilon,\nu)$,
$F_{\varepsilon|w}(\varepsilon)$ and $F_{\nu|w}(\nu)$ instead, for
all $w\in\text{supp}(W)$. Then, $(\alpha{}^{\prime},\beta{}^{\prime},\delta_{1},\gamma,\rho,\mu,\sigma)$
are point identified in an open and convex parameter space if (i)
$\gamma$ is a nonzero vector, and (ii) $\mathcal{X}$ does not lie
in a proper linear subspace of $\mathbb{R}^{k}$ a.s. conditional
on $W$.\end{theorem}

\begin{theorem}\label{thm_full_unknown-1} In model \eqref{eq:model},
suppose Assumptions \ref{as:independence}$^{\prime}$ and \ref{as:normalize}
hold. Also, suppose Assumptions \ref{assumption:xz} and \ref{as:add_support}(i)
hold conditional on $W$, and Assumptions \ref{as:joint_dist} and
\ref{as_copula} hold with $C_{w}(u_{1},u_{2})$, $F_{\varepsilon\nu|w}(\varepsilon,\nu)$,
$F_{\varepsilon|w}(\varepsilon)$ and $F_{\nu|w}(\nu)$ instead, for
all $w\in\text{supp}(W)$. Then $(\alpha{}^{\prime},\beta{}^{\prime},\delta_{1},\gamma,\rho)$
are point identified in an open and convex parameter space if (i)
$\gamma$ is a nonzero vector; and (ii) $\mathcal{X}$ does not lie
in a proper linear subspace of $\mathbb{R}^{k}$ a.s. In addition,
if Assumption \ref{as:add_support}(ii) holds conditional on $W$,
$F_{\varepsilon|w}(\cdot)$ and $F_{\nu|w}(\cdot)$ are identified
up to additive constants for all $w\in\text{supp}(W)$.\end{theorem}

\subsection{The Failures of Identification}

In this section, we discuss two sources of identification failure
in the class of models \eqref{eq:model}: the absence of exclusion
restrictions and the absence of restrictions on the dependence structure
of the unobservables $(\varepsilon,\nu)$.

\subsubsection{No Exclusion Restrictions}

There are empirical works where \eqref{eq:model} is used without
excluded instruments; see, e.g., \citet{white_wolaver2003} and \citet{rhine2006importance}.
Identification in these papers relies on the results of \citet{wilde2000identification},
who provides an identification argument by counting the number of
equations and unknowns in the system. Here, we show that this argument
is insufficient for identification. We show that without the excluded
instruments (i.e., when $\gamma=0$), the structural parameters are
not identified, even with a full parametric specification of the joint
distribution (Assumptions \ref{as_copula} and \ref{as:para_dist}).
The existence of common exogenous covariates $X$ in both equations
is not very helpful for identification in a sense that becomes clear
below.

Before considering the lack of identification in a general case with
possibly continuous $X_{1}$ in $X=(1,X_{1})$, we start the analysis
with binary $X_{1}$. \citet{mourifie2014note} show the lack of identification
when there is no excluded instrument in a bivariate probit model with
binary $X_{1}$. They, however, only provide a numerical counter-example.
Moreover, their analysis does not consider the full set of observed
fitted probabilities, and hence possibly neglects information that
could have contributed to the identification. Here, we provide an
analytical counter-example in a more general parametric class of model
\eqref{eq:model} that nests the bivariate probit model. We show that
$(\delta_{1},\rho,\mu_{\varepsilon},\sigma_{\varepsilon})$ are not
identified, even if the full set of probabilities are used. Note that
the reduced-form parameters $(\mu_{\nu},\sigma_{\nu})$ are always
identified from the equation for $D$, and $\alpha=\beta=(0,1)^{\prime}$
as a normalization using scalar $X_{1}$.

\begin{theorem}\label{thm:lack}In model \eqref{eq:model} with $X=(1,X_{1})$
where $X_{1}\in\mbox{supp}(X_{1})=\{0,1\}$, suppose that the assumptions
in Theorem \ref{thm_full} hold, except that $\gamma=0$. Then, there
exist two element-wise distinct sets of $(\delta_{1},\rho,\mu_{\varepsilon},\sigma_{\varepsilon})$
that generate the same observed data.\end{theorem}

In showing this lack-of-identification result, we find a counter-example
where the copula density induced by $C(u_{1},u_{2})$ is symmetric
around $u_{2}=u_{1}$ and $u_{2}=1-u_{1}$, and the density induced
by $F_{\varepsilon}$ is symmetric. Note that the bivariate normal
distribution, namely, the normal copula with normal marginals, satisfies
these symmetry properties. That is, \textit{in the bivariate probit
model with a common binary exogenous covariate and no excluded instruments,
the structural parameters are }not\textit{ identified. }

The proof of Theorem \ref{thm:lack} proceeds as follows. Under Assumption
\ref{as:normalize}, let 
\[
\begin{array}{cc}
q_{0}\equiv F_{\tilde{\nu}}(-\mu_{\nu}/\sigma_{\nu}), & \quad q_{1}\equiv F_{\tilde{\nu}}((1-\mu_{\nu})/\sigma_{\nu}),\\
t_{0}\equiv F_{\tilde{\varepsilon}}(-\mu_{\varepsilon}/\sigma_{\varepsilon}), & \quad t_{1}\equiv F_{\tilde{\varepsilon}}((1-\mu_{\varepsilon})/\sigma_{\varepsilon}).
\end{array}
\]
 Then, we have
\[
\begin{array}{cc}
\tilde{p}_{11,0}=C(F_{\tilde{\varepsilon}}(F_{\tilde{\varepsilon}}^{-1}(t_{0})+\delta_{1}),q_{0};\rho), & \quad\tilde{p}_{11,1}=C(F_{\tilde{\varepsilon}}(F_{\tilde{\varepsilon}}^{-1}(t_{1})+\delta_{1}),q_{1};\rho),\\
\tilde{p}_{10,0}=t_{0}-C(t_{0},q_{0};\rho), & \quad\tilde{p}_{10,1}=t_{1}-C(t_{1},q_{1};\rho),\\
\tilde{p}_{00,0}=1-t_{0}-q_{0}+C(t_{0},q_{0};\rho), & \quad\tilde{p}_{00,1}=1-t_{1}-q_{1}+C(t_{1},q_{1};\rho),
\end{array}
\]
where $\tilde{p}_{yd,x}\equiv\Pr[Y=y,D=d|X_{1}=x]$. We want to show
that, given $(q_{0},q_{1})$ which are identified from the reduced-form
equation, there are two distinct sets of parameter values $(t_{0},t_{1},\delta_{1},\rho)$
and $(t_{0}^{*},t_{1}^{*},\delta_{1}^{*},\rho^{*})$ (with $(t_{0},t_{1},\delta_{1},\rho)$
$\ne$ $(t_{0}^{*},t_{1}^{*},\delta_{1}^{*},\rho^{*})$) that generate
the same observed fitted probabilities $\tilde{p}_{yd,0}$ and $\tilde{p}_{yd,1}$
for all $(y,d)\in\{0,1\}^{2}$ under some choices of $C(u_{1},u_{2})$
and $F_{\varepsilon}$. The detailed proof can be found in the online
appendix.

One might argue that the lack of identification in Theorem \ref{thm:lack}
is due to the limited variation of $X$. Although this is a plausible
conjecture, this does not seem to be the case in the model considered
here.\footnote{In fact, in \citet{Heckman_1979_Econometrica}'s sample selection
model under normality, although identification fails with binary exogenous
covariates in the absence of the exclusion restriction, it is well
known that identification is achieved with continuous covariates by
exploiting the nonlinearity of the model (\citet{Vella98}).} We now consider a general case with possibly \textit{continuous}
$X_{1}$, and discuss what can be said about the existence of two
distinct sets of $(\beta,\delta_{1},\rho,\mu_{\varepsilon},\sigma_{\varepsilon})$
that generate the same observed data. To this end, define 
\begin{align*}
q(x) & \equiv F_{\tilde{\nu}}((x'\alpha-\mu_{\nu})/\sigma_{\nu}),\quad t(x)\equiv F_{\tilde{\varepsilon}}((x'\beta-\mu_{\varepsilon})/\sigma_{\varepsilon}).
\end{align*}
Then, 
\begin{align*}
p_{11,x} & =C(F_{\tilde{\varepsilon}}(F_{\tilde{\varepsilon}}^{-1}(t(x))+\delta_{1}),q(x);\rho),\\
p_{10,x} & =t(x)-C(t(x),q(x);\rho),\\
p_{00,x} & =1-t(x)-q(x)+C(t(x),q(x);\rho).
\end{align*}
Similar to the proof strategy for the binary $X_{1}$ case, we want
to show that, given $(\alpha,\mu_{\nu},\sigma_{\nu})$, there are
two distinct sets of parameter values $(\beta,\delta_{1},\rho,\mu_{\varepsilon},\sigma_{\varepsilon})$
and $(\beta^{*},\delta_{1}^{*},\rho^{*},\mu_{\varepsilon}^{*},\sigma_{\varepsilon}^{*})$
that generate the same observed fitted probabilities $p_{yd,x}$ for
all $(y,d)\in\{0,1\}^{2}$ and $x\in\mbox{supp}(X)$ under some choices
of $C(u_{1},u_{2})$ and $F_{\varepsilon}$.

Let $t(x)\equiv F_{\tilde{\varepsilon}}(x'\beta)\in(0,1)$ for all
$x$ and for some $\beta$. Also, choose $\delta_{1}=0$ and some
$\rho\in\Omega$. For $\rho^{*}>\rho$, we want to show that there
exists $(\beta^{*},\delta_{1}^{*})$ such that, for $t^{*}(x)\equiv F_{\tilde{\varepsilon}}(x'\beta^{*})$,
\begin{align}
p_{10,x} & =t(x)-C(t(x),q(x);\rho)=t^{*}(x)-C(t^{*}(x),q(x);\rho^{*})\label{eq:conti_x1}\\
p_{11,x} & =C(F_{\tilde{\varepsilon}}(F_{\tilde{\varepsilon}}^{-1}(t(x))+0),q(x);\rho)=C(s^{\dagger}(x),q(x);\rho^{*})\label{eq:conti_x2}
\end{align}
for all $x$, where 
\begin{equation}
s^{\dagger}(x)=F_{\tilde{\varepsilon}}(F_{\tilde{\varepsilon}}^{-1}(t^{*}(x))+\delta_{1}^{*}).\label{eq:conti_x3}
\end{equation}
The question is whether we find $(\beta,\delta_{1},\rho)$ and $(\beta^{*},\delta_{1}^{*},\rho^{*})$
such that \eqref{eq:conti_x1}\textendash \eqref{eq:conti_x3} hold
simultaneously. First, note that, since $\rho^{*}>\rho$, we have
$t^{*}>t$ and hence $\beta^{*}\neq\beta$ by the assumption that
there is no linear subspace in the space of $X$. Now, choose $C(\cdot,\cdot;\rho)$
to be a normal copula and choose $\rho=0$ and $\rho^{*}=1$. Then,
using arguments similar to those of the binary case (found in the
online appendix), we obtain 
\begin{equation}
t^{*}(x)=q(x)+(1-q(x))t(x)\label{eq:conti_x4}
\end{equation}
and $s^{\dagger}(x)=q(x)t(x)$. Then, \eqref{eq:conti_x3} can be
rewritten as 
\begin{align}
\delta_{1}^{*} & =F_{\tilde{\varepsilon}}^{-1}(s^{\dagger}(x))-F_{\tilde{\varepsilon}}^{-1}(t^{*}(x))=F_{\tilde{\varepsilon}}^{-1}(q(x)t(x))-F_{\tilde{\varepsilon}}^{-1}(q(x)+(1-q(x))t(x)).\label{eq:conti_x5}
\end{align}

The complication here is to ensure that this equation is satisfied
for all $x$. Note that \eqref{eq:conti_x4} and \eqref{eq:conti_x5}
are consistent with the definition of a distribution function of a
continuous r.v.: $F_{\tilde{\varepsilon}}(+\infty)=1$, $F_{\tilde{\varepsilon}}(-\infty)=0$,
and $F_{\tilde{\varepsilon}}(\varepsilon)$ is strictly increasing.
We can then numerically show that a distribution function that is
close to a normal distribution satisfies the conditions with a particular
choice of $(\beta^{*},\delta_{1}^{*})$; see Figure 1.
\begin{figure}
\begin{centering}
\includegraphics[scale=0.2]{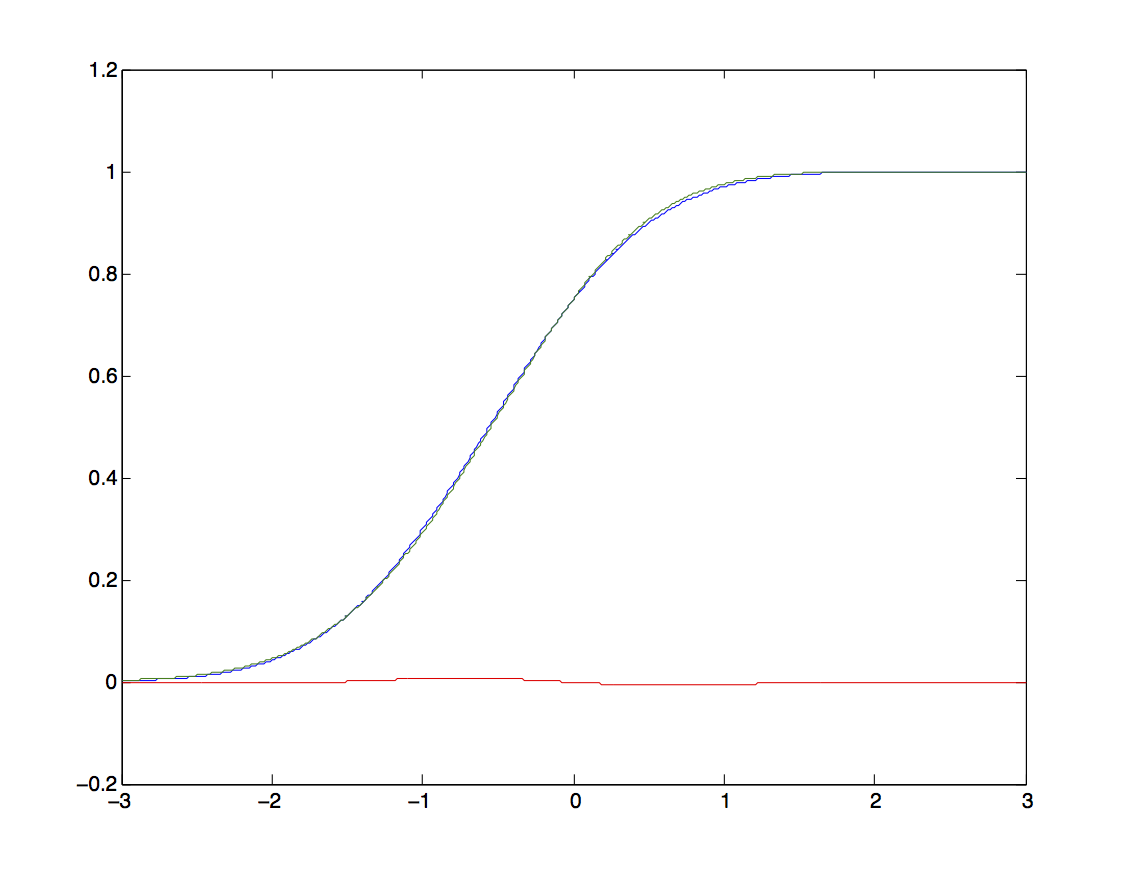}
\par\end{centering}
\caption{A numerical calculation of a distribution function under which identification
fails (blue line), compared with a normal distribution function (green
line).}
\end{figure}
 Although no formal derivation of the counterexample is given, this
result suggests the following:
\begin{itemize}
\item [(i)]In the bivariate probit model with continuous common exogenous
covariates and no excluded instruments, the parameters will be, \textit{at
best,} weakly identified; 
\item [(ii)]This also implies that, in the semiparametric model considered
in Theorem \ref{thm_full_unknown}, the structural parameters and
the marginal distributions are not identified without an exclusion
restriction, even if $X_{1}$ has large support.
\end{itemize}

\subsubsection{No Restrictions on Dependence Structures}

When the restriction imposed on $C(\cdot,\cdot)$ (i.e., Assumption
\ref{as_copula}) is completely relaxed, the underlying parameters
of model \eqref{eq:model} may fail to be identified, regardless of
whether the exclusion restriction holds. That is, a structure describing
how the unobservables $(\varepsilon,\nu)$ are dependent on each other
is necessary for identification. This is closely related to the results
in the literature that the treatment parameters (which are lower dimensional
functions of the individual parameters) in triangular models similar
to \eqref{eq:model} are only partially identified without distributional
assumptions; see \citet{bhattacharya2008treatment}, \citet{Chiburis},
\citet{shaikh2011partial}, and \citet{mourifie2015sharp}.

Suppose Assumptions \ref{as:independence}\textendash \ref{as:normalize}
hold. Then the model becomes a semiparametric threshold crossing model
in that the joint distribution is completely unspecified. Then, as
a special case of \citet{shaikh2011partial}, one can easily derive
bounds for the ATE $F_{\varepsilon}(x^{\prime}\beta+\delta_{1})-F_{\varepsilon}(x^{\prime}\beta)$.
The sharpness of these bounds is shown in their paper under a rectangular
support assumption for $(X,Z)$, which is, in turn, relaxed in \citet{mourifie2015sharp}.
In addition, using Assumption \ref{as:para_dist}, one can also derive
bounds for the individual parameters $x^{\prime}\beta$ and $\delta_{1}$,
as shown in \citet{Chiburis}. When there are no excluded instruments
in the model, \citet{Chiburis} shows that the bounds on the ATE do
not improve on the bounds of Manski (1990), whose argument applies
to the individual parameters.

\section{\label{sec:Estimation}Sieve and Parametric ML Estimations}

Based on the identification results, we now consider estimation. Let
$\psi\equiv(\alpha^{'},\beta^{'},\delta_{1},\gamma,\rho)$ denote
the vector of the structural individual parameters. Let $f_{\epsilon}$
and $f_{\nu}$ be the density functions associated with the distribution
functions $F_{\epsilon}$ and $F_{\nu}$, respectively, of the unobservables.
Then, $(\psi^{'},f_{\epsilon},f_{\nu})^{'}$ is the set of parameters
in the \textit{semiparametric version} of the model. The model becomes
fully parametric, once the infinite-dimensional parameters $f_{\epsilon}$
and $f_{\nu}$ are fully characterized by some finite-dimensional
parameters, i.e., $f_{\epsilon}(\cdot;\eta_{\epsilon})$ and $f_{\nu}(\cdot;\eta_{\nu})$
for $\eta_{\epsilon}\in\mathbb{R}^{d_{\eta_{\epsilon}}}$ and $\eta_{\nu}\in\mathbb{R}^{d_{\eta_{\nu}}}$.
This yields $(\psi^{'},\eta_{\epsilon}^{'},\eta_{\nu}^{'})^{'}$ to
be the set of parameters in the \textit{parametric version} of the
model. For either case, the parameter of the model is denoted as $\theta$
for convenience. That is, $\theta\equiv(\psi^{'},f_{\epsilon},f_{\nu})^{'}$
in the semiparametric model and $\theta\equiv(\psi^{'},\eta_{\epsilon}^{'},\eta_{\nu}^{'})^{'}$
in the parametric model. For the rest of this paper, we explicitly
express $\theta_{0}$ to be the true parameter value for $\theta$.
This applies to all the other parameter expressions.

Let $\tilde{\Psi}$ be the parameter space for $\psi$. For the parametric
model, the spaces for the finite-dimensional parameters $\eta_{\epsilon}$
and $\eta_{\nu}$ are denoted as $\mathbf{H}_{\epsilon}\subseteq\mathbb{R}^{d_{\eta_{\epsilon}}}$
and $\mathbf{H}_{\nu}\subseteq\mathbb{R}^{d_{\eta_{\nu}}}$, respectively.
Then, the parameter space $\tilde{\Theta}$ for $\theta\equiv(\psi^{'},\eta_{\epsilon}^{'},\eta_{\nu}^{'})^{'}$
becomes a Cartesian product of $\tilde{\Psi}$, $\mathbf{H}_{\epsilon}$,
and $\mathbf{H}_{\nu}$, i.e., $\tilde{\Theta}\equiv\tilde{\Psi}\times\mathbf{H}_{\epsilon}\times\mathbf{H}_{\nu}\subseteq\mathbb{R}^{d_{\psi}+d_{\eta_{\epsilon}}+d_{\eta_{\nu}}}$,
in the parametric model.\footnote{For example, if one imposes Assumption \ref{as:para_dist}, then $\eta_{\epsilon}=(\mu_{\epsilon},\sigma_{\epsilon})^{'}$
and $\eta_{\nu}=(\mu_{\nu},\sigma_{\nu})^{'}$.} For the semiparametric model, we consider the following function
spaces as the spaces for $f_{\epsilon}$ and $f_{\nu}$: 

\begin{equation}
\mathcal{F}_{j}\equiv\left\{ f=q^{2}:q\in\mathcal{F},\int\{q(x)\}^{2}dx=1\right\} ,\label{eq:density_space}
\end{equation}
where $j\in\{\epsilon,\nu\}$ and $\mathcal{F}$ is a space of functions,
which we specify later. Then, the parameter space $\tilde{\Theta}$
of $\theta\equiv(\psi^{'},f_{\epsilon},f_{\nu})^{'}$ can be written
as $\tilde{\Theta}\equiv\tilde{\Psi}\times\mathcal{F}_{\epsilon}\times\mathcal{F}_{\nu}$
in the semiparametric model. Note that the function spaces $\mathcal{F}_{\epsilon}$
and $\mathcal{F}_{\nu}$ contain functions that are nonnegative.

We adopt the ML method to estimate the parameters in the model. Let
$\{W_{i}=\{Y_{i},D_{i},X_{i}^{'},Z_{i}^{'}\}:i=1,2,...,n\}$ be the
random sample. For both parametric and semiparametric models with
corresponding $\theta$, we define the conditional density function
of $(Y_{i},D_{i})$ conditional on $(X_{i}^{'},Z_{i}^{'})^{'}$ as
\[
f(Y_{i},D_{i}|X_{i},Z_{i};\theta)={\textstyle \prod\limits _{y,d=0,1}}[p_{yd}(X_{i},Z_{i};\theta)]^{\mathbf{1}\{Y_{i}=y,D_{i}=d\}},
\]
where $p_{yd}(x,z;\theta)$ abbreviates the right hand side expression
that equates $p_{yd,xz}$ in \eqref{eq:fitted}. Then, the log of
density $l(\theta,w)\equiv\log f(y,d|x,z;\theta)$ becomes
\begin{equation}
l(\theta,W_{i})\equiv{\textstyle \sum\limits _{y,d=0,1}}\mathbf{1}_{yd}(Y_{i},D_{i})\cdot\log p_{yd}(X_{i},Z_{i};\theta),\label{eq:log_den}
\end{equation}
 where $\mathbf{1}_{yd}(Y_{i},D_{i})\equiv\mathbf{1}\{Y_{i}=y,D_{i}=d\}$.
Consequently, the log-likelihood function can be written as $Q_{n}(\theta)=\frac{1}{n}\sum\limits _{i=1}^{n}l(\theta,W_{i})$.

Now, the ML estimator $\tilde{\theta}_{n}$ of $\theta_{0}\equiv(\psi_{0}^{'},\eta_{\epsilon0},\eta_{\nu0})^{'}$
in the parametric model is defined as 
\begin{equation}
\tilde{\theta}_{n}\equiv\arg\max_{\theta\in\tilde{\Theta}}{\textstyle Q_{n}(\theta)}.\label{eq:log_likelihood}
\end{equation}
For the semiparametric model, let $\mathcal{F}_{\varepsilon n}$ and
$\mathcal{F}_{\nu n}$ be appropriate sieve spaces for $\mathcal{F}_{\varepsilon}$
and $\mathcal{F}_{\nu}$, respectively, and let $f_{\epsilon n}(\cdot;a_{\epsilon n})$
and $f_{\nu n}(\cdot;a_{\nu n})$ be the sieve approximations of $f_{\epsilon}$
and $f_{\nu}$ on their sieve spaces $\mathcal{F}_{\epsilon n}$ and
$\mathcal{F}_{\nu n}$, respectively. Then, we define the sieve ML
estimator $\hat{\theta}_{n}$ of $\theta_{0}\equiv(\psi_{0}^{'},f_{\epsilon0},f_{\nu0})^{'}$
in the semiparametric model as follows: 
\begin{equation}
\hat{\theta}_{n}\equiv\arg\max_{\theta\in\tilde{\Theta}_{n}}{\textstyle Q_{n}(\theta)},\label{eq:sieve_log_likelihood}
\end{equation}
where $\tilde{\Theta}_{n}\equiv\tilde{\Psi}\times\mathcal{F}_{\epsilon n}\times\mathcal{F}_{\nu n}$
is the sieve space for $\theta$. 

With the parameter spaces $\mathcal{F}_{\epsilon}$ and $\mathcal{F}_{\nu}$
in \eqref{eq:density_space}, we are interested in a class of ``smooth''
univariate square root density functions. Specifically, we assume
that $\sqrt{f_{\epsilon}}$ and $\sqrt{f_{\nu}}$ belong to the class
of \emph{p-smooth }functions and we restrict our attention to linear
sieve spaces for $\mathcal{F}_{\epsilon}$ and $\mathcal{F}_{\nu}$.\footnote{The definition of $p$-smooth functions can be found in \citet[p.5570]{chen2007large}
or CFT06 (p.1230). We give the formal definition of $p$-smooth functions
in Section 4.} In this case, the choice of sieve spaces for $\mathcal{F}_{\epsilon}$
and $\mathcal{F}_{\nu}$ depends on the supports of $\epsilon$ and
$\nu$. If the supports are bounded, then one can use the polynomial
sieve, trigonometric sieve, or cosine sieve. When the supports are
unbounded, then we can use the Hermite polynomial sieve or the spline
wavelet sieve. 

In this paper, we implicitly assume that the copula function is correctly
specified. As mentioned earlier, using a parametric copula may lead
to model misspecification. It is well known that when the model is
misspecified, the ML estimator converges to a pseudo-true value which
minimizes the Kullback-Leibler (KL) divergence (e.g., \citet{white1982maximum}).
This result applies to a semiparametric model (e.g., \citet{chen2006mis_cop}
and \citet{chen2006estimation}) as in our semiparametric case. We,
however, do not investigate the asymptotic properties of the sieve
estimators under copula misspecification, as it is beyond the scope
of this paper. Instead, later in simulation, we investigate how the
copula misspecification affects the performance of estimators.\footnote{For related issues of copula misspecification, refer to, e.g., \citet{chen2006mis_cop}
and \citet{liao2017uniform}. In particular, \citet{chen2006mis_cop}
propose a test procedure for model selection that is based on the
test of \citet{vuong1989likelihood}. \citet{liao2017uniform} extend
Vuong's test to cases where models contain infinite-dimensional parameters
and propose a uniformly asymptotically valid Vuong test for semi/non-parametric
models. Their setting encompasses those models that can be estimated
by the sieve ML as a special case. }

\section{\label{sec:Asymptotic}Asymptotic Theory for Sieve ML Estimators}

In this section, we provide the asymptotic theory for the sieve ML
estimator $\hat{\theta}_{n}$ of $\theta\equiv(\psi^{'},f_{\epsilon},f_{\nu})^{'}$
in the semiparametric model. This theory will be useful for practitioners
to conduct inference. The asymptotic theory for the ML estimator $\tilde{\theta}_{n}$
of $\theta\equiv(\psi^{'},\eta_{\epsilon},\eta_{\nu})^{'}$ in the
parametric model is relatively standard and can be found in, e.g.,
\citet{newey1994large}. The theory establishes that the parametric
ML estimator is consistent, asymptotically normal, and efficient under
some regularity conditions. To investigate the asymptotic properties
of the sieve ML estimator, we slightly modify our model as follows. 

Let $G(\cdot$) be a strictly increasing function mapping from $\mathbb{R}$
to $[0,1]$. We further assume that $G$ is differentiable and that
its derivative $g(x)\equiv\frac{dG(x)}{dx}$ is bounded away from
zero on $\mathbb{R}$. Then, without loss of generality (e.g., \citet{bierens2014consistency}),
we consider the following transformation of $F_{\epsilon0}$ and $F_{\nu0}$
as: 
\begin{align}
F_{\epsilon0}(x) & =H_{\epsilon0}(G_{\epsilon}(x)),\quad F_{\nu0}(x)=H_{\nu0}(G_{\nu}(x)),\label{eq:transform}
\end{align}
where $H_{\epsilon0}(\cdot)$ and $H_{\nu0}(\cdot)$ are unknown distribution
functions on $[0,1]$. For $G$, we can choose the standard normal
distribution function or the logistic distribution function. Since
we assume that the distribution functions of $\epsilon$ and $\nu$
admit density functions, we require that $H_{\epsilon0}$ and $H_{\nu0}(\cdot)$
be differentiable, and write their derivatives as $h_{\epsilon0}(\cdot)$
and $h_{\nu0}(\cdot)$, respectively. For each $j\in\{\epsilon,\nu\}$,
let $\mathcal{H}_{j}\equiv\{h_{j}=q^{2}:q\in\mathcal{F}\}$ for some
function space $\mathcal{F}$. With this modification, we redefine
the parameter as $\theta=(\psi^{'},h_{\epsilon},h_{\nu})^{'}\in\tilde{\Theta}^{\dagger}\equiv\tilde{\Psi}\times\mathcal{H}_{\epsilon}\times\mathcal{H}_{\nu}$.
Note that, using the transformation of the distribution functions
in equation \eqref{eq:transform}, the unknown infinite-dimensional
parameters are defined on a bounded domain. In the online appendix,
we show that the transformation does not affect the identification
result.

We redefine the parameter space to facilitate developing the asymptotic
theory. The identification requires that the space of the finite-dimensional
parameter $\tilde{\Psi}$ be open and convex (see Theorems \ref{thm_full}
and \ref{thm_full_unknown}), and thus $\tilde{\Psi}$ cannot be compact.
We introduce an ``optimization space'' that contains the true parameter
$\psi_{0}$ and consider it as the parameter space of $\psi$. Formally,
we restrict the parameter space for estimation in the following way. 

\begin{assumption}\label{as:compact_para} There exists a compact
and convex subset $\Psi\subseteq\tilde{\Psi}$ such that $\psi_{0}\in int(\Psi)$,
where $int(A)$ is the interior of the set $A$. \end{assumption}

With the optimization space, we define the parameter space as $\Theta\equiv\Psi\times\mathcal{H}_{\epsilon}\times\mathcal{H}_{\nu}$,
and the corresponding sieve space is denoted by $\Theta_{n}\equiv\Psi\times\mathcal{H}_{\epsilon n}\times\mathcal{H}_{\nu n}$.
Then, the sieve ML estimator in equation \eqref{eq:sieve_log_likelihood}
is also redefined as follows: 
\begin{equation}
\hat{\theta}_{n}\equiv\text{arg max}_{\theta\in\Theta_{n}}Q_{n}(\theta).\label{eq:sieve_ML_est}
\end{equation}

\subsection{Consistency of the Sieve ML Estimators}

We begin by showing the consistency of the sieve ML estimator. Since
the parameter involves both finite- and infinite-dimensional objects,
we establish the consistency of the sieve ML estimators with respect
to a pseudo distance function $d_{c}$ on $\Theta\times\Theta$.\footnote{It is important to choose appropriate norms to ensure the compactness
of the original parameter space, as compactness plays a key role in
establishing the asymptotic theory. Since the parameter space is infinite-dimensional,
it may be compact under certain norms but not under other norms. An
infinite-dimensional space that is closed and bounded is not necessarily
compact, and thus it is more demanding to show that the parameter
space is compact under certain norms. To overcome this difficulty,
we take the approach introduced by \citet{gallant1987semi}, which
uses two norms to obtain the consistency. Their idea is to use the
strong norm to define the parameter space as a ball, and then to ensure
the compactness of the parameter space using the consistency norm.
In our setting, the H\"older norm is the strong norm and $||\cdot||_{c}$
is the consistency norm. Related to this issue, \citet{freyberger2015compactness}
recently extend the idea to more cases and present compactness results
for several parameter spaces.} All of the norms and the definitions of function spaces in this paper
are provided in the online appendix. 

We present the following assumptions, under which the sieve ML estimator
in equation \eqref{eq:sieve_ML_est} is consistent with respect to
the pseudo-metric $d_{c}(\cdot,\cdot)$. 

\begin{assumption}\label{as:WD} There exists a measurable function
$\underline{p}(X,Z)$ such that for all $\theta\in\Theta$ and for
all $y,d=0,1$, $p_{yd,XZ}(\theta)\geq\underline{p}(X,Z)$, with $E|\log(\underline{p}(X,Z))|<\infty$
and $E\left[\frac{1}{\underline{p}(X,Z)^{2}}\right]<\infty$. \end{assumption}

\begin{assumption}\label{as:data}$\{W_{i}:i=1,2,...,n\}$ is a random
sample, with $E\left[||(X_{i}^{'},Z_{i}^{'})^{'}||_{E}^{2}\right]<\infty$.
\end{assumption}

\begin{assumption}\label{as:para_sp} (i) $\sqrt{h_{\epsilon0}},\sqrt{h_{\nu0}}\in\Lambda_{R}^{p}([0,1])$,
with $p>\frac{1}{2}$ and some $R>0$; (ii) $\mathcal{H}_{\epsilon}=\mathcal{H}_{\nu}=\mathcal{H}$
where $\mathcal{H}\equiv\left\{ h=q^{2}:q\in\Lambda_{R}^{p}([0,1]),\int_{0}^{1}q=1\right\} $,
with $R$ being defined as in (i) and $\Lambda_{R}^{p}([0,1])$ being
a H\"older ball with radius $R$; (iii) the density functions $h_{\epsilon0}$
and $h_{\nu0}$ are bounded away from zero on $[0,1]$. \end{assumption} 

\begin{assumption}\label{as:sieve}(i) $\mathcal{H}_{\epsilon n}=\mathcal{H}_{\nu n}\equiv\{h\in\mathcal{H}:h(x)=p^{k_{n}}(x)^{'}a_{k_{n}},a_{k_{n}}\in\mathbb{R}^{k_{n}},||h||_{\infty}<2R^{2}\}$,
where $k_{n}\rightarrow\infty$ and $k_{n}/n\rightarrow0$ as $n\rightarrow\infty$;
(ii) for all $j\geq1$, we have $\Theta_{j}\subseteq\Theta_{j+1}$,
and there exists a sequence $\{\pi_{j}\theta_{0}\}_{j}$ such that
$d_{c}(\pi_{j}\theta_{0},\theta_{0})\rightarrow0$ as $j\rightarrow\infty$.
\end{assumption} 

\begin{assumption}\label{as:copula_derivative} For $j=1,2$, let
$C_{j}(u_{1},u_{2};\rho)\equiv\frac{\partial C(u_{1},u_{2};\rho)}{\partial u_{j}}$
and $C_{\rho}(u_{1},u_{2};\rho)\equiv\frac{\partial C(u_{1},u_{2};\rho)}{\partial\rho}$.
The derivatives $C_{j}(\cdot,\cdot;\cdot)$ and $C_{\rho}(\cdot,\cdot;\cdot)$
are uniformly bounded for all $j=1,2$. \end{assumption}

Assumption \ref{as:WD} guarantees that the log-likelihood function
$l(\theta,W_{i})$ is well defined for all $\theta\in\Theta$ and
that $Q_{0}(\theta_{0})>-\infty$. Assumption \ref{as:data} restricts
the data generating process (DGP), and assumes the existence of moments
of the data. Assumption \ref{as:para_sp} defines the parameter space
and implies that the infinite-dimensional parameters are in some smooth
class called a H\"older class. Note that conditions (i) and (ii)
in Assumption \ref{as:para_sp} together imply that $h_{\epsilon0}$
and $h_{\nu0}$ belong to $\Lambda_{\tilde{R}}^{p}([0,1])$, where
$\tilde{R}\equiv2^{m+1}R^{2}<\infty$.\footnote{See the online appendix for details.}
Thus, we may assume that $h_{\epsilon0}$ and $h_{\nu0}$ belong to
a H\"older ball with smoothness $p$ under Assumption \ref{as:para_sp}.\footnote{These conditions implicitly define the strong norm (H\"older norm).}
The condition that $\mathcal{H}_{\epsilon}$ and $\mathcal{H}_{\nu}$
are the same can be relaxed, but it is imposed for simplicity. The
first part of Assumption \ref{as:sieve} restricts our choice of sieve
spaces for $\mathcal{H}_{\epsilon}$ and $\mathcal{H}_{\nu}$ to linear
sieve spaces with order $k_{n}$. This can be relaxed so that the
choice of $k_{n}$ is different for $h_{\epsilon}$ and $h_{\nu}$.
The latter part of Assumption \ref{as:sieve} requires that the sieve
space be chosen appropriately so that the unknown parameters can be
well-approximated. Because the unknown infinite-dimensional parameters
belong to a H\"older ball and are defined on bounded supports, we
can choose the polynomial sieve, trigonometric sieve, cosine sieve,
or spline sieve.\footnote{Refer to \citet{chen2007large} or CFT06 for details on the choice
of sieve spaces.} For example, if we choose the polynomial sieve or the spline sieve,
then one can show that $d_{c}(\pi_{k_{n}}\theta_{0},\theta_{0})=O(k_{n}^{-p})$
(e.g., \citet{lorentz1966approximation}). Assumption \ref{as:copula_derivative}
imposes the boundedness of the derivatives of the copula function. 

The following theorem demonstrates that under the above assumptions,
the sieve estimator $\hat{\theta}_{n}$ is consistent with respect
to the pseudo metric, $d_{c}$. 

\begin{theorem}\label{thm:consistency} Suppose that Assumptions
\ref{as:independence}\textendash \ref{as_copula} and \ref{as:add_support}
hold. If Assumptions \ref{as:compact_para}\textendash \ref{as:copula_derivative}
are satisfied, then $d_{c}(\hat{\theta}_{n},\theta_{0})\overset{p}{\rightarrow}0$.
\end{theorem}

\subsection{Convergence Rates}

In this section, we derive the convergence rate of the sieve ML estimator.
The convergence rate provides information on how fast the estimator
converges to the true parameter value. Heuristically, the faster the
convergence rate, the larger the effective sample size is for estimation.
The next theorem demonstrates the convergence rate of the sieve ML
estimator with respect to the $L^{2}$-norm $||\cdot||_{2}$.

\begin{theorem}\label{thm:convergence_rate} Suppose that Assumptions
\ref{as:independence}\textendash \ref{as_copula} and \ref{as:add_support}\textendash \ref{as:copula_derivative}
hold. If Assumption \ref{as:KL-L2} in the online appendix additionally
holds, then we have $||\hat{\theta}_{n}-\theta_{0}||_{2}=O_{p}\left(\max\left\{ \sqrt{k_{n}/n},k_{n}^{-p}\right\} \right)$.
Furthermore, if we choose $k_{n}\propto n^{\frac{1}{2p+1}}$, then
we have $||\hat{\theta}_{n}-\theta_{0}||_{2}=O_{p}\left(n^{-\frac{p}{2p+1}}\right)$.\end{theorem}

The former convergence rate is standard in the literature, where the
first term corresponds to variance, which increases in $k_{n}$, and
the second term corresponds to the approximation error $||\theta_{0}-\pi_{k}\theta_{0}||_{2}$,
which decreases in $k_{n}$. The choice of $k_{n}\propto n^{\frac{1}{2p+1}}$
yields the optimal convergence rate, which is slower than the parametric
rate ($n^{-1/2}$). Note that this rate increases with the degree
of smoothness, $p$.

\subsection{Asymptotic Normality of Smooth Functionals}

We now establish the asymptotic normality of smooth functionals. The
parameters in our model contains both finite- and infinite-dimensional
parameters, and many objects of interest are written as functionals
of both types of the parameters. The results of this section can be
used to calculate the standard error of the estimate of a functional
of interest (including the individual finite-dimensional parameters),
or to conduct inference (i.e., testing hypotheses and constructing
confidence intervals) based on normal approximation.

Before proceeding, we strengthen the smoothness condition in Assumption
\ref{as_copula}. Let $C_{ij}(u_{1},u_{2};\rho)$ denote the second-order
partial derivative of a copula function $C(u_{1},u_{2};\rho)$ with
respect to $i$ and $j$, for $i,j\in\{u_{1},u_{2},\rho\}$.

\begin{assumption}\label{as:copula_second} The copula function $C(u_{1},u_{2};\rho)$
is twice continuously differentiable with respect to $u_{1},u_{2},$
and $\rho$, and its first- and second- order partial derivatives
are well defined in a neighborhood of $\theta_{0}$. \end{assumption}

Let $\mathbb{V}$ be the linear span of $\Theta-\{\theta_{0}\}$.
For $t\in[0,1]$, define the directional derivative of $l(\theta,W)$
at the direction $v\in\mathbb{V}$ as 
\begin{align}
\left.\frac{dl(\theta_{0}+tv,W)}{dt}\right|_{t=0} & \equiv\lim_{t\rightarrow0}\frac{l(\theta_{0}+tv,W)-l(\theta_{0})}{t}=\frac{\partial l(\theta_{0},W)}{\partial\psi^{'}}v_{\psi}+\sum_{j\in\{\epsilon,\nu\}}\frac{\partial l(\theta_{0},W)}{\partial h_{j}}[v_{j}],\label{eq:derivative_ll}
\end{align}
where $\frac{\partial l(\theta_{0},W)}{\partial\psi^{'}}v_{\psi}$,
$\frac{\partial l(\theta_{0},W)}{\partial h_{\epsilon}}[v_{\epsilon}]$,
and $\frac{\partial l(\theta_{0},W)}{\partial h_{\nu}}[v_{\nu}]$
are given by equations \eqref{eq:direc_derivative_1}\textendash \eqref{eq:direc_derivative_3}
in the online appendix. If we denote the closed linear span of $\mathbb{V}$
under the Fisher norm $||\cdot||$ by $\bar{\mathbb{V}}$, then $(\bar{\mathbb{V}},||\cdot||)$
is a Hilbert space. 

Let $T:\Theta\rightarrow\mathbb{R}$ be a functional. For any $v\in\mathbb{V}$,
we write 
\[
\frac{\partial T(\theta_{0})}{\partial\theta^{'}}[v]\equiv\lim_{t\rightarrow0}\frac{T(\theta_{0}+tv)-T(\theta_{0})}{t},
\]
 provided the right hand side limit is well defined. The following
assumption characterizes the smoothness of the functional $T$. 

\begin{assumption}\label{as:smooth_T} The following conditions hold: 

(i) there exist constants $w>1+\frac{1}{2p}$ and a small $\epsilon_{0}>0$
such that for any $v\in\mathbb{V}$ with $||v||\leq\epsilon_{0}$,
\[
\left|T(\theta_{0}+v)-T(\theta_{0})-\frac{\partial T(\theta_{0})}{\partial\theta^{'}}[v]\right|=O(||v||^{w});
\]

(ii) For any $v\in\mathbb{V}$, $T(\theta_{0}+tv)$ is continuously
differentiable in $t\in[0,1]$ around $t=0$, and 
\[
\left\Vert \frac{\partial T(\theta_{0})}{\partial\theta^{'}}\right\Vert \equiv\sup_{v\in\mathbb{V},||v||>0}\frac{\left|\frac{\partial T(\theta_{0})}{\partial\theta^{'}}[v]\right|}{||v||}<\infty.
\]

\end{assumption} 

Assumption \ref{as:smooth_T} defines a smooth functional $T$ and
guarantees the existence of $v^{*}\in\bar{\mathbb{V}}$ such that
$<v^{*},v>=\frac{\partial T(\theta_{0})}{\partial\theta^{'}}[v]$
for all $v\in\mathbb{V}$ and $||v^{*}||^{2}=\left\Vert \frac{\partial T(\theta_{0})}{\partial\theta^{'}}\right\Vert ^{2}$.
Here, we call $v^{*}$ the Riesz representer for the functional $T$. 

The next assumption requires that the Riesz representer be well approximated
over the sieve space and that it converges at a rate with respect
to the Fisher norm. 

\begin{assumption}\label{as:smooth_Riesz} There exists $\pi_{n}v^{*}\in\Theta_{n}-\{\theta_{0}\}$
such that $||\pi_{n}v^{*}-v^{*}||=o(n^{-1/4})$. \end{assumption}

The following proposition states that the plug-in sieve ML estimator
$T(\hat{\theta}_{n})$ of $T(\theta_{0})$ is $\sqrt{n}$-asymptotically
normally distributed under certain conditions. The technical conditions
(Assumptions \ref{as:KL-L2}, \ref{as:secondorder} and \ref{as:emp_process})
can be found in the online appendix. 

\begin{proposition}\label{prop:asym general} Suppose that Assumptions
\ref{as:independence}\textendash \ref{as_copula}, \ref{as:add_support}\textendash \ref{as:smooth_Riesz},
\ref{as:KL-L2}\textendash \ref{as:emp_process} are satisfied. If
$k_{n}\propto n^{\frac{1}{2p+1}}$, then we have 
\[
\sqrt{n}(T(\hat{\theta}_{n})-T(\theta_{0}))\overset{d}{\rightarrow}\mathcal{N}\left(0,\left\Vert \frac{\partial T(\theta_{0})}{\partial\theta^{'}}\right\Vert ^{2}\right).
\]
 \end{proposition}It is worth noting that, although the parameter
$T(\theta_{0})$ contains an infinite-dimensional object (i.e., the
marginal distributions of $\epsilon$ and $\nu$), the sieve plug-in
estimator is $\sqrt{n}$-estimable due to the fact that $T$ is a
smooth functional. 

\subsubsection{Example 1: Asymptotic Normality for the Finite-Dimensional Parameter
$\psi_{0}$}

The finite-dimensional parameter $\psi_{0}$ is a special case of
the smooth functionals. Here, we demonstrate the asymptotic normality
of the sieve estimator of the finite-dimensional parameter $\psi_{0}$.

\begin{theorem}\label{thm:AN_psi} Suppose that Assumptions \ref{as:independence}\textendash \ref{as_copula},
\ref{as:add_support}\textendash \ref{as:copula_second}, \ref{as:smooth_Riesz},
\ref{as:KL-L2}\textendash \ref{as:nonsingular} hold. Then, we have
\begin{equation}
\sqrt{n}(\hat{\psi}_{n}-\psi_{0})\overset{d}{\rightarrow}\mathcal{N}\left(0,\mathcal{I}_{*}(\psi_{0})^{-1}\right),\label{eq:normal_finite}
\end{equation}
 and the form of $\mathcal{I}_{*}(\psi)$ is given in the online appendix.
\end{theorem}

The covariance matrix in \eqref{eq:normal_finite} needs to be estimated.
To do so, CFT06 adopt the covariance estimation method proposed by
\citet{ai2003efficient}. Since an infinite-dimensional optimization
is involved in calculating $\mathcal{S}_{\psi_{0}}$, we provide a
sieve estimator of $\mathcal{I}_{*}(\psi_{0})^{-1}$. The sieve spaces
for $b_{\epsilon}$ and $b_{\nu}$ can be the same as those for $h_{\epsilon}$
and $h_{\nu}$, respectively. As in \citet{ai2003efficient}, we first
estimate efficient score functions  by solving the following minimization
problem: for all $k=1,2,...,d_{\psi}$, 
\[
(\hat{b}_{\epsilon k},\hat{b}_{\nu k})\equiv\arg\min_{(b_{\epsilon k},b_{\nu k})\in\mathcal{H}_{\epsilon n}\times\mathcal{H}_{\nu n}}\frac{1}{n}\sum_{i=1}^{n}\left\{ \frac{\partial l(\hat{\theta}_{n},W_{i})}{\partial\psi_{k}}-\left(\frac{\partial l(\hat{\theta}_{n},W_{i})}{\partial h_{\epsilon}}[b_{\epsilon k}]+\frac{\partial l(\hat{\theta}_{n},W_{i})}{\partial h_{\nu}}[b_{\nu k}]\right)\right\} ^{2}.
\]
 Let $\hat{b}_{j}=(\hat{b}_{j1},\hat{b}_{j2},...,\hat{b}_{jd_{\psi}})^{'}$
for given $j\in\{\epsilon,\nu\}$ and compute 
\begin{align*}
\hat{\mathcal{I}}_{*}(\hat{\psi}_{n})= & \frac{1}{n}\sum_{i=1}^{n}\left\{ \left[\frac{\partial l(\hat{\theta}_{n},W_{i})}{\partial\psi}-\left(\frac{\partial l(\hat{\theta}_{n},W_{i})}{\partial h_{\epsilon}}[\hat{b}_{\epsilon}]+\frac{\partial l(\hat{\theta}_{n},W_{i})}{\partial h_{\nu}}[\hat{b}_{\nu}]\right)\right]\right.\\
 & \times\left.\left[\frac{\partial l(\hat{\theta}_{n},W_{i})}{\partial\psi}-\left(\frac{\partial l(\hat{\theta}_{n},W_{i})}{\partial h_{\epsilon}}[\hat{b}_{\epsilon}]+\frac{\partial l(\hat{\theta}_{n},W_{i})}{\partial h_{\nu}}[\hat{b}_{\nu}]\right)\right]^{'}\right\} 
\end{align*}
 to obtain a consistent estimator of $\mathcal{I}_{*}(\psi_{0})$.
We now summarize this result as follows:

\begin{theorem}\label{thm:var_est_para} Suppose that assumptions
in Theorem \ref{thm:AN_psi} hold. Then, $\hat{\mathcal{I}}_{*}(\hat{\psi}_{n})=\mathcal{I}_{*}(\psi_{0})+o_{p}(1)$.
\end{theorem}

The proof of the theorem can be found in Theorem 5.1 in \citet{ai2003efficient}. 

\subsubsection{Example 2: Asymptotic Normality for the Conditional ATE}

We now consider the conditional ATE, $E[Y_{1}-Y_{0}|X=x]=F_{\epsilon0}(x^{'}\beta_{0}+\delta_{10})-F_{\epsilon0}(x^{'}\beta_{0})$.
From Proposition \ref{prop:asym general}, we provide the asymptotic
normality of the sieve plug-in estimator of the conditional ATE: 

\begin{theorem} Let $x\in\text{supp}(X)$ be given. Suppose that
the conditions in Proposition \ref{prop:asym general} hold with $T(\theta_{0})=ATE(\theta_{0};x)$.
Then, we have 
\begin{equation}
\sqrt{n}(ATE(\hat{\theta}_{n};x)-ATE(\theta_{0};x))\overset{d}{\rightarrow}\mathcal{N}\left(0,\left\Vert \frac{\partial ATE(\theta_{0};x)}{\partial\theta^{'}}[v]\right\Vert ^{2}\right),\label{eq:ate}
\end{equation}
 where $\left\Vert \frac{\partial ATE(\theta_{0};x)}{\partial\theta^{'}}[v]\right\Vert ^{2}=\sup_{v\in\mathbb{V},||v||>0}\frac{\left|\frac{\partial ATE(\theta_{0};x)}{\partial\theta^{'}}[v]\right|}{||v||}$,
and the form of $\frac{\partial ATE(\theta_{0};x)}{\partial\theta^{'}}[v]$
is given by \eqref{eq:diff_ATE} in the online appendix. 

Furthermore, the asymptotic variance in \eqref{eq:ate} can be estimated
as follows: 
\[
\hat{\sigma}_{ATE(\theta;x)}^{2}\equiv\max_{v\in\Theta_{n}}\left\Vert \frac{\partial ATE(\hat{\theta}_{n};x)}{\partial\theta^{'}}[v]\right\Vert ^{2}.
\]
 \end{theorem}

\subsection{Weighted Bootstrap\label{subsec:bootstrap}}

The asymptotic variances characterized in the previous subsection
can be estimated using the sieve methods. In practice, estimating
asymptotic variances may be sensitive to the choice of the number
of sieve approximation terms. Furthermore, when the dimension of $\theta_{0}$
is large, it is relatively cumbersome to estimate the asymptotic variance
of the sieve estimator for the finite-dimensional parameter. In this
subsection, we briefly discuss the weighted bootstrap as an alternative
procedure.

For general semiparametric M-estimation, \citet{ma2005robust} and
\citet{cheng2010bootstrap} provide the validity of the weighted bootstrap
for finite-dimensional parameters in a class of semiparametric models
that includes our model. Related to these results, \citet{chen2009efficient}
provide the bootstrap validity in semiparametric conditional moment
models. We do not pursue to prove the bootstrap validity in this paper,
as these references sufficiently address it. In our empirical exercise,
we use the weighted bootstrap scheme proposed in these papers to obtain
the standard errors of the estimated functionals of interest. Let
$T(\theta_{0})$ be a smooth functional of interest and $B$ be the
number of bootstrap iterations. The weighted bootstrap is carried
out as follows: 
\begin{enumerate}
\item For each $b=1,2,...,B$, let $\{B_{i}^{(b)}:i=1,2,...,n\}$ be a random
sample generated from a positive random variable $B_{i}$ such that
$EB_{i}=1$, $Var(B_{i})=1$, and is independent of $\{W_{i}:i=1,2,...n\}$.\footnote{Note that the condition on the variance of $B_{i}$ can be relaxed.
In our empirical example, we use $B_{i}\sim\exp(1)$. } 
\item For each bootstrap iteration $b=1,2,...,B$, define $\hat{\theta}_{n}^{*(b)}$
be a bootstrap estimate of $\theta_{0}$: 
\[
\hat{\theta}_{n}^{*(b)}\equiv\arg\max_{\theta\in\tilde{\Theta}_{n}}{\textstyle Q_{n}^{*(b)}(\theta)},
\]
 where $Q_{n}^{*(b)}(\theta)\equiv\frac{1}{n}\sum\limits _{i=1}^{n}B_{i}^{(b)}\cdot l(\theta,W_{i})$.
Obtain the bootstrap estimate of the functional of interest by using
$\hat{\theta}_{n}^{*(b)}$ and denote it by $T(\hat{\theta}_{n}^{*(b)})$. 
\item The bootstrap standard error of $T(\hat{\theta}_{n})$ is given by
$\sqrt{\frac{1}{B}\sum_{b=1}^{B}\left(T(\hat{\theta}_{n}^{*(b)})-\bar{T}_{B}^{*}\right)}$,
where $\bar{T}_{B}^{*}\equiv\frac{1}{B}\sum_{b=1}^{B}T(\hat{\theta}_{n}^{*(b)})$.
 
\end{enumerate}
One may use the bootstrap standard errors to construct confidence
intervals, and such confidence intervals rely on the normal approximation.
As an alternative to the normal approximation, one can use percentile
confidence intervals. For a small $p\in(0,1)$, a $(1-p)\times100$
percent percentile confidence interval for a functional $T(\theta_{0})$
is constructed as follows: 
\[
PCI(p)\equiv\left[Q_{T}^{*}(p/2),\quad Q_{T}^{*}(1-p/2)\right],
\]
 where $Q_{T}^{*}(\tau)$ is the $\tau$-th quantile of bootstrap
estimates $\{T(\hat{\theta}_{n}^{*(b)}):b=1,2,...,B\}$. We suggest
that practitioners use the percentile confidence intervals rather
than the confidence intervals with the bootstrap standard errors.

\section{\label{sec:Monte_Carlo}Monte Carlo Simulation and Sensitivity Analysis}

In this section, we conduct a sensitivity analysis via Monte Carlo
simulation exercises to provide guidance for empirical researchers.
To this end, we investigate the finite sample performance of the sieve
ML estimators of the finite-dimensional parameter $\psi_{0}$ and
the ATE. We compare them with the performance of the parametric ML
estimators under various DGPs and model specifications, and illustrate
how the parametric estimators of $\psi_{0}$ and the ATE suffer from
misspecification of the marginal distribution of $\epsilon$. Note
that the ATE involves $\psi_{0}$ and the marginal of $\epsilon$.

\subsection{Simulation Design}

We compare the performance of the parametric and semiparametric estimators
when the marginal distributions are misspecified in the parametric
models. To calculate the parametric estimators, we specify the parametric
models with normal distributions for the marginals of $\epsilon$
and $\nu$, owing to their popularity. For the DGPs, we consider two
marginals of $\epsilon$ and $\nu$: the standard normal distribution
(to reflect correct specification) and a mixture of normal distributions
(to reflect misspecification).

The DGPs are as follows: 
\begin{align*}
Y_{i} & =\mathbf{1}\{X_{i}\beta+D_{i}\delta_{1}\geq\epsilon\},\quad D_{i}=\mathbf{1}\{X_{i}\alpha+Z_{i}\gamma\geq\nu\},
\end{align*}
 where $(\alpha,\gamma,\beta,\delta_{1})=(-1,0.8,-1,1.1)$, $(X,Z)'\sim\mathcal{N}\left((0,0)',\begin{pmatrix}1 & -0.1\\
-0.1 & 1
\end{pmatrix}\right)$, and $(\epsilon,\nu)^{'}\sim C(F_{\epsilon0}(\cdot),F_{\nu0}(\cdot);\rho)$.
Here, $F_{\epsilon0}$ and $F_{\nu0}$ are normal or a mixture of
normal.\footnote{For the mixture of normal distributions, $\epsilon$ and $\nu$ are
generated from $0.6\mathcal{N}(-1,\sigma^{2})+0.4\mathcal{N}(1.5,\sigma^{2})$
for appropriate $\sigma>0$, so that the mean is zero and the variance
is one.} For $C(\cdot,\cdot;\rho)$, we consider the Gaussian, Frank, Clayton,
and Gumbel copulas, which satisfy the identifying assumption (Assumption
\ref{as_copula}). The dependence structure between $\epsilon$ and
$\nu$ is characterized by a one-dimensional parameter $\rho$ in
all copulas considered, but the interpretation of the dependence parameter
differs across the copulas. To resolve this issue, we report the Spearman's
$\rho$ corresponding to the estimated dependence parameter in each
copula specification. We estimate the models with several values of
$\rho$ to examine whether the performance of the estimators varies
with the degree of dependence. Although we assume that the copula
is correctly specified, economic theory does not provide a justification
for the choice of copula. In this simulation study, we also examine
the effect of copula misspecification on the performance of the estimators.\footnote{Misspecification problems in copula-based models have been documented
using Monte Carlo simulations in the statistic literature (e.g., \citet{kim2007comparison,kim2007semiparametric,lawless2011comparison}).
In particular, \citet{lawless2011comparison} compare the performance
of the parametric and semiparametric ML estimators in a copula-based
model and show that the semiparametric two-step method outperforms
the parametric estimation method when the copula function is misspecified.}

We impose a restriction that $X$ has no constant for the location
normalization, and fix $\alpha$ and $\beta$ to -1 for the scale
normalization. We use these normalizations in both parametric and
semiparametric models, and it allows us to easily compare the performance
of the parametric and semiparametric estimators. We consider two sample
sizes, 500 and 1000, and all results are obtained from 2000 Monte
Carlo replications. As a performance measure of the estimators, we
consider the root mean squared errors (RMSEs) in our simulation. 

\subsection{Estimation of Parametric and Semiparametric Models }

The parametric models can be estimated by the standard ML method.
Since bivariate probit models are commonly used in practice, we specify
the model using the Gaussian copula and normal marginals. In addition
to that, we also try different copulas and normal marginals.\footnote{Such an estimation method in related parameteric models can be found
in \citet{marra2011estimation}. The R package (GJRM) used in their
paper can be used to estimate our parametric model as well.}

Consider semiparametric models. Recall that we assume that $\sqrt{h_{j}}\in\Lambda^{p}([0,1])$.
Therefore, for each $j\in\{\epsilon,\nu\}$, we approximate $h_{j}$
to 
\begin{equation}
h_{j}(x)=\frac{\left(\sum_{k=0}^{k_{nj}}a_{jk}\psi_{jk}(x)\right)^{2}}{\int_{0}^{1}\left(\sum_{k=0}^{k_{nj}}a_{jk}\psi_{jk}(x)\right)^{2}dx},\label{eq:sim_sieve_app}
\end{equation}
where $\{\psi_{jk}(\cdot)\}_{k=0}^{k_{nj}}$ is the set of approximating
functions for $h_{j}(\cdot)$, and $k_{nj}$ is the number of approximating
functions. The approximation in \eqref{eq:sim_sieve_app} guarantees
that $\int_{0}^{1}h_{j}(x)dx=1$ by construction. We take the space
of the polynomials as the sieve space for $h_{\epsilon}$ and $h_{\nu}$.
The orders of the polynomials ($k_{n\epsilon}$ and $k_{n\nu}$) are
set to be proportional to $n^{1/7}$.  To incorporate the specification
given in \eqref{eq:transform}, we choose the standard normal distribution
function for $G$. 

\subsection{Simulation Results}

We begin by examining the simulation results under correct specification
(i.e., the true marginal distributions and the specified marginal
distributions are both normal). Table \ref{tab:correct_500} shows
the simulation results for $n=500$. We find that the ML estimators
of $\psi$ and the ATE perform well in the parametric models, with
negligible biases and small variances.\footnote{The ATE is evaluated at the mean of $X$.}
The performance of the sieve ML estimators of $\psi$ and the ATE
in the semiparametric models is as good as that in the parametric
models, even with this moderate sample size. 

Now, we consider the cases where the marginal distributions are misspecified
in the parametric models. Table \ref{tab:mar_500} considers the case
where the true marginal distributions are a mixture of normal distributions,
but the researcher specifies them as normal distributions. In this
table, the RMSEs of the parametric ML estimators are larger than those
of the sieve ML estimators. This implies that the parametric ML estimators
suffer from misspecification while the sieve ML estimators do not.
Moreover, the parametric estimators of the ATE are substantially distorted
under this misspecification, presumably because the ATE is a function
of the misspecified distribution of $\epsilon$. Note that the poor
performance of the parametric estimators is attributed not only to
large bias, but also large variance. For instance, the bias of the
parametric estimator of the ATE with the Gaussian copula is 0.1377,
which is about eight times larger than that of the corresponding sieve
estimator. These biases of the parametric estimators of the ATE are
substantial in that they do not disappear with the increased sample
size.\footnote{We provide simulation results with a larger sample size ($n=1000$),
and they can be found in the online appendix.} Therefore, the simulation results demonstrate that when the marginal
distributions are misspecified, the sieve estimators outperform the
parametric estimators in terms of the RMSE. The online appendix also
contains simulation results for the cases where both the copula and
the marginal distributions are misspecified. The results show that,
even under copula misspecification, the sieve ML estimators remain
to outperform the parametric counterparts when the marginal distributions
are misspecified.

Overall, the simulation results suggest that researchers are recommended
to use the semiparametric models and the sieve ML estimation proposed
in this paper when they are concerned about model misspecification.
The following is the summary of the main findings from our simulation
study: 
\begin{itemize}
\item [(i)]When the model is correctly specified, the performance of the
sieve ML estimators is comparable to that of the parametric ML estimators. 
\item [(ii)]When the marginal distributions are misspecified, the sieve
ML estimation is recommended in order to improve the performance. 
\item [(iii)]The semiparametric ML estimators performs better than the
parametric ML estimators under both copula and marginal misspecification.
Therefore, the semiparametric models are preferred to the parametric
models in such cases. 
\item [(iv)]Especially for the ATE, whenever the marginal distributions
are misspecified, the parametric ML estimates can be significantly
distorted. 
\end{itemize}
We provide additional simulation results in the online appendix, where
we consider (a) a larger sample size, (b) both copula and marginal
misspecification, (c) different degrees of dependence, (d) marginal
density functions of heavy tails, and (e) the coverage probabilities
of bootstrap confidence intervals. Here is the summary. Across various
simulation designs ((a)\textendash (c)), our main findings remain
the same. When the marginal distributions are believed to have fat
tails, we recommend practitioners to use the transformation function
$G$ that has fat tails. Lastly, the percentile bootstrap works well
with the coverage probabilities close to its nominal level.

\section{\label{sec:Empirical-Example}Empirical Example}

In this section, we illustrate in an application the practical relevance
of the theoretical results developed in this paper. It is widely recognized
that health insurance coverage can be an important factor for patients'
decisions for making medical visits. At the same time, having insurances
is endogenously determined by individual's health status and socioeconomic
characteristics. In our empirical application, we analyze how health
insurance coverage affects an individual's decision to visit a doctor.
In this example, $Y$ is a binary outcome variable indicating whether
an individual visited a doctor's office, and $D$ is the endogenous
treatment variable that indicates whether an individual has her own
private insurance. 

We use the 2010 wave of the Medical Expenditure Panel Survey (MEPS)
as our main data source. We focus on all the visits happened in January,
2010. We restrict the sample to contain individuals with age between
25 and 64, and exclude individuals who have retained any kinds of
federal or state insurance in 2010. For $Z$, we consider two instrumental
variables that are used in \citet{zimmer2018using}\textemdash the
number of employees in the firm at which the individual works and
a dummy variable that indicates whether a firm has multiple locations.
These variables reflect how big the firm is, and the underlying rationale
for using these variables as instruments is as follows: the bigger
the firm is, the more likely it provides fringe benefits including
health insurance. Therefore, it is likely that these instruments affect
insurance status. We can argue, however, that they do not have direct
effects on decisions to visit doctors.\footnote{Note that it is difficult to justify these instruments for individuals
who are either self-employed or unemployed. To avoid this issue, we
exclude those individuals from our analysis.} We assume that these variables are exogenous conditional on covariates.
For additional covariates $W$, we include age, gender, years of education,
family size (the number of family members), income, region, race,
marital status, subjective physical and mental health status evaluations,
and whether living in a metropolitan statistical area. For the exogenous
variable $X$ in our model, we use information about the provision
of paid sick leave, which is separately collected from the National
Compensation Survey published by the U.S. Bureau of Labor Statistics.
We match the information for various industries with the primary dataset
we use. Conditional on the covariates listed above, we assume that
the number of sick leave days and leave benefits are exogenous, by
the same argument as for the instruments. Since $X$ and $Z$ are
assumed to be exogenous only conditional on $W$, we rely on Assumption
\ref{as:independence}$^{\prime}$ instead of Assumption \ref{as:independence}
for identification.

Since we include various control variables, one may concern that the
resulting estimators are imprecise with a moderate sample size. It
is worth emphasizing, however, that our semiparametric estimators
do not suffer from the curse of dimensionality as theoretically shown
in Section \ref{sec:Asymptotic}. This is because of the parametric
index structure in our model. Moreover, we do not attempt to estimate
the distributions of the unobservables conditional on these covariates,
but only estimate the marginal distributions.

Table \ref{tab:Summary} summarizes the variables used in estimation
and shows their summary statistics. While $65.7\%$ of individuals
had private health insurances in January 2010, only $18.2\%$ of them
visited doctors during the period. We use two variables for the pay
sick leave provision (i.e., $X$)\textemdash within each industry,
the percentage of workers who are provided with paid sick leave benefits
and the percentage of workers who are provided with a fixed number
of days for sick leave per year. The summary statistics for these
two variables show that there are sufficient variations across individuals
in different industries. Note that all the continuous variables are
standardized in order to ensure stability in estimation.\footnote{That is, for a continuous random variable $X$, define $\tilde{X}=\frac{X-\bar{X}_{n}}{\hat{sd}(X)}$,
where $\bar{X}_{n}$ and $\hat{sd}(X)$ are the sample average and
standard deviation of $X$, respectively.}

Before estimating the parametric and semiparametric models, we run
a first-stage OLS regression of $D$ on $X$, $W$, and $Z$ to see
if the excluded instruments are weak. The $F$-statistic value is
$167.19$, and thus we assume that the instruments are strong.\footnote{The $F$-statistic in the first-stage linear regression may not be
the best indicator for detecting weak instruments in nonlinear models.
\citet{han2019weakestimation} develop inference methods that are
robust to weak identification for a class of nonlinear models, and
consider bivariate probit models as one of the leading examples.} For the normalization of the parametric model, we you the convention\textemdash $E[\epsilon]=E[\nu]=0$
and $Var(\epsilon)=Var(\nu)=1$. On the other hand, for the semiparametric
model, we impose the normalization used in our simulation studies\textemdash i.e.
we exclude the constant terms and the coefficients on $sick34$ are
fixed to be corresponding parametric estimates.  We choose the Gaussian
copula to capture the dependence structure between $\epsilon$ and
$\nu$. In both models, the standard errors are obtained by the bootstrap
procedure (Section \ref{subsec:bootstrap}), where the bootstrap weights
are generated from the exponential distribution with the parameter
value $1$.

Tables \ref{tab:pe_selection-3234} and \ref{tab:pe_outcome-3234}
present the estimation results for the selection equation and the
outcome equation, respectively. Between the parametric and semiparametric
models, the magnitude and significance of the estimates differs, although,
overall, the signs of the estimates are similar. Table \ref{tab:pe_ates-3234}
shows the ATE estimates evaluated at various values of $X$, as well
as the estimates of the copula parameter $\rho$. The parametric estimate
of $\rho$ is statistically significant under $5\%$ level, whereas
the semiparametric estimate is not. We can find that the parametric
estimates of the ATE are different from the corresponding semiparametric
estimates. For example, the parametric ATE estimate evaluated at the
$50\%$ quantile of $X$ is about $0.129$, which means that having
private insurance increases the probability of visiting doctors by
$12.9\%$. On the other hand, the corresponding semiparametric estimate
shows that the effect is $10.4\%$. The discrepancy in the ATE estimates
between the parametric and semiparametric models suggests the possible
misspecification of the marginals, which is consistent with the premise
of this paper.

\section{\label{sec:Conclusions}Conclusions}

In this paper, we propose semiparametric estimation and inference
methods for generalized bivariate probit models. Specifically, we
develop the asymptotic theory for the sieve ML estimators of semiparametric
copula-based triangular systems with binary endogenous variables.
We show that the sieve ML estimators are consistent and that their
smooth functionals are $\sqrt{n}$-asymptotically normal under some
regularity conditions. This semiparametric estimation approach allows
for flexibility in the models and thus provides robustness in estimation
and inference. 

We conduct a sensitivity analysis to examine how sensitive the estimation
results are to model specifications. The results show that, overall,
the semiparametric sieve ML estimators perform well in terms of both
bias and variance. When the marginal distributions are misspecified,
the sieve ML estimators substantially outperform the parametric ML
estimators and the latter exhibit substantial bias. In particular,
we find that the parametric estimates of the parameters involving
the misspecified marginal distributions, such as the ATE, are highly
misleading. When the model is correctly specified, we find that the
performance of the sieve ML estimators is comparable to that of the
parametric ones. When the copula is also misspecified, the distortion
of the parametric estimates under misspecification of the marginals
can become even more severe, whereas the semiparametric estimates
do not seem to be affected by this misspecification as long as the
copula of the true DGP is within the stochastic ordering class. A
related and interesting question is how the results would change if
the data are not generated from this class of copulas.

We also formally show that the exclusion restriction is not only sufficient,
but is also necessary for identification. Without the exclusion restriction,
the model parameters are not identified or, under the normality assumption,
are, at best, weakly identified. Some empirical studies ignore the
exclusion restriction when estimating the model, and our non-identification
result provides a caveat for practitioners.

\bibliographystyle{chicago}
\bibliography{Est_Tri,Est_Tri2}

\newpage{}

\begin{table}[H]
\begin{centering}
\caption{\label{tab:correct_500}Correct Specification ($n=500$) (True marginal:
normal)}
\bigskip{}
\par\end{centering}
\centering{}{\small{}}%
\begin{tabular}{|c|c|c|c|c||c|c|c|c|c|}
\hline 
\multicolumn{5}{|c||}{Parametric Estimation, Gaussian Copula} & \multicolumn{5}{c|}{Semiparametric Estimation, Gaussian Copula}\tabularnewline
\hline 
 & {\small{}$\gamma$} & {\small{}$\delta_{1}$} & {\small{}$\rho_{sp}$} & {\small{}$ATE$} &  & {\small{}$\gamma$} & {\small{}$\delta_{1}$} & {\small{}$\rho_{sp}$} & {\small{}$ATE$}\tabularnewline
\hline 
{\small{}True Values} & {\small{}0.8000} & {\small{}1.1000} & {\small{}0.5000} & {\small{}0.3643} & {\small{}True Values} & {\small{}0.8000} & {\small{}1.1000} & {\small{}0.5000} & {\small{}0.3643}\tabularnewline
\hline 
{\small{}Estimate} & {\small{}0.8074} & {\small{}1.1469} & {\small{}0.4956} & {\small{}0.3657} & {\small{}Estimate} & {\small{}0.8070} & {\small{}1.1577} & {\small{}0.5037} & {\small{}0.3584}\tabularnewline
\hline 
{\small{}S.D} & {\small{}0.0934} & {\small{}0.3954} & {\small{}0.1537} & {\small{}0.0897} & {\small{}S.D} & {\small{}0.0940} & {\small{}0.4141} & {\small{}0.1528} & {\small{}0.0935}\tabularnewline
\hline 
{\small{}Bias} & {\small{}0.0074} & {\small{}0.0469} & {\small{}-0.0044} & {\small{}0.0014} & {\small{}Bias} & {\small{}0.0070} & {\small{}0.0577} & {\small{}0.0038} & {\small{}-0.0060}\tabularnewline
\hline 
{\small{}RMSE} & {\small{}0.0936} & {\small{}0.3982} & {\small{}0.1537} & {\small{}0.0897} & {\small{}RMSE} & {\small{}0.0943} & {\small{}0.4181} & {\small{}0.1528} & {\small{}0.0937}\tabularnewline
\hline 
\multicolumn{5}{|c||}{Parametric Estimation, Frank Copula} & \multicolumn{5}{c|}{Semiparametric Estimation, Frank Copula}\tabularnewline
\hline 
\end{tabular}{\small\par}
\end{table}
\begin{table}[H]
\caption{\label{tab:mar_500}Misspecification of Marginals ($n=500$) (True
marginal: mixture of normals)}
\bigskip{}

\centering{}{\small{}}%
\begin{tabular}{|c|c|c|c|c||c|c|c|c|c|}
\hline 
\multicolumn{5}{|c||}{Parametric Estimation, Gaussian Copula} & \multicolumn{5}{c|}{Semiparametric Estimation, Gaussian Copula}\tabularnewline
\hline 
 & {\small{}$\gamma$} & {\small{}$\delta_{1}$} & {\small{}$\rho_{sp}$} & {\small{}$ATE$} &  & {\small{}$\gamma$} & {\small{}$\delta_{1}$} & {\small{}$\rho_{sp}$} & {\small{}$ATE$}\tabularnewline
\hline 
{\small{}True Values} & {\small{}0.8000} & {\small{}1.1000} & {\small{}0.5000} & {\small{}0.1066} & {\small{}True Values} & {\small{}0.8000} & {\small{}1.1000} & {\small{}0.5000} & {\small{}0.1066}\tabularnewline
\hline 
{\small{}Estimate} & {\small{}0.7994} & {\small{}1.0925} & {\small{}0.4496} & {\small{}0.2443} & {\small{}Estimate} & {\small{}0.8562} & {\small{}1.2696} & {\small{}0.4895} & {\small{}0.1241}\tabularnewline
\hline 
{\small{}S.D} & {\small{}0.1281} & {\small{}0.6285} & {\small{}0.1651} & {\small{}0.1129} & {\small{}S.D} & {\small{}0.1113} & {\small{}0.3728} & {\small{}0.1059} & {\small{}0.0653}\tabularnewline
\hline 
{\small{}Bias} & {\small{}-0.0006} & {\small{}-0.0075} & {\small{}-0.0504} & {\small{}0.1377} & {\small{}Bias} & {\small{}0.0562} & {\small{}0.1696} & {\small{}-0.0105} & {\small{}0.0174}\tabularnewline
\hline 
{\small{}RMSE} & {\small{}0.1281} & {\small{}0.6285} & {\small{}0.1726} & {\small{}0.1780} & {\small{}RMSE} & {\small{}0.1247} & {\small{}0.4096} & {\small{}0.1064} & {\small{}0.0675}\tabularnewline
\hline 
\multicolumn{5}{|c||}{Parametric Estimation, Frank Copula} & \multicolumn{5}{c|}{Semiparametric Estimation, Frank Copula}\tabularnewline
\hline 
\end{tabular}{\small\par}
\end{table}
\begin{table}[H]
\caption{\label{tab:Summary}Summary Statistics}

\bigskip{}

\begin{centering}
\begin{tabular}{cccccc}
\hline 
 & {\scriptsize{}Variable} & {\scriptsize{}Mean} & {\scriptsize{}S.D} & {\scriptsize{}Min} & {\scriptsize{}Max}\tabularnewline
\hline 
{\scriptsize{}$Y$} & {\scriptsize{}Whether or not visit doctors} & {\scriptsize{}0.182} & {\scriptsize{}0.386} & {\scriptsize{}0} & {\scriptsize{}1}\tabularnewline
\hline 
{\scriptsize{}$D$} & {\scriptsize{}Whether or not have insurance} & {\scriptsize{}0.657} & {\scriptsize{}0.475} & {\scriptsize{}0} & {\scriptsize{}1}\tabularnewline
\hline 
\multirow{17}{*}{{\scriptsize{}$W$}} & {\scriptsize{}Age} & {\scriptsize{}42.591} & {\scriptsize{}10.574} & {\scriptsize{}25} & {\scriptsize{}64}\tabularnewline
 & {\scriptsize{}Years of education} & {\scriptsize{}13.433} & {\scriptsize{}2.892} & {\scriptsize{}0} & {\scriptsize{}17}\tabularnewline
 & {\scriptsize{}Income (hourly)} & {\scriptsize{}20.094} & {\scriptsize{}11.990} & {\scriptsize{}0.4} & {\scriptsize{}73.08}\tabularnewline
 & {\scriptsize{}Family size} & {\scriptsize{}2.932} & {\scriptsize{}1.595} & {\scriptsize{}1} & {\scriptsize{}14}\tabularnewline
 & {\scriptsize{}Living in MSA} & {\scriptsize{}0.868} & {\scriptsize{}0.338} & {\scriptsize{}0} & {\scriptsize{}1}\tabularnewline
\hline 
\end{tabular}
\par\end{centering}
\bigskip{}

{\footnotesize{}$\dagger$: The original variables for these variables
are coded into 5 groups - Excellent, Very Good, Good, Fair, and Poor.
These variables show how much portion of individuals in the sample
considers their physical/mental health is below Good (i.e. Fair or
Poor). }{\footnotesize\par}
\end{table}
\begin{table}[H]
\caption{\label{tab:pe_selection-3234}Estimates in Selection Equation}

\bigskip{}

\begin{centering}
\begin{tabular}{ccc}
\hline 
 & {\footnotesize{}Parametric } & {\footnotesize{}Semiparametric}\tabularnewline
\hline 
{\footnotesize{}age$\dagger$} & {\footnotesize{}0.130{*}{*}{*}} & {\footnotesize{}0.077{*}{*}{*}}\tabularnewline
 & {\footnotesize{}(0.018)} & {\footnotesize{}(0.038)}\tabularnewline
{\footnotesize{}years of education$\dagger$} & {\footnotesize{}0.190{*}{*}{*}} & {\footnotesize{}0.098{*}{*}}\tabularnewline
 & {\footnotesize{}(0.018)} & {\footnotesize{}(0.044)}\tabularnewline
{\footnotesize{}family size$\dagger$} & {\footnotesize{}-0.120{*}{*}{*}} & {\footnotesize{}-0.041{*}}\tabularnewline
 & {\footnotesize{}(0.017)} & {\footnotesize{}(0.023)}\tabularnewline
{\footnotesize{}income$\dagger$} & {\footnotesize{}0.268{*}{*}{*}} & {\footnotesize{}0.416{*}{*}{*}}\tabularnewline
\hline 
\end{tabular}
\par\end{centering}
\bigskip{}

\begin{itemize}
\item {\scriptsize{}$\dagger$ indicates that the variable is standardized. }{\scriptsize\par}
\item {\scriptsize{}The coefficient on T34 in the semiparametric model is
fixed for normalization. }{\scriptsize\par}
\item {\scriptsize{}Gaussian copula is used. }{\scriptsize\par}
\end{itemize}
\end{table}
\begin{table}[H]
\caption{\label{tab:pe_outcome-3234}Estimates in Outcome Equation}

\bigskip{}

\begin{centering}
\begin{tabular}{ccc}
\hline 
 & {\footnotesize{}Parametric } & {\footnotesize{}Semiparametric}\tabularnewline
\hline 
{\footnotesize{}Treatment ($\delta$)} & {\footnotesize{}0.493{*}{*}{*}} & {\footnotesize{}0.368{*}{*}}\tabularnewline
 & {\footnotesize{}(0.168)} & {\footnotesize{}(0.183)}\tabularnewline
{\footnotesize{}age$\dagger$} & {\footnotesize{}0.055{*}{*}{*}} & {\footnotesize{}0.059}\tabularnewline
 & {\footnotesize{}(0.020)} & {\footnotesize{}(0.047)}\tabularnewline
{\footnotesize{}years of education$\dagger$} & {\footnotesize{}0.142{*}{*}{*}} & {\footnotesize{}0.126{*}}\tabularnewline
 & {\footnotesize{}(0.028)} & {\footnotesize{}(0.066)}\tabularnewline
{\footnotesize{}family size$\dagger$} & {\footnotesize{}-0.055{*}{*}{*}} & {\footnotesize{}-0.052{*}}\tabularnewline
\hline 
\end{tabular}
\par\end{centering}
\bigskip{}

\begin{itemize}
\item {\scriptsize{}$\dagger$ indicates that the variable is standardized. }{\scriptsize\par}
\item {\scriptsize{}The coefficient on T34 in the semiparametric model is
fixed for normalization. }{\scriptsize\par}
\item {\scriptsize{}Gaussian copula is used.}{\scriptsize\par}
\end{itemize}
\end{table}
\begin{table}
\caption{\label{tab:pe_ates-3234}Estimated ATE's and Spearman's $\rho$}

\bigskip{}

\begin{centering}
\begin{tabular}{ccc}
\hline 
 & {\footnotesize{}Parametric } & {\footnotesize{}Semiparametric}\tabularnewline
\hline 
{\footnotesize{}ATE at the mean} & {\footnotesize{}0.114{*}{*}{*}} & {\footnotesize{}0.100{*}{*}}\tabularnewline
 & {\footnotesize{}(0.037)} & {\footnotesize{}(0.048)}\tabularnewline
{\footnotesize{}ATE at 50\% quantile} & {\footnotesize{}0.129{*}{*}{*}} & {\footnotesize{}0.104{*}}\tabularnewline
 & {\footnotesize{}(0.045)} & {\footnotesize{}(0.054)}\tabularnewline
{\footnotesize{}ATE at 25\% quantile} & {\footnotesize{}0.121{*}{*}} & {\footnotesize{}0.104}\tabularnewline
 & {\footnotesize{}(0.050)} & {\footnotesize{}(0.058)}\tabularnewline
{\footnotesize{}ATE at 75\% quantile} & {\footnotesize{}0.139{*}{*}{*}} & {\footnotesize{}0.105{*}}\tabularnewline
 & {\footnotesize{}(0.043)} & {\footnotesize{}(0.056)}\tabularnewline
{\footnotesize{}Spearman's $\rho$} & {\footnotesize{}-0.200{*}{*}} & {\footnotesize{}-0.154}\tabularnewline
 & {\footnotesize{}(0.105)} & {\footnotesize{}(0.134)}\tabularnewline
\hline 
{\footnotesize{}Number of Observations} & {\footnotesize{}7,555} & {\footnotesize{}7,555}\tabularnewline
\hline 
\multicolumn{3}{c}{{\scriptsize{}Standard errors in parentheses, {*} $p<0.10$, {*}{*}
$p<0.05$, {*}{*}{*} $p<0.01$. }}\tabularnewline
\hline 
\end{tabular}
\par\end{centering}
\bigskip{}
\end{table}
\newpage{}

\begin{appendix}

\part*{Online Appendix}

\section{\label{sec:ID_proofs}Proofs of Results in Section \ref{sec:Identification}}

\subsection{Proof of Theorem \ref{thm:lack}}

Recall 
\[
\begin{array}{cc}
q_{0}\equiv F_{\tilde{\nu}}(-\mu_{\nu}/\sigma_{\nu}), & \quad q_{1}\equiv F_{\tilde{\nu}}((1-\mu_{\nu})/\sigma_{\nu}),\\
t_{0}\equiv F_{\tilde{\varepsilon}}(-\mu_{\varepsilon}/\sigma_{\varepsilon}), & \quad t_{1}\equiv F_{\tilde{\varepsilon}}((1-\mu_{\varepsilon})/\sigma_{\varepsilon}),
\end{array}
\]
and
\begin{align*}
\tilde{p}_{11,0} & =C(F_{\tilde{\varepsilon}}(F_{\tilde{\varepsilon}}^{-1}(t_{0})+\delta_{1}),q_{0};\rho),\\
\tilde{p}_{11,1} & =C(F_{\tilde{\varepsilon}}(F_{\tilde{\varepsilon}}^{-1}(t_{1})+\delta_{1}),q_{1};\rho),\\
\tilde{p}_{10,0} & =t_{0}-C(t_{0},q_{0};\rho),\\
\tilde{p}_{10,1} & =t_{1}-C(t_{1},q_{1};\rho),\\
\tilde{p}_{00,0} & =1-t_{0}-q_{0}+C(t_{0},q_{0};\rho),\\
\tilde{p}_{00,1} & =1-t_{1}-q_{1}+C(t_{1},q_{1};\rho),
\end{align*}
where $\tilde{p}_{yd,x}\equiv\Pr[Y=y,D=d|X_{1}=x]$. Again, we want
to show that, given $(q_{0},q_{1})$ which are identified from the
reduced-form equation, there are two distinct sets of parameter values
$(t_{0},t_{1},\delta_{1},\rho)$ and $(t_{0}^{*},t_{1}^{*},\delta_{1}^{*},\rho^{*})$
(with $(t_{0},t_{1},\delta_{1},\rho)$ $\ne$ $(t_{0}^{*},t_{1}^{*},\delta_{1}^{*},\rho^{*})$)
that generate the same observed fitted probabilities $\tilde{p}_{yd,0}$
and $\tilde{p}_{yd,1}$ for all $(y,d)\in\{0,1\}^{2}$. In showing
this, the following lemma is useful:

\begin{lemma}\label{lem:si_to_c}Assumption \ref{as_copula} implies
that, for any $(u_{1},u_{2})\in(0,1)^{2}$ and $\rho\in\Omega$, 
\begin{equation}
C_{\rho}(u_{1},u_{2};\rho)>0.\label{concordance}
\end{equation}
\end{lemma}The proof of this lemma can be found below.

Now fix $(q_{0},q_{1})\in(0,1)^{2}$. First, consider the fitted probability
$\tilde{p}_{10,0}$. Given $t_{0}\in(0,1)$ and $\rho\in\Omega$,
note that, for $\rho^{*}>\rho$,\footnote{The inequality here and other inequalities implied from this (e.g.,
$t_{0}^{*}>t_{0}$, and etc.) are assumed only for concreteness.} there exists a solution $t_{0}^{*}=t_{0}^{*}(t_{0},q_{0},\rho,\rho^{*})$
such that 
\begin{align}
t_{0}-C(t_{0},q_{0};\rho) & =\Pr[u_{1}\leq t_{0},u_{2}\ge q_{0};\rho]\label{eq:quad1}\\
 & =\Pr[u_{1}\leq t_{0}^{*},u_{2}\ge q_{0};\rho^{*}]\label{eq:quad2}\\
 & =t_{0}^{*}-C(t_{0}^{*},q_{0};\rho^{*}),\nonumber 
\end{align}
and note that by Assumption \ref{as_copula} and a variant of Lemma
\ref{lem:si_to_c}, we have that $t_{0}^{*}>t_{0}$. Here, $(t_{0},q_{0},\rho)$
and $(t_{0}^{*},q_{0},\rho^{*})$ result in the same observed probability
$\tilde{p}_{10,0}=t_{0}-C(t_{0},q_{0};\rho)=t_{0}^{*}-C(t_{0}^{*},q_{0};\rho^{*})$.
Now consider the fitted probability $\tilde{p}_{11,0}$. Choose $\delta_{1}=0$.
Also let $F_{\tilde{\varepsilon}}\sim Unif(0,1)$ only for simplicity,
which is relaxed later. Then there exists a solution $t_{0}^{\dagger}=t_{0}^{\dagger}(t_{0},q_{0},\rho,\rho^{*})$
such that 
\begin{align}
C(t_{0},q_{0};\rho) & =\Pr[u_{1}\leq t_{0},u_{2}\leq q_{0};\rho]\label{eq:quad3}\\
 & =\Pr[u_{1}\leq t_{0}^{\dagger},u_{2}\leq q_{0};\rho^{*}]\label{eq:quad4}\\
 & =C(t_{0}^{\dagger},q_{0};\rho^{*}),\nonumber 
\end{align}
and note that $t_{0}^{\dagger}<t_{0}$ by Assumption \ref{as_copula}
and Lemma \ref{lem:si_to_c}. Then, by letting $\delta_{1}^{*}=t_{0}^{\dagger}-t_{0}^{*}$,
$(t_{0},q_{0},\delta_{1},\rho)$ and $(t_{0}^{*},q_{0},\delta_{1}^{*},\rho^{*})$
satisfy $\tilde{p}_{11,0}=C(t_{0}+0,q_{0};\rho)=C(t_{0}^{*}+\delta_{1}^{*},q_{0};\rho^{*})$.
Lastly, note that $\tilde{p}_{00,0}=1-q_{0}-\tilde{p}_{10,0}$ and
$\tilde{p}_{01,0}=q_{0}-\tilde{p}_{11,0}$, and so $(t_{0},\delta_{1},\rho)$
and $(t_{0}^{*},\delta_{1}^{*},\rho^{*})$ above will also result
in the same values of $\tilde{p}_{00,0}$ and $\tilde{p}_{01,0}$.

It is tempting to have a parallel argument for $\tilde{p}_{10,1}$,
$\tilde{p}_{11,1}$, $\tilde{p}_{00,1}$, and $\tilde{p}_{01,1}$,
but there is a complication. Although other parameters are not, $\delta_{1}$
and $\rho$ are common in both sets of probabilities. Therefore, we
proceed as follows. First, consider $\tilde{p}_{10,1}$. Given $t_{1}\in(0,1)$
and the above choice of $\rho^{*}\in\Omega$, note that there exists
a solution $t_{1}^{*}=t_{1}^{*}(t_{1},q_{1},\rho,\rho^{*})$ such
that 
\begin{align}
t_{1}-C(t_{1},q_{1};\rho) & =\Pr[u_{1}\leq t_{1},u_{2}\ge q_{1};\rho]\label{eq:quad5}\\
 & =\Pr[u_{1}\leq t_{1}^{*},u_{2}\ge q_{1};\rho^{*}]\label{eq:quad6}\\
 & =t_{1}^{*}-C(t_{1}^{*},q_{1};\rho^{*}),\nonumber 
\end{align}
and similarly as before, we have $t_{1}^{*}>t_{1}$. Here, $(t_{1},q_{1},\rho)$
and $(t_{1}^{*},q_{1},\rho^{*})$ result in the same observed probability
$\tilde{p}_{10,1}=t_{1}-C(t_{1},q_{1};\rho)=t_{1}^{*}-C(t_{1}^{*},q_{1};\rho^{*})$.
Now consider $\tilde{p}_{11,1}$. Recall $\delta_{1}=0$ and $F_{\varepsilon}\sim Unif(0,1)$.
Then there exists a solution $t_{1}^{\dagger}=t_{1}^{\dagger}(t_{1},q_{1},\rho,\rho^{*})$
such that 
\begin{align}
C(t_{1},q_{1};\rho) & =\Pr[u_{1}\leq t_{1},u_{2}\leq q_{1};\rho]\label{eq:quad7}\\
 & =\Pr[u_{1}\leq t_{1}^{\dagger},u_{2}\leq q_{1};\rho^{*}]\label{eq:quad8}\\
 & =C(t_{1}^{\dagger},q_{1};\rho^{*}),\nonumber 
\end{align}
and thus $t_{1}^{\dagger}<t_{1}$. Then, if we can show that 
\begin{equation}
t_{1}^{\dagger}=t_{1}^{*}+\delta_{1}^{*},\label{eq:t_dagger}
\end{equation}
where $t_{1}^{*}$ and $\delta_{1}^{*}$ are the values already determined
above, then $(t_{1},q_{1},\delta_{1},\rho)$ and $(t_{1}^{*},q_{1},\delta_{1}^{*},\rho^{*})$
result in $\tilde{p}_{11,1}=C(t_{1}+0,q_{1};\rho)=C(t_{1}^{*}+\delta_{1}^{*},q_{1};\rho^{*})$.
Then similar as before, the two sets of parameters will generate the
same values of $\tilde{p}_{00,1}=1-q_{1}-\tilde{p}_{10,1}$ and $\tilde{p}_{01,1}=q_{1}-\tilde{p}_{11,1}$.
Consequently, $(t_{0},t_{1},q_{0},q_{1},\delta_{1},\rho)$ and $(t_{0}^{*},t_{1}^{*},q_{0},q_{1},\delta_{1}^{*},\rho^{*})$
generate the same entire observed fitted probabilities. The remaining
question is whether we can find $(t_{0},t_{1},\delta_{1},\rho)$ and
$(t_{0}^{*},t_{1}^{*},\delta_{1}^{*},\rho^{*})$ such that \eqref{eq:t_dagger}
holds.

To show this, we choose further specifications. We assume a normal
copula.\footnote{This choice is not critical except that we can have $\rho$ reach
to 1.} We choose $\rho=0$, $\rho^{*}=1$, $q_{0}=t_{0}=1/3$, and $q_{1}=t_{1}=2/3$.
Since $(U_{1},U_{2})$ are jointly uniform, note that when $\rho=0$,
the probability of the quadrant in $[0,1]^{2}$ specified by each
of \eqref{eq:quad1}, \eqref{eq:quad3}, \eqref{eq:quad5}, and \eqref{eq:quad7}
equals the volume of the quadrant. When $\rho^{*}=1$, all the probability
mass lies on the 45 degree line in $[0,1]^{2}$ and no where else,
so the probability of a quadrant specified by each of \eqref{eq:quad2},
\eqref{eq:quad4}, \eqref{eq:quad6}, and \eqref{eq:quad8} equals
the length of the 45 line which intersects with that quadrant. Suppose
that the following observational equivalence holds: 
\begin{align*}
\Pr[u_{1}\leq t_{0},u_{2}\ge q_{0};\rho] & =\Pr[u_{1}\leq t_{0}^{*},u_{2}\ge q_{0};\rho^{*}]=2/9,\\
\Pr[u_{1}\leq t_{0},u_{2}\leq q_{0};\rho] & =\Pr[u_{1}\leq t_{0}^{\dagger},u_{2}\leq q_{0};\rho^{*}]=1/9,\\
\Pr[u_{1}\leq t_{1},u_{2}\ge q_{1};\rho] & =\Pr[u_{1}\leq t_{1}^{*},u_{2}\ge q_{1};\rho^{*}]=2/9,\\
\Pr[u_{1}\leq t_{1},u_{2}\leq q_{1};\rho] & =\Pr[u_{1}\leq t_{1}^{\dagger},u_{2}\leq q_{1};\rho^{*}]=4/9.
\end{align*}
One can easily show that these equations yield that $t_{0}^{*}=5/9$,
$t_{0}^{\dagger}=1/9$, $t_{1}^{*}=8/9$, and $t_{1}^{\dagger}=4/9$.
Consider the equation \eqref{eq:t_dagger}, which can be rewritten
as $t_{1}^{\dagger}=t_{1}^{*}+t_{0}^{\dagger}-t_{0}^{*}$ or $t_{1}^{\dagger}-t_{1}^{*}=t_{0}^{\dagger}-t_{0}^{*}$.
Then, note that we have $t_{1}^{\dagger}-t_{1}^{*}=t_{0}^{\dagger}-t_{0}^{*}=-4/9$,
which is, in fact, the value of $\delta_{1}^{*}$. In sum, the values
of parameters that give the observationally equivalent fitted probabilities
are 
\begin{align}
(t_{0},t_{1},q_{0},q_{1},\delta_{1},\rho) & =\left(\frac{1}{3},\frac{2}{3},\frac{1}{3},\frac{2}{3},0,0\right),\label{eq:counterex1}\\
(t_{0}^{*},t_{1}^{*},q_{0},q_{1},\delta_{1}^{*},\rho^{*}) & =\left(\frac{5}{9},\frac{8}{9},\frac{1}{3},\frac{2}{3},-\frac{4}{9},1\right).\label{eq:counterex2}
\end{align}

This argument can be made slightly more general, and thus the counterexample
more realistic, by relaxing $F_{\tilde{\varepsilon}}\sim Unif(0,1)$
and $\rho^{*}=1$. We show that a similar argument goes through with
$F_{\tilde{\varepsilon}}$ being a general distribution function with
a symmetric density function, and $-1\leq\rho^{*}\leq1$ as long as
the copula density is symmetric around $u_{2}=u_{1}$ (i.e., the 45
degree line) and $u_{2}=1-u_{1}$. Let $F\equiv F_{\tilde{\varepsilon}}$
be a general distribution whose density function is symmetric. Then
there exists a solution $s_{0}^{\dagger}=s_{0}^{\dagger}(t_{0},q_{0},\rho,\rho^{*})$
such that 
\begin{align*}
C(F(F^{-1}(t_{0})+0),q_{0};\rho) & =\Pr[u_{1}\leq t_{0},u_{2}\leq q_{0};\rho]\\
 & =\Pr[u_{1}\leq s_{0}^{\dagger},u_{2}\leq q_{0};\rho^{*}]\\
 & =C(s_{0}^{\dagger},q_{0};\rho^{*}).
\end{align*}
Then, by letting $\delta_{1}^{*}=F^{-1}(s_{0}^{\dagger})-F^{-1}(t_{0}^{*})$,
we have $s_{0}^{\dagger}=F(F^{-1}(t_{0}^{*})+\delta_{1}^{*})$ and
therefore $(t_{0},q_{0},\delta_{1},\rho)$ and $(t_{0}^{*},q_{0},\delta_{1}^{*},\rho^{*})$
result in $p_{11,x}=C(F(F^{-1}(t_{0})+0),q_{0};\rho)=C(F(F^{-1}(t_{0}^{*})+\delta_{1}^{*}),q_{0};\rho^{*})$.
Suppose that $\delta_{1}=0$. Then there exists a solution $s_{1}^{\dagger}=s_{1}^{\dagger}(t_{1},q_{1},\rho,\rho^{*})$
such that 
\begin{align*}
C(F(F^{-1}(t_{1})+0),q_{1};\rho) & =\Pr[u_{1}\leq t_{1},u_{2}\leq q_{1};\rho]\\
 & =\Pr[u_{1}\leq s_{1}^{\dagger},u_{2}\leq q_{1};\rho^{*}]\\
 & =C(s_{1}^{\dagger},q_{1};\rho^{*}).
\end{align*}
Then, if we can show that 
\[
F^{-1}(s_{1}^{\dagger})=F^{-1}(t_{1}^{*})+\delta_{1}^{*},
\]
then $s_{1}^{\dagger}=F(F^{-1}(t_{1}^{*})+\delta_{1}^{*})$ and therefore
$(t_{1},q_{1},\delta_{1},\rho)$ and $(t_{1}^{*},q_{1},\delta_{1}^{*},\rho^{*})$
result in $\tilde{p}_{11,1}=C(F(F^{-1}(t_{1})+0),q_{1};\rho)=C(F(F^{-1}(t_{1}^{*})+\delta_{1}^{*}),q_{1};\rho)$.
Note $F^{-1}(s_{1}^{\dagger})=F^{-1}(t_{1}^{*})+\delta_{1}^{*}$ can
be rewritten as $F^{-1}(s_{1}^{\dagger})=F^{-1}(t_{1}^{*})+F^{-1}(s_{0}^{\dagger})-F^{-1}(t_{0}^{*})$
or 
\begin{equation}
F^{-1}(s_{1}^{\dagger})-F^{-1}(t_{1}^{*})=F^{-1}(s_{0}^{\dagger})-F^{-1}(t_{0}^{*}).\label{eq:delta_star2}
\end{equation}
But note that since the density of $F$ is symmetric, any two values
$s$ and $\tilde{s}$ in $(0,1)$ that are symmetric around $u_{1}=1/2$
will satisfy $F^{-1}(s)=-F^{-1}(\tilde{s})$. Therefore, since in
our example $s_{0}^{\dagger}$ and $t_{1}^{*}$ are symmetric around
$u_{1}=1/2$, and so are $s_{1}^{\dagger}$ and $t_{0}^{*}$, we have
the desired result \eqref{eq:delta_star2}, and the counterexample
\eqref{eq:counterex1}\textendash \eqref{eq:counterex2} remains valid.
Note that the symmetry of the density function of $F$ plays a key
role here; the uniform distribution trivially satisfies the condition
as does the normal distribution.

The above counter-example to identification involves a parameter on
the boundary of the parameter space ($\rho^{*}=1$), while the identification
results in the paper assume that the parameter space is open and thus
that $\rho\in(-1,1)$. We now show that the key idea of the argument
remains the same with $-1<\rho^{*}<1$. Suppose that the copula density
is symmetric around $u_{2}=u_{1}$ and $u_{2}=1-u_{1}$. The normal
copula satisfies this condition for any $\rho\in(-1,1)$. Because
of this condition, the symmetry of $s_{0}^{\dagger}$ and $t_{1}^{*}$
(and of $s_{1}^{\dagger}$ and $t_{0}^{*}$) around $u_{1}=1/2$ does
not break at a different value of $\rho^{*}$, even though the values
of $s_{0}^{\dagger}$, $t_{1}^{*}$, $s_{1}^{\dagger}$, and $t_{0}^{*}$
themselves change. Therefore, \eqref{eq:delta_star2} continues to
hold with $\rho^{*}\neq1$.

\subsection{Proof of Lemma \ref{lem:si_to_c}}

The proof of Lemma \ref{lem:si_to_c} is a slight modification of
the proof of Theorem 2.14 of \citet[p. 44]{Joe1997}. Suppose $C_{2|1}\prec_{S}\tilde{C}_{2|1}$.
Let $(U_{1},U_{2})\sim C$, $(\tilde{U}_{1},\tilde{U}_{2})\sim\tilde{C}$,
with $U_{j}\overset{d}{=}\tilde{U}_{j}$, $j=1,2$. By Theorem 2.9
of \citet[p. 40]{Joe1997}, $(U_{1},U_{2})\overset{d}{=}(\tilde{U}_{1},\psi(U_{1},U_{2}))$
with $\psi(u_{1},u_{2})=\tilde{C}_{2|1}^{-1}(C_{2|1}(u_{2}|u_{1})|u_{1})$.
Since $C_{2|1}\prec_{S}\tilde{C}_{2|1}$, $\psi$ is increasing in
$u_{1}$ and $u_{2}$. We consider two cases: 
\begin{itemize}
\item Case 1: Suppose that $u_{1}$ and $u_{2}$ are such that $\psi(u_{1},u_{2})\leq u_{2}$.
Then 
\begin{align*}
\tilde{C}(u_{1},u_{2}) & =\Pr[\tilde{U}_{1}\leq u_{1},\tilde{U}_{2}\leq u_{2})]=\Pr[\tilde{U}_{1}<u_{1},\tilde{U}_{2}<u_{2})]\\
 & =\Pr[U_{1}<u_{1},\psi(U_{1},U_{2})<u_{2})]\geq\Pr[U_{1}<u_{1},\psi(u_{1},U_{2})<u_{2}]\\
 & >\Pr[U_{1}<u_{1},U_{2}<u_{2})]=C(u_{1},u_{2})
\end{align*}
where the strict inequality holds since $U_{2}<u_{2}$ implies $\psi(u_{1},U_{2})\leq\psi(u_{1},u_{2})\leq u_{2}$
(but not vice versa since $\psi(u_{1},U_{2})\leq u_{2}$ and $\psi(u_{1},u_{2})\leq u_{2}$
does not necessarily imply $U_{2}<u_{2}$ and $\Pr[\psi(u_{1},u_{2})<\psi(u_{1},U_{2})]=\Pr[u_{2}<U_{2}]\neq0$),
and the second last inequality holds since, given $U_{1}<u_{1}$,
$\psi(U_{1},U_{2})\leq\psi(u_{1},U_{2})<u_{2}$. 
\item Case 2: Suppose that $u_{1}$ and $u_{2}$ are such that $\psi(u_{1},u_{2})>u_{2}$.
Then 
\begin{align*}
u_{2}-C(u_{1},u_{2}) & =\Pr[U_{1}>u_{1},U_{2}<u_{2})]>\Pr[U_{1}>u_{1},\psi(u_{1},U_{2})\leq u_{2})]\\
 & \geq\Pr[U_{1}>u_{1},\psi(U_{1},U_{2})\leq u_{2})]=\Pr[\tilde{U}_{1}>u_{1},\tilde{U}_{2}<u_{2}]=u_{2}-\tilde{C}(u_{1},u_{2})
\end{align*}
where the strict inequality holds since $U_{2}>u_{2}$ implies $\psi(u_{1},U_{2})\geq\psi(u_{1},u_{2})>u_{2}$
or $\psi(u_{1},U_{2})\leq u_{2}$ implies $U_{2}\leq u_{2}$ (but
not vice versa). 
\end{itemize}
Therefore in both cases, $C(u_{1},u_{2})<\tilde{C}(u_{1},u_{2})$
for any $u_{1}$ and $u_{2}$.

\section{\label{sec:Est_proofs}Proofs of Results in Section \ref{sec:Asymptotic}}

\subsection{Identification under Transformation of Marginal Distribution Functions}

Recall that we consider the following specification of the marginal
distribution functions to derive the asymptotic theory for the sieve
ML estimator: 
\begin{equation}
F_{\epsilon0}(x)=H_{\epsilon0}(G(x)),\quad F_{\nu0}(x)=H_{\nu0}(G_{\nu}(x)),\label{eq:app_transform}
\end{equation}
 where $G:\mathbb{R}\rightarrow[0,1]$ is a strictly increasing function
with its derivative $g(x)\equiv\frac{dG(x)}{dx}$ and $g(x)$ is bounded
away from zero on $\mathbb{R}$.

We first verify that there exist $H_{\epsilon0}$ and $H_{\nu0}$
that satisfy \eqref{eq:app_transform} for given $F_{\epsilon0}$,
$F_{\nu0}$, and $G$. Since $G$ is assumed to be strictly increasing,
there exists an inverse function $G^{-1}$. Letting $H_{\epsilon0}(\cdot)=F_{\epsilon0}(G^{-1}(\cdot))$
and $H_{\nu0}(\cdot)=F_{\nu0}(G^{-1}(\cdot))$, it is straightforward
to show that $H_{\epsilon0}$ and $H_{\nu0}$ are mappings from $[0,1]$
to $[0,1]$ and that satisfy the relations in \eqref{eq:app_transform}.
Note too that this transformation does not change the identification
results. That is, $F_{0}$ is identified on $\mathbb{R}$ if and only
if $H_{0}$ is identified on $[0,1]$. Assuming that $g$ is bounded
away from zero on $\mathbb{R}$ and bounded above, the unknown density
function $h_{j0}$ can be written as $h_{j0}(x)=\frac{f_{j0}(G^{-1}(x))}{g(G^{-1}(x))}$
for each $j\in\{\epsilon,\nu\}$, which is well-defined on $[0,1]$.
In addition, we can see that $h_{\epsilon0}$ and $h_{\nu0}$ are
identified if and only if the unknown marginal density functions $f_{\epsilon0}$
and $f_{\nu0}$ are identified. 

We note that the choice of $G$ depends on the tail behavior of $f_{\epsilon0}$
and $f_{\nu0}$. If researchers believe that the unknown marginal
density functions have fat tails, then they should choose a distribution
function with fat tails for $G$. On the other hand, one can choose
the logistic or the standard normal distribution function for $G$
when $f_{\epsilon0}$ and $f_{\nu0}$ are likely to have thin tails.
This is because Assumption \ref{as:para_sp} implicitly requires that
the unknown marginal density functions and $g$ decay at the same
rate at the tails. Specifically, we observe that 
\[
h_{\epsilon0}(0)=\lim_{x\rightarrow0^{+}}h_{\epsilon0}(x)=\lim_{t\rightarrow-\infty}\frac{f_{\epsilon0}(t)}{g(t)},
\]
 and the limit exists if the decaying rates are of the same order.
We also provide simulation results to examine how the performance
of our semiparametric estimator varies across the choice of $G$ when
the marginal density functions have fat tails (see Section \eqref{subsec:fat_tail}). 

\subsection{\label{subsec:Tech_Expressions}Technical Expressions}

\subsubsection{H\"older Norm and H\"older Class}

Let $\mathcal{C}^{m}(\mathcal{X})$ be the space of $m$-times continuously
differentiable real-valued functions on $\mathcal{X}$. Let $\zeta\in(0,1]$
and, given a $d$-tuple $\omega$, let $[\omega]=\omega_{1}+...+\omega_{d}$.
Denote the differential operator by $\mathcal{D}$ and let $\mathcal{D}^{\omega}=\frac{\partial^{[\omega]}}{\partial x_{1}^{\omega_{1}}...\partial x_{d}^{\omega_{d}}}$.
Letting $p=m+\zeta$, the H\"older norm of $h\in\mathcal{C}^{m}(\mathcal{X})$
is defined as follows: 

\[
||h||_{\Lambda^{p}}\equiv\sup_{[\omega]\leq m,x}|\mathcal{D}^{\omega}h(x)|+\sup_{[\omega]=m}\sup_{x,y\in\mathcal{X},||x-y||_{E}\neq0}\frac{|\mathcal{D}^{\omega}h(x)-\mathcal{D}^{\omega}h(y)|}{||x-y||_{E}^{\zeta}},
\]
where $\zeta$ is the H\"older exponent. 

A H\"older class with smoothness $p>0$, denoted by $\Lambda^{p}(\mathcal{X})$,
is defined as $\Lambda^{p}(\mathcal{X})\equiv\{h\in\mathcal{C}^{m}(\mathcal{X}):||h||_{\Lambda^{p}}<\infty\}$.
A H\"older ball with radius $R$, $\Lambda_{R}^{p}(\mathcal{X})$,
is defined as $\Lambda_{R}^{p}(\mathcal{X})\equiv\{h\in\Lambda^{p}(\mathcal{X}):||h||_{\Lambda^{p}}\leq R<\infty\}$. 

\subsubsection{Sup-norm and Pseudo-metric $d_{c}$ }

For any $h\in\mathcal{H}_{\epsilon}$ (or $\mathcal{H}_{\nu}$), define
the sup-norm on $\mathcal{H}_{\epsilon}$ (or $\mathcal{H}_{\nu}$)
as follows: 
\[
||h||_{\infty}\equiv\sup_{t\in[0,1]}|h(t)|.
\]
 Let $\theta=(\psi^{'},h_{\epsilon},h_{\nu})^{'}\in\Theta$ be given.
We define the consistency norm $||\cdot||_{c}$ as follows: 
\[
||\theta||_{c}\equiv||\psi||_{E}+||h_{\epsilon}||_{\infty}+||h_{\nu}||_{\mathcal{\infty}},
\]
 where $||\cdot||_{E}$ is the Euclidean norm. The pseudo-metric $d_{c}(\cdot,\cdot):\Theta\times\Theta\rightarrow[0,\infty)$,
which induced by the consistency norm $||\cdot||_{c}$, is defined
as 
\[
d_{c}(\theta_{1},\theta_{2})=||\theta_{1}-\theta_{2}||_{c}.
\]

\subsubsection{$L^{2}$-norm }

\begin{equation}
||\theta-\theta_{0}||_{2}\equiv||\psi-\psi_{0}||_{E}+||h_{\epsilon}-h_{\epsilon0}||_{2}+||h_{\nu}-h_{\nu0}||_{2},\label{eq:cr-norm}
\end{equation}
 where $||h-\tilde{h}||_{2}^{2}\equiv\int_{0}^{1}(h(t)-\tilde{h}(t))^{2}dt$
for any $h,\tilde{h}\in\mathcal{H}$. It is straightforward to show
that $||\theta-\theta_{0}||_{2}\leq d_{c}(\theta,\theta_{0})$, where
$d_{c}(\theta,\theta_{0})=||\psi-\psi_{0}||_{E}+||h_{\epsilon}-h_{\epsilon0}||_{\infty}+||h_{\nu}-h_{\nu0}||_{\infty}$. 

\subsubsection{Fisher inner product and Fisher norm }

Recall that $\mathbb{V}$ is the linear span of $\Theta-\{\theta_{0}\}$.
Define the Fisher inner product on the space $\mathbb{V}$ as 
\[
<v,\tilde{v}>\equiv E\left[(\frac{\partial l(\theta_{0},W)}{\partial\theta}[v])(\frac{\partial l(\theta_{0},W)}{\partial\theta}[\tilde{v}])\right]
\]
 for given $v,\tilde{v}\in\mathbb{V}$. Then, the Fisher norm for
$v\in\mathbb{V}$ is defined as 
\[
||v||^{2}\equiv<v,v>.
\]

\subsubsection{Relationship between the Fisher norm and $L^{2}$-norm }

Note that for any $\theta_{1},\theta_{2}\in\Theta$, we have 
\begin{align}
||\theta_{1}-\theta_{2}||^{2} & =E\left[(\frac{\partial l(\theta_{0},W_{i})}{\partial\theta}[\theta_{1}-\theta_{2}])^{2}\right]\nonumber \\
 & \leq B\left\{ E\left[\{\frac{\partial l(\theta_{0},W_{i})}{\partial\psi^{'}}(\psi_{1}-\psi_{2})\}^{2}\right]+E\left[\{\frac{\partial l(\theta_{0},W_{i})}{\partial h_{\epsilon}}[h_{\epsilon1}-h_{\epsilon2}]\}^{2}\right]+E\left[\{\frac{\partial l(\theta_{0},W_{i})}{\partial h_{\nu}}[h_{\nu1}-h_{\nu2}]\}^{2}\right]\right\} \nonumber \\
 & \leq\tilde{B}||\theta_{1}-\theta_{2}||_{2}^{2}\label{eq:l2-fisher}
\end{align}
for some $B,\tilde{B}>0$ under Assumptions \ref{as:data}, \ref{as:para_sp},
and \ref{as:copula_derivative}. From equation \eqref{eq:l2-fisher},
it is straightforward to see that the convergence rate of the sieve
ML estimator with respect to the Fisher norm $||\cdot||$ is at least
as fast as the convergence rate with respect to the $L^{2}$-norm. 

\subsubsection{Directional derivatives of the log-likelihood function }

Let $r_{10}=F_{\epsilon0}(x^{'}\beta_{0}+\delta_{10})$, $r_{00}=F_{\epsilon0}(x^{'}\beta_{0})$,
and $s_{0}=F_{\nu0}(x^{'}\alpha_{0}+z^{'}\gamma_{0})$. For given
$v=(v_{\psi}^{'},v_{\epsilon},v_{\nu})^{'}\in\mathbb{V}$, we have
\begin{align}
\frac{\partial l(\theta_{0},w)}{\partial\psi^{'}}v_{\psi} & =\sum_{\tilde{y},\tilde{d}\in\{0,1\}}(\mathbf{1}_{\tilde{y},\tilde{d}}\cdot\frac{1}{p_{\tilde{y}\tilde{d},xz}(\theta_{0})}\cdot\frac{\partial p_{\tilde{y}\tilde{d},xz}(\theta_{0})}{\partial\psi^{'}})v_{\psi},\label{eq:direc_derivative_1}
\end{align}
\begin{align}
\frac{\partial l(\theta_{0},w)}{\partial h_{\epsilon}}[v_{\epsilon}] & =\mathbf{1}_{11}(y,d)\times\left[\frac{1}{p_{11,xz}(\theta_{0})}C_{1}(r_{10},s_{0};\rho_{0})\int_{0}^{G(x^{'}\beta_{0}+\delta_{10})}v_{\epsilon}(t)dt\right]\nonumber \\
 & +\mathbf{1}_{10}(y,d)\times\left[\frac{1}{p_{10,xz}(\theta_{0})}\left\{ \left(1-C_{1}(r_{00},s_{0};\rho_{0})\right)\int_{0}^{G(x^{'}\beta_{0})}v_{\epsilon}(t)dt\right\} \right]\nonumber \\
 & +\mathbf{1}_{01}(y,d)\times\left[\frac{1}{p_{01,xz}(\theta_{0})}\left\{ -C_{1}(r_{10},s_{0};\rho_{0})\int_{0}^{G(x^{'}\beta_{0}+\delta_{10})}v_{\epsilon}(t)dt\right\} \right]\nonumber \\
 & +\mathbf{1}_{00}(y,d)\times\left[\frac{1}{p_{00,xz}(\theta_{0})}\left\{ \left(1-C_{1}(r_{00},s_{0};\rho_{0})\right)\int_{0}^{G(x^{'}\beta_{0})}v_{\epsilon}(t)dt\right\} \right],\label{eq:direc_derivative_2}
\end{align}
and 
\begin{align}
\frac{\partial l(\theta{}_{0},w)}{\partial h_{\nu}}[v_{\nu}] & =\left\{ \mathbf{1}_{11}(y,d)\times\frac{1}{p_{11,xz}(\theta_{0})}C_{2}(r_{10},s_{0};\rho_{0})+\mathbf{1}_{10}(y,d)\times\frac{1}{p_{10,xz}(\theta_{0})}(-C_{2}(r_{00},s_{0};\rho_{0}))\right.\nonumber \\
 & \left.+\mathbf{1}_{01}(y,d)\times\frac{1}{p_{01,xz}(\theta_{0})}(1-C_{2}(r_{10},s_{0};\rho_{0}))+\mathbf{1}_{00}(y,d)\times\frac{1}{p_{00,xz}(\theta_{0})}(1-C_{2}(r_{00},s_{0};\rho_{0}))\right\} \nonumber \\
 & \times\int_{0}^{G(x^{'}\alpha_{0}+z^{'}\gamma_{0})}v_{\nu}(t)dt.\label{eq:direc_derivative_3}
\end{align}

\subsubsection{Directional derivative of the ATE}

Let $v=(v_{\psi}^{'},v_{\epsilon},v_{\nu})^{'}\in\mathbb{V}$. Then,
\begin{equation}
\frac{\partial ATE(\theta_{0};x)}{\partial\theta^{'}}[v]=\left\{ f_{\epsilon0}(x^{'}\beta_{0}+\delta_{10})(x^{'}v_{\beta}+v_{\delta})-f_{\epsilon0}(x^{'}\beta_{0})x^{'}v_{\beta}\right\} +\int_{G(x^{'}\beta_{0})}^{G(x^{'}\beta_{0}+\delta_{10})}v_{\epsilon}(t)dt,\label{eq:diff_ATE}
\end{equation}
 where $f_{\epsilon0}(x)=h_{\epsilon0}(G(x))g(x)$.

\subsection{Proof of Theorem \ref{thm:consistency}}

Define $Q_{0}(\theta)\equiv E[l(\theta,W_{i})]$. The following proposition
is a modification of Theorem 3.1 in \citet{chen2007large} and establishes
the consistency of sieve M-estimator.\footnote{See also Remark 3.3 in \citet{chen2007large}.} 

\begin{proposition}\label{prop:consistency} Let $\hat{\theta}_{n}$
be the sieve extremum estimator defined in \eqref{eq:sieve_ML_est}.
Suppose that the following conditions hold : 

(i) $Q_{0}(\theta)$ is uniquely maximized at $\theta_{0}$ in $\Theta$
and $Q_{0}(\theta_{0})>-\infty$;

(ii) $\Theta$ is compact under $d_{c}(\cdot,\cdot)$, and $Q_{0}(\theta)$
is upper semicontinuous on $\Theta$ under $d_{c}(\cdot,\cdot)$; 

(iii) The sieve spaces, $\Theta_{n}$, are compact under $d_{c}(\cdot,\cdot)$
; 

(iv) $\Theta_{k}\subseteq\Theta_{k+1}\subseteq\Theta$ for all $k\geq1$,
and there exists a sequence $\pi_{k}\theta_{0}\in\Theta_{k}$ such
that $d_{c}(\theta_{0},\pi_{k}\theta_{0})\rightarrow0$ as $k\rightarrow\infty$
;

(v) For all $k\geq1$, $p\lim_{n\rightarrow\infty}\sup_{\theta\in\Theta_{k}}|Q_{n}(\theta)-Q_{0}(\theta)|=0$. 

Then, $d_{c}(\hat{\theta}_{n},\theta_{0})=o_{p}(1)$. 

\end{proposition}

We show that the conditions in Theorem \ref{thm:consistency} imply
those in this proposition to prove consistency of the sieve estimator.
We first need to verify that (i) the true parameter $\theta_{0}$
is the unique maximizer of $Q_{0}(\cdot)$ over $\Theta$ and that
(ii) the sample log-likelihood function $Q_{n}(\cdot)$ uniformly
converges to $Q_{0}(\cdot)$ over the sieve space in probability to
establish the consistency of the sieve ML estimator. The following
lemma shows that if the model with unknown marginal distributions
are identified and some additional conditions are satisfied, then
the true parameter $\theta_{0}$ is the unique maximizer of $Q_{0}(\cdot)$
over $\Theta$. 

\begin{lemma}\label{lemma:unique}Suppose that Assumptions \ref{as:independence}\textendash \ref{as_copula},
\ref{as:add_support}, \ref{as:compact_para} and \ref{as:WD} are
satisfied. Then the condition (i) in Proposition \ref{prop:consistency}
is satisfied. \end{lemma}

\begin{proof} By Theorem \ref{thm_full_unknown}, the model parameter
is identified. Under Assumption \ref{as:WD}, we can see that for
any $\theta\in\Theta$, $|Q_{0}(\theta)|\leq E|l(\theta,W_{i})|\leq\sum_{y,d\in\{0,1\}}E|\log(p_{yd,XZ}(\theta))|<\infty$,
and thus the function $Q_{0}(\theta)$ is well-defined on $\Theta$
and $Q_{0}(\theta)>-\infty$ for all $\theta\in\Theta$; hence $Q_{0}(\theta_{0})>-\infty$.
Since the model is identified, it implies that for $\theta\neq\theta_{0}$,
there exists a set $E\subset\text{supp}(X,Z)$ such that $\int_{E}dP_{XZ}>0$
and for some $y,d\in\{0,1\}$, $\frac{p_{yd,xz}(\theta)}{p_{yd,xz}(\theta_{0})}\neq1$
on $E$, where $P_{XZ}$ is the distribution function of $(X,Z)$.
Thus, we have 
\begin{align*}
Q_{0}(\theta)-Q_{0}(\theta_{0}) & =\int\sum_{y,d\in\{0,1\}}p_{yd,xz}(\theta_{0})\log\left(\frac{p_{yd,xz}(\theta)}{p_{yd,xz}(\theta_{0})}\right)dP_{XZ}<\log\left(\int_{E}\sum_{y,d\in\{0,1\}}p_{yd,xz}(\theta)dP_{XZ}\right)\leq0,
\end{align*}
 where the strict inequality holds by the fact that $p_{yd,xz}(\theta)\neq p_{yd,xz}(\theta_{0})$
on $E$ and Jensen's inequality. Hence, $\theta_{0}$ is the unique
maximizer of $Q_{0}(\cdot)$. \end{proof}

For any $\omega>0$, let $N(\omega,\Theta_{n},d_{c})$ be the covering
numbers without bracketing of $\Theta_{n}$ with respect to the pseudo-metric
$d_{c}$. We now establish the uniform convergence of $Q_{n}(\cdot)$
to $Q_{0}$ over the sieve space. 

\begin{lemma} \label{lemma:uniform}Suppose that Assumptions \ref{as:independence}\textendash \ref{as_copula},
\ref{as:add_support} are satisfied. If Assumptions \ref{as:compact_para}
through \ref{as:copula_derivative} hold, then $\sup_{\theta\in\Theta_{n}}|Q_{n}(\theta)-Q_{0}(\theta)|\overset{p}{\rightarrow}0$
for all $n\geq1$. \end{lemma} 

\begin{proof} We verify Condition 3.5M in \citet{chen2007large}.
Let $B$ stand for a generic constant and it can be different in each
place. By Assumptions \ref{as:WD} and \ref{as:data}, the first condition
in Condition 3.5M is satisfied. Let $n\geq1$ be a natural number
and $\theta,\tilde{\theta}\in\Theta_{n}$. Define $R_{1}(\theta)=F_{\epsilon}(X^{'}\beta+\delta_{1})$,
$R_{0}(\theta)=F_{\epsilon}(X^{'}\beta)$, and $S(\theta)=F_{\nu}(X^{'}\alpha+Z^{'}\gamma)$.
Similarly, we define $R_{1}(\tilde{\theta})=\tilde{F}_{\epsilon}(X^{'}\tilde{\beta}+\tilde{\delta}_{1})$,
$R_{0}(\tilde{\theta})=\tilde{F}_{\epsilon}(X^{'}\tilde{\beta})$,
and $S(\tilde{\theta})=\tilde{F}_{\nu}(X^{'}\tilde{\alpha}+Z^{'}\tilde{\gamma})$.
For the simplicity of the notations, we write $R_{j}=R_{j}(\theta)$,
$\tilde{R}_{j}=R_{j}(\tilde{\theta})$, $S=S(\theta)$, and $\tilde{S}=S(\tilde{\theta})$
for all $j=0,1$. Observe that 
\begin{align*}
|p_{11,XZ}(\theta)-p_{11,XZ}(\tilde{\theta})| & =|C(R_{1},S;\rho)-C(\tilde{R}_{1},\tilde{S};\tilde{\rho})|\\
 & \leq|C(R_{1},S;\rho)-C(\tilde{R}_{1},\tilde{S};\rho)|+|C(\tilde{R}_{1},\tilde{S};\rho)-C(\tilde{R}_{1},\tilde{S};\tilde{\rho})|\\
 & \leq|R_{1}-\tilde{R}_{1}|+|S-\tilde{S}|+|C_{\rho}(\tilde{R}_{1},\tilde{S};\hat{\rho})||\rho-\tilde{\rho}|\\
 & \leq|R_{1}-\tilde{R}_{1}|+|S-\tilde{S}|+B|\rho-\tilde{\rho}|,
\end{align*}
 where $C_{\rho}(\cdot,\cdot;\cdot)$ is the partial derivative of
$C(\cdot,\cdot;\cdot)$ with respect to $\rho$ and $\hat{\rho}$
is between $\rho$ and $\tilde{\rho}$ and $B<\infty$. Note that
the last inequality holds due to a generic property of copulas (see,
e.g. Theorem 2.2.4 in \citet{nelsen1999introduction}) and the mean
value theorem. We also have 
\begin{align}
|R_{1}-\tilde{R}_{1}| & =\left|F_{\epsilon}(X^{'}\beta+\delta_{1})-\tilde{F}_{\epsilon}(X^{'}\tilde{\beta}+\tilde{\delta}_{1})\right|\nonumber \\
 & \leq\left|F_{\epsilon}(X^{'}\beta+\delta_{1})-F_{\epsilon}(X^{'}\tilde{\beta}+\tilde{\delta}_{1})\right|+\left|F_{\epsilon}(X^{'}\tilde{\beta}+\tilde{\delta}_{1})-\tilde{F}_{\epsilon}(X^{'}\tilde{\beta}+\tilde{\delta}_{1})\right|\nonumber \\
 & \leq\left|f_{\epsilon}(X^{'}\hat{\beta}+\hat{\delta}_{1})\right|\cdot\left|X^{'}(\beta-\tilde{\beta})+(\delta_{1}-\tilde{\delta}_{1})\right|+\int_{0}^{G(X^{'}\tilde{\beta}+\tilde{\delta}_{1})}\left|h_{\epsilon}(t)-\tilde{h}_{\epsilon}(t)\right|dt\nonumber \\
 & \leq\sup_{x\in\mathbb{R}}|h_{\epsilon}(G(x))g(x)|\times||(X^{'},1)^{'}||_{E}\cdot||\psi-\tilde{\psi}||_{E}+||h_{\epsilon}-\tilde{h}_{\epsilon}||_{\infty}\nonumber \\
 & \leq B\times||(X^{'},1)^{'}||_{E}\times||(\beta^{'},\delta_{1})^{'}-(\tilde{\beta}^{'},\tilde{\delta}_{1})^{'}||_{E}+||h_{\epsilon}-\tilde{h}_{\epsilon}||_{\infty},\label{eq:r1-domination}
\end{align}
 for some constant $B<\infty$. Similarly, we can show that 
\begin{equation}
|R_{0}-\tilde{R}_{0}|\leq B\times||X||_{E}\times||\beta-\tilde{\beta}||_{E}+||h_{\epsilon}-\tilde{h}_{\epsilon}||_{\infty}\label{eq:r0-domination}
\end{equation}
 and 
\begin{equation}
|S-\tilde{S}|\leq B\times||(X^{'},Z^{'})^{'}||_{E}\times||(\alpha^{'},\gamma^{'})^{'}-(\tilde{\alpha}^{'},\tilde{\gamma}^{'})^{'}||_{E}+||h_{\nu}-\tilde{h}_{\nu}||_{\infty}.\label{eq:s-domination}
\end{equation}
Note that, for any comparable subvectors $\psi_{s}$ and $\tilde{\psi}_{s}$
of $\psi$ and $\tilde{\psi}$, respectively, we have $||\psi_{s}-\tilde{\psi}_{s}||_{E}\leq||\psi-\tilde{\psi}||_{E}$
and that, for any subvector $W_{s}$ of $W$, we have $||W_{S}||_{E}\leq||W||_{E}$
a.s. Thus we have 
\begin{align*}
|p_{11,XZ}(\theta)-p_{11,XZ}(\tilde{\theta})| & \leq B||(X^{'},1)^{'}||_{E}\cdot||\psi-\tilde{\psi}||_{E}+||h_{\epsilon}-\tilde{h}_{\epsilon}||_{\infty}\\
 & \leq B||(X^{'},1)^{'}||_{E}d_{c}(\theta,\tilde{\theta}).
\end{align*}
Consequently, it follows that
\begin{align*}
|p_{10,XZ}(\theta)-p_{10,XZ}(\tilde{\theta})| & \leq|R_{0}-\tilde{R}_{0}|+|C(R_{0},S;\rho)-C(\tilde{R}_{0},\tilde{S};\tilde{\rho})|\\
 & \leq2|R_{0}-\tilde{R}_{0}|+|S-\tilde{S}|+B|\rho-\tilde{\rho}|\\
 & \leq B\{||X||_{E}||\beta-\tilde{\beta}||_{E}+||(X^{'},Z^{'})^{'}||_{E}||(\alpha^{'},\gamma^{'})^{'}-(\tilde{\alpha}^{'},\tilde{\gamma}^{'})^{'}||_{E}\\
 & +||h_{\epsilon}-\tilde{h}_{\epsilon}||_{\infty}+||h_{\nu}-\tilde{h}_{\nu}||_{\infty}+|\rho-\tilde{\rho}|\}\\
 & \leq B\cdot||(X^{'},Z^{'},1)^{'}||_{E}d_{c}(\theta,\tilde{\theta}),\\
|p_{01,XZ}(\theta)-p_{01,XZ}(\tilde{\theta})| & \leq2|S-\tilde{S}|+|R_{1}-\tilde{R}_{1}|+B|\rho-\tilde{\rho}|\\
 & \leq B||(X^{'},Z^{'},1)^{'}||_{E}d_{c}(\theta,\tilde{\theta}),\\
|p_{00,XZ}(\theta)-p_{00,XZ}(\tilde{\theta})| & \leq|p_{11,XZ}(\theta)-p_{11,XZ}(\tilde{\theta})|+|p_{10,XZ}(\theta)-p_{10,XZ}(\tilde{\theta})|+|p_{01,XZ}(\theta)-p_{01,XZ}(\tilde{\theta})|\\
 & \leq B||(X^{'},Z^{'},1)^{'}||_{E}d_{c}(\theta,\tilde{\theta}).
\end{align*}
 In all, we have 
\begin{align}
|l(\theta,W_{i})-l(\tilde{\theta},W_{i})| & \leq{\textstyle \sum\limits _{y,d=0,1}}\mathbf{1}_{yd}(Y_{i},D_{i})\cdot\left|\log p_{yd}(X_{i},Z_{i};\theta)-\log p_{yd}(X_{i},Z_{i};\tilde{\theta})\right|\nonumber \\
 & \leq\frac{1}{\underline{p}(X_{i},Z_{i})}\sum_{y,d=0,1}\mathbf{1}_{yd}(Y_{i},D_{i})\left|p_{yd}(X_{i},Z_{i};\theta)-p_{yd}(X_{i},Z_{i};\tilde{\theta})\right|\nonumber \\
 & \leq\frac{B}{\underline{p}(X_{i},Z_{i})}||(X_{i}^{'},Z_{i}^{'},1)^{'}||_{E}d_{c}(\theta,\tilde{\theta})\nonumber \\
 & \equiv U(W_{i})d_{c}(\theta,\tilde{\theta}),\label{eq:loglikelihood-supnorm}
\end{align}
 where $E[U(W_{i})^{2}]<\infty$ by Assumptions \ref{as:WD} and \ref{as:data}.
This results in 
\begin{equation}
\sup_{\theta,\tilde{\theta}\in\Theta_{n},d_{c}(\theta,\tilde{\theta})\leq\epsilon_{0}}\left|l(\theta,W_{i})-l(\tilde{\theta},W_{i})\right|\leq U(W_{i})\epsilon_{0}\label{eq:domination}
\end{equation}
 and thus the second condition in Condition 3.5M is satisfied with
$s=1$. 

For the last condition in Condition 3.5M, note that for any $\omega>0$,
we have 
\begin{align*}
N(\omega,\Theta_{n},d_{c}) & \leq N(\frac{\omega}{2},\Psi,||\cdot||_{E})\cdot N(\frac{\omega}{4},\mathcal{H}_{\epsilon n},||\cdot||_{\infty})\cdot N(\frac{\omega}{4},\mathcal{H}_{\nu n},||\cdot||_{\infty}).
\end{align*}
 By Lemma 2.5 in \citet{geer2000empirical}, we have $\log N\left(\frac{\omega}{4},\mathcal{H}_{\epsilon n},||\cdot||_{\infty}\right)\leq k_{n}\log\left(1+\frac{32R}{\omega}\right)$
under Assumption \ref{as:sieve}-(i); and hence 
\begin{align*}
\log N\left(\omega,\Theta_{n},d_{c}\right) & \leq const.\times k_{n}\times\log\left(1+\frac{32R}{\omega}\right)=o(n)
\end{align*}
 if $k_{n}/n\rightarrow0$. Since the condition $k_{n}/n=o(1)$ is
imposed by Assumption \ref{as:sieve}-(i), the last condition in Condition
3.5M is also satisfied. In all, we have the uniform convergence of
$Q_{n}$ to $Q_{0}$ over $\Theta_{k}$. \end{proof}

To finish proving Theorem \ref{thm:consistency}, we verify the conditions
in Proposition \ref{prop:consistency}. By Lemmas \ref{lemma:unique}
and \ref{lemma:uniform}, the conditions (i) and (v) in Proposition
\ref{prop:consistency} are satisfied. Using \eqref{eq:loglikelihood-supnorm}
and Jensen's inequality, we can see that, for any $\theta,\tilde{\theta}\in\Theta$,
\begin{align*}
|Q_{0}(\theta)-Q_{0}(\tilde{\theta})| & \leq E|l(\theta,W_{i})-l(\tilde{\theta},W_{i})|\leq E[U(W_{i})]d_{c}(\theta,\tilde{\theta})=B\cdot d_{c}(\theta,\tilde{\theta})
\end{align*}
 for some $B<\infty$. Thus, $Q_{0}(\cdot)$ is continuous with respect
to $d_{c}$. Note that since the parameter space of the finite-dimensional
parameter $\psi$, $\Psi$, is assumed to be compact in Assumption
\ref{as:compact_para}, the original parameter space $\Theta$ is
compact under the $d_{c}$, by Theorems 1 and 2 in \citet{freyberger2015compactness},
and thus the conditions (ii) and (iii) are satisfied with the specified
parameter space and the norm. Since the condition (iv) is directly
imposed, we have $d(\hat{\theta}_{n},\theta_{0})=o_{p}(1)$ by Proposition
\ref{prop:consistency}. 

\subsection{Proof of Theorem \ref{thm:convergence_rate}}

To establish the convergence rate with respect to the norm $||\cdot||_{2}$,
we consider the following assumption: 

\begin{assumption}\label{as:KL-L2} Let $K(\theta_{0},\theta)\equiv E[l(\theta_{0},W_{i})-l(\theta,W_{i})]$.
Then, there exist $B_{1},B_{2}>0$ such that 
\[
B_{1}K(\theta_{0},\theta)\leq||\theta-\theta_{0}||_{2}^{2}\leq B_{2}K(\theta_{0},\theta)
\]
 for all $\theta\in\Theta_{n}$ with $d_{c}(\theta,\theta_{0})=o(1)$.
\end{assumption}

Assumption \ref{as:KL-L2} implies that the $L^{2}$-norm $||\cdot||_{2}$
and the square-root of the KL divergence are equivalent.

We derive the convergence rate of the sieve M-estimator with respect
to the norm $||\cdot||_{2}$ by checking the conditions in Theorem
3.2 in \citet{chen2007large}. Since $\{W_{i}\}_{i=1}^{n}$ is assumed
to be i.i.d by Assumption \ref{as:data}, Condition 3.6 in \citet{chen2007large}
is satisfied. For Condition 3.7 in \citet{chen2007large}, we note
that for a small $\epsilon_{1}>0$ and for any $\theta\in\Theta_{n}$
such that $||\theta-\theta_{0}||_{2}\leq\epsilon_{1}$, we have 
\begin{align*}
Var\left(l(\theta,W_{i})-l(\theta_{0},W_{i})\right) & \leq E\left[l(\theta,W_{i})-l(\theta_{0},W_{i})\right]^{2}\\
 & \leq E\left[\frac{1}{\underline{p}(X_{i},Z_{i})^{2}}\sum_{y,d=0,1}\mathbf{1}_{yd}(Y_{i},D_{i})|p_{yd}(X_{i},Z_{i};\theta)-p_{yd}(X_{i},Z_{i};\theta_{0})|^{2}\right]\\
 & \leq E\left[\frac{1}{\underline{p}(X_{i},Z_{i})^{2}}\sum_{y,d\in\{0,1\}}|p_{yd}(X_{i},Z_{i};\theta)-p_{yd}(X_{i},Z_{i};\theta_{0})|^{2}\right].
\end{align*}
 By the same logic in \eqref{eq:loglikelihood-supnorm}, we have 
\[
Var\left(l(\theta,W_{i})-l(\theta_{0},W_{i})\right)\leq E\left[U(W_{i})^{2}\right]d_{c}(\theta,\theta_{0})^{2}.
\]
Note that 
\begin{align*}
d_{c}(\theta,\theta_{0})^{2} & =(||\psi-\psi_{0}||_{E}+||h_{\epsilon}-h_{\epsilon0}||_{\infty}+||h_{\nu}-h_{\nu0}||_{\infty})^{2}\\
 & \leq4(||\psi-\psi_{0}||_{E}^{2}+||h_{\epsilon}-h_{\epsilon0}||_{\infty}^{2}+||h_{\nu}-h_{\nu0}||_{\infty}^{2}).
\end{align*}
 By Lemma 2 in \citet{chen1998sieve}, we have 
\begin{equation}
||h_{j}-h_{j0}||_{\infty}^{2}\leq||h_{j}-h_{j0}||_{2}^{\frac{4p}{2p+1}}\label{eq:gabushin}
\end{equation}
 for all $j\in\{\epsilon,\nu\}$. Since $\frac{4p}{2p+1}>1$ under
Assumption \ref{as:para_sp}, we can show that 
\[
\sup_{\{\theta\in\Theta_{n}:||\theta-\theta_{0}||_{2}\leq\epsilon_{1}\}}Var\left(l(\theta,W_{i})-l(\theta_{0},W_{i})\right)\leq B_{1}\epsilon_{1}^{2}
\]
 with $\epsilon_{1}\leq1$ and some constant $B_{1}$, and thus Condition
3.7 in \citet{chen2007large} is satisfied. 

We recall equation \eqref{eq:loglikelihood-supnorm} to verify Condition
3.8 in \citet{chen2007large}. Let $\epsilon_{2}>0$ be given and
consider 
\begin{align}
|l(\theta,W_{i})-l(\theta_{0},W_{i})| & \leq U(W_{i})d_{c}(\theta,\theta_{0})\nonumber \\
 & =U(W_{i})\left\{ ||\psi-\psi_{0}||_{E}+||h_{\epsilon}-h_{\epsilon0}||_{\infty}+||h_{\nu}-h_{\nu0}||_{\infty}\right\} \nonumber \\
 & \leq U(W_{i})\left\{ ||\psi-\psi_{0}||_{E}+||h_{\epsilon}-h_{\epsilon0}||_{2}^{\frac{2p}{2p+1}}+||h_{\nu}-h_{\nu0}||_{2}^{\frac{2p}{2p+1}}\right\} \nonumber \\
 & \leq U(W_{i})\left\{ ||\psi-\psi_{0}||_{E}^{\frac{2p+1}{2p}}+||h_{\epsilon}-h_{\epsilon0}||_{2}+||h_{\nu}-h_{\nu0}||_{2}\right\} ^{\frac{2p}{2p+1}}\nonumber \\
 & \leq U(W_{i})\left\{ ||\psi-\psi_{0}||_{E}\times(\sup_{\psi\in\Psi}||\psi||_{E}+||\psi_{0}||_{E})^{\frac{1}{2p}}+||h_{\epsilon}-h_{\epsilon0}||_{2}+||h_{\nu}-h_{\nu0}||_{2}\right\} ^{\frac{2p}{2p+1}}\nonumber \\
 & \leq\tilde{U}(W_{i})\left\{ ||\psi-\psi_{0}||_{E}+||h_{\epsilon}-h_{\epsilon0}||_{2}+||h_{\nu}-h_{\nu0}||_{2}\right\} ^{\frac{2p}{2p+1}},\label{eq:cond3.8}
\end{align}
 where $\tilde{U}(W_{i})=\max\{1,(\sup_{\psi\in\Psi}||\psi||_{E}+||\psi_{0}||_{E})^{\frac{1}{2p}}\}\times U(W_{i})$.
Since the parameter space for $\psi$, $\Psi$, is compact under Assumption
\ref{as:compact_para}, $E[\tilde{U}(W_{i})^{2}]<\infty$. Thus, we
have 
\[
\sup_{\{\theta\in\Theta_{n}:||\theta-\theta_{0}||_{2}\leq\epsilon_{2}\}}\left|l(\theta,W_{i})-l(\theta_{0},W_{i})\right|\leq\epsilon_{2}^{\frac{2p}{2p+1}}\tilde{U}(W_{i})
\]
 with $E[\tilde{U}_{i}(W_{i})^{2}]<\infty$ and this implies that,
under Assumption \ref{as:para_sp}, Condition 3.8 in \citet{chen2007large}
is satisfied with $s=\frac{2p}{2p+1}\in(0,2)$ and $\gamma=2$. 

Let $\mathcal{L}_{n}\equiv\{l(\theta_{0},W_{i})-l(\theta,W_{i}):\theta\in\Theta_{n},||\theta-\theta_{0}||_{2}\leq\epsilon_{2}\}$.
For given $\omega>0$, let $N_{[]}(\omega,\mathcal{L}_{n},||\cdot||_{L^{2}})$
be the covering number with bracketing of $\mathcal{L}_{n}$ with
respect to the norm $||\cdot||_{L^{2}}$. We now need to calculate
$\kappa_{n}$ which is defined as 
\[
\kappa_{n}\equiv\inf\left\{ \kappa\in(0,1):\frac{1}{\sqrt{n}\kappa^{2}}\int_{b\kappa^{2}}^{\kappa}\sqrt{H_{[]}(\omega,\mathcal{L}_{n},||\cdot||_{L^{2}})}d\omega\leq const.\right\} ,
\]
 where, for $f\in\mathcal{L}_{n}$, $||f(\theta,W_{i})||_{L^{2}}^{2}\equiv E[f(\theta,W_{i})^{2}]$
is the $L^{2}$-norm on $\mathcal{L}_{n}$ and $H_{[]}(\omega,\mathcal{L}_{n},||\cdot||_{L^{2}})$
is the $L_{2}$-metric entropy with bracketing of the class $\mathcal{L}_{n}$
(see \citet{vvw96} or \citet{geer2000empirical} for the definition
of $L_{2}$-metric entropy with bracketing). Let $B_{0}=E[U(W_{i})^{2}]$,
where $U(W_{i})$ is the same to the one in \eqref{eq:loglikelihood-supnorm}.
By Theorem 2.7.11 in \citet{vvw96} and equation \eqref{eq:loglikelihood-supnorm},
we can show that 
\begin{align*}
N_{[]}\left(\omega,\mathcal{L}_{n},||\cdot||_{L^{2}}\right) & \leq N\left(\frac{\omega}{2B_{0}},\Theta_{n},d_{c}\right)\\
 & \leq N\left(\frac{\omega}{4B_{0}},\Psi,||\cdot||_{E}\right)\cdot N\left(\frac{\omega}{8B_{0}},\mathcal{H}_{\epsilon n},||\cdot||_{\infty}\right)\cdot N\left(\frac{\omega}{8B_{0}},\mathcal{H}_{\nu n},||\cdot||_{\infty}\right),
\end{align*}
 and this leads to 
\begin{align*}
H_{[]}\left(\omega,\mathcal{L}_{n},||\cdot||_{L^{2}}\right) & =\log\left(N_{[]}\left(\omega,\mathcal{L}_{n},||\cdot||_{L^{2}}\right)\right)\leq const.\times k_{n}\times\log(1+\frac{64B_{0}R}{\omega}).
\end{align*}
 In all, $\kappa_{n}$ solves 
\begin{align*}
\frac{1}{\sqrt{n}\kappa_{n}^{2}}\int_{b\kappa_{n}^{2}}^{\kappa_{n}}\sqrt{H_{[]}(\omega,\mathcal{L}_{n},||\cdot||_{L^{2}})}d\omega & \leq\frac{const.}{\sqrt{n}\kappa_{n}^{2}}\int_{b\kappa_{n}^{2}}^{\kappa_{n}}\sqrt{k_{n}\cdot\log(1+\frac{64B_{0}R}{\omega})}d\omega\\
 & \leq\frac{const.}{\sqrt{n}\kappa_{n}^{2}}\sqrt{k_{n}}\int_{b\kappa_{n}^{2}}^{\kappa_{n}}\sqrt{\frac{1}{\omega}}d\omega\leq const.\times\frac{1}{\sqrt{n}\kappa_{n}^{2}}\sqrt{k_{n}}\kappa_{n}\leq const.,
\end{align*}
and thus $\kappa_{n}\propto\sqrt{\frac{k_{n}}{n}}$. 

Lastly, since $||\theta_{0}-\pi_{n}\theta_{0}||_{2}\leq||\theta_{0}-\pi_{n}\theta_{0}||_{c}=O(k_{n}^{-p})$
by \citet{lorentz1966approximation}, we have 
\[
||\hat{\theta}_{n}-\theta_{0}||_{2}=O_{p}\left(\max\left\{ \sqrt{\frac{k_{n}}{n}},k_{n}^{-p}\right\} \right)
\]
 by Theorem 3.2 in \citet{chen2007large}. By choosing $k_{n}\propto n^{\frac{1}{2p+1}}$,
we have 
\[
||\hat{\theta}_{n}-\theta_{0}||_{2}=O_{p}\left(n^{-\frac{p}{2p+1}}\right).
\]

\subsection{Proof of Proposition \ref{prop:asym general}}

We first provide some technical assumptions for the asymptotic normality.
Let $\mu_{n}(g)=\frac{1}{n}\sum_{i=1}^{n}\{g(W_{i})-E[g(W_{i})]\}$
be the empirical process indexed by $g$. Let the convergence rate
of the sieve estimator be $\delta_{n}$ (i.e., $||\hat{\theta}_{n}-\theta_{0}||=O_{p}(\delta_{n})$). 

\begin{assumption}\label{as:secondorder} There exist $\xi_{1}>0$
and $\xi_{2}>0$ with $2\xi_{1}+\xi_{2}<1$ and a constant $K$, such
that $(\delta_{n})^{3-(2\xi_{1}+\xi_{2})}=o(n^{-1})$. In addition,
the following hold for all $\tilde{\theta}\in\Theta_{n}$ with $||\tilde{\theta}-\theta_{0}||\leq\delta_{n}$,
and all $v\in\mathbb{V}$ with $||v||\leq\delta_{n}$: 

(i) $\left|E\left[\frac{\partial^{2}l(\tilde{\theta},W)}{\partial\psi\partial\psi^{'}}-\frac{\partial^{2}l(\theta_{0},W)}{\partial\psi\partial\psi^{'}}\right]\right|<K\left\Vert \tilde{\theta}-\theta_{0}\right\Vert ^{1-\xi_{2}}$; 

(ii) $\left|E\left[\sum_{j\in\{\epsilon,\nu\}}\left\{ \frac{\partial^{2}l(\tilde{\theta},W)}{\partial\psi\partial h_{j}}[v_{j}]-\frac{\partial^{2}l(\theta_{0},W)}{\partial\psi\partial h_{j}}[v_{j}]\right\} \right]\right|\leq K\left\Vert v\right\Vert ^{1-\xi_{1}}\left\Vert \tilde{\theta}-\theta_{0}\right\Vert ^{1-\xi_{2}}$; 

(iii) $\left|E\left[\sum_{i,j\in\{\epsilon,\nu\}}\left\{ \frac{\partial^{2}l(\tilde{\theta},W)}{\partial h_{i}\partial h_{j}}[v,v]-\frac{\partial^{2}l(\theta_{0},W)}{\partial h_{i}\partial h_{j}}[v,v]\right\} \right]\right|\leq K||v||^{2(1-\xi_{1})}||\tilde{\theta}-\theta_{0}||^{1-\xi_{2}}$. 

\end{assumption}

\begin{assumption}\label{as:emp_process} The following hold:

(i) $\sup_{\theta\in\Theta_{n}:||\theta-\theta_{0}||=O(\delta_{n})}\mu_{n}\left(\frac{\partial l(\theta,W)}{\partial\psi^{'}}-\frac{\partial l(\theta_{0},W)}{\partial\psi^{'}}\right)=o_{p}\left(n^{-\frac{1}{2}}\right)$
; 

(ii) For all $j\in\{\epsilon,\nu\}$, $\sup_{\theta\in\Theta_{n}:||\theta-\theta_{0}||=O(\delta_{n})}\mu_{n}\left(\frac{\partial l(\theta,W)}{\partial h_{j}}[\pi_{n}v_{j}^{*}]-\frac{\partial l(\theta_{0},W)}{\partial h_{j}}[\pi_{n}v_{j}^{*}]\right)=o_{p}\left(n^{-\frac{1}{2}}\right)$. 

\end{assumption} 

Assumptions \ref{as:secondorder} and \ref{as:emp_process} are modifications
of Assumptions 5 and 6 in CFT06, which are needed to control for the
second-order expansion of the log-likelihood function $l(\theta,W)$.
Under Assumption \ref{as:copula_second}, these conditions require
that the unknown marginal density functions be sufficiently smooth.
For example, the sieve estimator needs to converge at a faster rate
than $1/(3-(2\xi_{1}+\xi_{2}))$ to satisfy $(\delta_{n})^{3-(2\xi_{1}+\xi_{2})}=o(n^{-1})$.
Usually, the convergence rate depends positively on the smoothness
parameter $p$ in Assumption \ref{as:para_sp} and thus the class
of models should be restricted to that in which the density functions
are sufficiently smooth. 

Note that since the sieve ML estimator $\hat{\theta}_{n}$ is consistent
with respect to the pseudo-metric $d_{c}$ by Theorem \ref{thm:consistency},
it is consistent with respect to the norm $||\cdot||_{2}$ and thus
with respect to the Fisher norm by equation \eqref{eq:l2-fisher}.
We also point out that $||\hat{\theta}_{n}-\theta_{0}||=O_{p}(n^{-\frac{p}{2p+1}})$
by equation \eqref{eq:l2-fisher} and Theorem \ref{thm:convergence_rate}
under the given set of Assumptions. We follow the proof of Theorem
1 in CFT06. Assumptions 1 and 2 in CFT06 are implied by Assumption
\ref{as:independence}-\ref{as_copula}, \ref{as:add_support}-\ref{as:WD},
and \ref{as:copula_second}. The first two parts in Assumption \ref{as:smooth_T}
correspond to Assumption 3 in CFT06. Since $p>1/2$ by Assumption
\ref{as:para_sp}, $||\hat{\theta}_{n}-\theta_{0}||=o_{p}(n^{-1/4})$
by Theorem \ref{thm:convergence_rate} and this implies that $||\hat{\theta}_{n}-\theta_{0}||\times||\pi_{n}v^{*}-v^{*}||=o(n^{-1/2})$
under Assumption \ref{as:smooth_Riesz}. In addition, since $w>1+\frac{1}{2p}$,
$\delta_{n}^{w}=o(n^{-1/2})$ by that $||\hat{\theta}_{n}-\theta_{0}||=O_{p}(n^{-\frac{p}{2p+1}})$.
Hence, Assumptions 3 and 4 in CFT06 are satisfied. 

Define $r[\theta,\theta_{0},W_{i}]\equiv l(\theta,W_{i})-l(\theta_{0},Z_{i})-\frac{\partial l(\theta_{0},W_{i})}{\partial\theta^{'}}[\theta-\theta_{0}]$
and $\xi_{0}=2\xi_{1}+\xi_{2}$. Let $\zeta_{n}$ be a positive sequence
with $\zeta_{n}=o(n^{-1/2})$ and $(\delta_{n})^{3-(2\xi_{1}+\xi_{2})}=\zeta_{n}o(n^{-1/2})$.
Then we have 
\begin{align}
0 & \leq\frac{1}{n}\sum_{i=1}^{n}l(\hat{\theta}_{n},W_{i})-l(\hat{\theta}_{n}\pm\zeta_{n}\pi_{n}v^{*},W_{i})\leq\mp\zeta_{n}\frac{1}{n}\sum_{i=1}^{n}\frac{\partial l(\theta_{0},W_{i})}{\partial\theta^{'}}[\pi_{n}v^{*}]\nonumber \\
 & +\mu_{n}(r[\hat{\theta}_{n},\theta_{0},W_{i}]-r[\hat{\theta}_{n}\pm\zeta_{n}\pi_{n}v^{*},\theta_{0},W_{i}])+E[r[\hat{\theta}_{n},\theta_{0},W_{i}]-r[\hat{\theta}_{n}\pm\zeta_{n}\pi_{n}v^{*},\theta_{0},W_{i}]].\label{eq:AN0}
\end{align}

We first note that, by Assumption \ref{as:smooth_Riesz}, 
\begin{align}
E\left[\frac{1}{n}\sum_{i=1}^{n}\frac{\partial l(\theta_{0},W_{i})}{\partial\theta^{'}}[\pi_{n}v^{*}-v^{*}]\right]^{2} & \leq\frac{1}{n}E\left[\left\{ \frac{\partial l(\theta_{0},W_{i})}{\partial\theta^{'}}[\pi_{n}v^{*}-v^{*}]\right\} ^{2}\right]\nonumber \\
 & =\frac{1}{n}||\pi_{n}v^{*}-v^{*}||^{2}=o(n^{-1}),\label{eq:AN1}
\end{align}
 and hence $\frac{1}{n}\sum_{i=1}^{n}\frac{\partial l(\theta_{0},W_{i})}{\partial\theta^{'}}[\pi_{n}v^{*}-v^{*}]=o_{p}(n^{-1/2})$. 

Observe that, by the mean value theorem, 
\begin{align}
E\left[r[\theta,\theta_{0},W_{i}]\right] & =E\left[l(\theta,W_{i})-l(\theta_{0},W_{i})-\frac{\partial l(\theta_{0},W_{i})}{\partial\theta^{'}}[\theta-\theta_{0}]\right]\nonumber \\
 & =E\left[\frac{1}{2}\frac{\partial^{2}l(\theta_{0},W_{i})}{\partial\theta\partial\theta^{'}}[\theta-\theta_{0},\theta-\theta_{0}]\right]\nonumber \\
 & +\frac{1}{2}E\left[\frac{\partial^{2}l(\tilde{\theta},W_{i})}{\partial\theta\partial\theta^{'}}[\theta-\theta_{0},\theta-\theta_{0}]-\frac{\partial^{2}l(\theta_{0},W_{i})}{\partial\theta\partial\theta^{'}}[\theta-\theta_{0},\theta-\theta_{0}]\right],\label{eq:AN1-1}
\end{align}
 where $\theta,\tilde{\theta}\in\Theta_{n}$ and $\tilde{\theta}$
is between $\theta$ and $\theta_{0}$. In addition, for any $v=(v_{\psi}^{'},v_{\epsilon},v_{\nu})^{'}\in\mathbb{V}$
and $\tilde{\theta}\in\Theta_{n}$ with $||\tilde{\theta}-\theta_{0}||=O(\delta_{n})$,
we have 
\begin{align*}
E\left[\frac{\partial^{2}l(\tilde{\theta},W_{i})}{\partial\theta\partial\theta^{'}}[v,v]-\frac{\partial^{2}l(\theta_{0},W_{i})}{\partial\theta\partial\theta^{'}}[v,v]\right] & =v_{\psi}^{'}E\left[\frac{\partial^{2}l(\tilde{\theta},W_{i})}{\partial\psi\partial\psi^{'}}-\frac{\partial^{2}l(\theta_{0},W_{i})}{\partial\psi\partial\psi^{'}}\right]v_{\psi}\\
 & +\sum_{j\in\{\epsilon,\nu\}}2v_{\theta}^{'}E\left[\frac{\partial^{2}l(\tilde{\theta},W_{i})}{\partial\psi\partial h_{j}}[v_{j}]-\frac{\partial^{2}l(\theta_{0},W_{i})}{\partial\psi\partial h_{j}}[v_{j}]\right]\\
 & +\sum_{k\in\{\epsilon,\nu\}}\sum_{j\in\{\epsilon,\nu\}}E\left[\frac{\partial^{2}l(\tilde{\theta},W_{i})}{\partial h_{k}\partial h_{j}}[v_{k},v_{j}]-\frac{\partial^{2}l(\theta_{0},W_{i})}{\partial h_{k}\partial h_{j}}[v_{k},v_{j}]\right],
\end{align*}
 and this term can be controlled under Assumption \ref{as:secondorder}
in the same way of CFT06. This leads us to that 
\begin{align}
E[r[\hat{\theta}_{n},\theta_{0},W_{i}]-r[\hat{\theta}_{n}\pm\zeta_{n}\pi_{n}v^{*},\theta_{0},W_{i}]] & =-\frac{1}{2}(||\hat{\theta}_{n}-\theta_{0}||^{2}-||\hat{\theta}_{n}\pm\zeta_{n}\pi_{n}v^{*}-\theta_{0}||)+\zeta_{n}o(n^{-1/2})\nonumber \\
 & =\pm\zeta_{n}\times<\hat{\theta}_{n}-\theta_{0},v^{*}>+\zeta_{n}o(n^{-1/2})\label{eq:AN2}
\end{align}
 because we have $<\hat{\theta}_{n}-\theta_{0},\pi_{n}v^{*}-v^{*}>=o_{p}(n^{-1/2})$
and $||\pi_{n}v^{*}||^{2}\rightarrow||v^{*}||^{2}<\infty$. 

We also have that 
\begin{align*}
 & \mu_{n}\left(r[\hat{\theta}_{n},\theta_{0},W_{i}]-r[\hat{\theta}_{n}\pm\zeta_{n}\pi_{n}v^{*},\theta_{0},W_{i}]\right)\\
= & \mu_{n}\left(l(\hat{\theta}_{n},W_{i})-l(\hat{\theta}_{n}\pm\zeta_{n}\pi_{n}v^{*},W_{i})-\frac{\partial l(\theta_{0},W_{i})}{\partial\theta^{'}}[\mp\zeta_{n}\pi_{n}v^{*}]\right)\\
= & \mp\zeta_{n}\cdot\mu_{n}\left(\frac{\partial l(\tilde{\theta},W_{i})}{\partial\theta^{'}}[\pi_{n}v^{*}]-\frac{\partial l(\theta_{0},W_{i})}{\partial\theta^{'}}[\pi_{n}v^{*}]\right),
\end{align*}
 where $\tilde{\theta}\in\Theta_{n}$ is between $\hat{\theta}_{n}$
and $\hat{\theta}_{n}\pm\zeta_{n}\pi_{n}v^{*}$. By Assumption \ref{as:emp_process},
we have 
\begin{equation}
\mu_{n}\left(r[\hat{\theta}_{n},\theta_{0},W_{i}]-r[\hat{\theta}_{n}\pm\zeta_{n}\pi_{n}v^{*},\theta_{0},W_{i}]\right)=o_{p}(\zeta_{n}n^{-1/2}).\label{eq:AN3}
\end{equation}

Combining equations \eqref{eq:AN0} through \eqref{eq:AN3} with the
fact that $E\left[\frac{\partial l(\theta_{0},W_{i})}{\partial\theta^{'}}[v^{*}]\right]=0$,
we have 
\begin{align*}
0 & \leq\frac{1}{n}\sum_{i=1}^{n}l(\hat{\theta}_{n},W_{i})-l(\hat{\theta}_{n}\pm\zeta_{n}\pi_{n}v^{*},W_{i})\\
 & =\mp\zeta_{n}\cdot\mu_{n}\left(\frac{\partial l(\theta_{0},W_{i})}{\partial\theta^{'}}[v^{*}]\right)\pm\zeta_{n}<\hat{\theta}_{n}-\theta_{0},v^{*}>+\zeta_{n}\cdot o_{p}(n^{-1/2}),
\end{align*}
 and this results in that 
\begin{align*}
\sqrt{n}<\hat{\theta}_{n}-\theta_{0},v^{*}> & =\sqrt{n}\mu_{n}\left(\frac{\partial l(\theta_{0},W_{i})}{\partial\theta^{'}}[v^{*}]\right)+o_{p}(1)\overset{d}{\rightarrow}\mathcal{N}\left(0,||v^{*}||^{2}\right).
\end{align*}
 By Assumption \ref{as:smooth_T}, we have 
\[
\sqrt{n}\left(T(\hat{\theta}_{n})-T(\theta_{0})\right)=\sqrt{n}<\hat{\theta}_{n}-\theta_{0},v^{*}>\overset{d}{\rightarrow}\mathcal{N}\left(0,||v^{*}||^{2}\right)
\]
 by the same way in CFT06. 

\subsection{Proof of Theorem \ref{thm:AN_psi}}

Define 
\begin{equation}
\mathcal{S}_{\psi_{0}}^{'}\equiv\frac{\partial l(\theta_{0},W)}{\partial\psi^{'}}-\left\{ \frac{\partial l(\theta_{0},W)}{\partial h_{\epsilon}}[b_{\epsilon}^{*}]+\frac{\partial l(\theta_{0},W)}{\partial h_{\nu}}[b_{\nu}^{*}]\right\} ,\label{eq:score-1}
\end{equation}
 where $b_{\epsilon}^{*}=(b_{\epsilon1}^{*},...,b_{\epsilon d_{\psi}}^{*})\in\Pi_{k=1}^{d_{\psi}}(\mathcal{H}_{\epsilon}-\{h_{\epsilon0}\})$
and $b_{\nu}^{*}=(b_{\nu1}^{*},...,b_{\nu d_{\psi}}^{*})\in\Pi_{k=1}^{d_{\psi}}(\mathcal{H_{\nu}}-\{h_{\nu0}\})$
are the solutions to the following optimization problems for $k=1,2,...,d_{\psi}$:
\[
\inf_{(b_{\epsilon k},b_{\nu k})\in\bar{\mathbb{V}}_{\epsilon}\times\bar{\mathbb{V}}_{\nu}}E\left[\left(\frac{\partial l(\theta_{0},W)}{\partial\theta_{k}}-\left\{ \frac{\partial l(\theta_{0},W)}{\partial h_{\epsilon}}[b_{\epsilon k}]+\frac{\partial l(\theta_{0},W)}{\partial h_{\nu}}[b_{\nu k}]\right\} \right)^{2}\right].
\]

We consider the following assumption to establish the asymptotic normality
for $\psi_{0}$. 

\begin{assumption}\label{as:nonsingular} $\mathcal{I}_{*}(\psi_{0})\equiv E[\mathcal{S}_{\psi_{0}}\mathcal{S}_{\psi_{0}}^{'}]$
is non-singular. \end{assumption}

To prove Theorem \ref{thm:AN_psi}, take any arbitrary $\lambda\in\mathbb{R}^{d_{\psi}}-\{0\}$
with $|\lambda|\in(0,\infty)$ and let $T:\Theta\rightarrow\mathbb{R}$
be a functional of the form $T(\theta)=\lambda^{'}\psi$. Then, for
any $v\in\mathbb{V}$, we have $\frac{\partial T(\theta_{0})}{\partial\theta}[v]=\lambda^{'}v_{\psi}$
and there exist a small $\eta>0$ such that $||v||\leq\eta$ and a
constant $\tilde{c}>0$ such that 
\begin{equation}
\left|T(\theta_{0}+v)-T(\theta_{0})-\frac{\partial T(\theta_{0})}{\partial\theta}\right|\leq\tilde{c}||v||^{w}\label{eq:smoothT-1}
\end{equation}
with $w=\infty$. Therefore, Assumption \ref{as:smooth_T}-(i) is
satisfied with $w=\infty$ in this case. In addition, we have 
\begin{align*}
\sup_{v\in\mathbb{V}:||v||>0}\frac{|\lambda^{'}v_{\psi}|^{2}}{||v||^{2}} & =\sup_{v\in\mathbb{V}:||v||>0}\frac{|\lambda^{'}v_{\psi}|^{2}}{E\left[(\frac{\partial l(\theta_{0},W)}{\partial\psi^{'}}v_{\psi}+\sum_{j\in\{\epsilon,\nu\}}\frac{\partial l(\theta_{0},W)}{\partial h_{j}}[v_{j}])^{2}\right]}\\
 & =\lambda^{'}E[\mathcal{S}_{\psi_{0}}\mathcal{S}_{\psi_{0}}^{'}]^{-1}\lambda=\lambda^{'}\mathcal{I}_{*}(\theta_{0})^{-1}\lambda.
\end{align*}
 Note that the Riesz representer $v^{*}$ exists if and only if $\lambda^{'}E[\mathcal{S}_{\psi_{0}}\mathcal{S}_{\psi_{0}}^{'}]^{-1}\lambda$
is finite. Since Assumption \ref{as:nonsingular} implies that $\lambda^{'}E[\mathcal{S}_{\psi_{0}}\mathcal{S}_{\psi_{0}}^{'}]^{-1}\lambda$
is finite, Assumption \ref{as:smooth_T}-(ii) holds. Hence, by Proposition
\ref{prop:asym general}, we have 
\[
\sqrt{n}\left(\lambda^{'}\hat{\psi}_{n}-\lambda^{'}\psi_{0}\right)\overset{d}{\rightarrow}\mathcal{N}\left(0,\lambda^{'}\mathcal{I}_{*}(\psi_{0})^{-1}\lambda\right).
\]
Since $\lambda$ was arbitrary, we obtain the result by Cram\'er-Wold
device. 

\subsection{Hölder ball }

Suppose that $h\in\Lambda_{R}^{p}([0,1])$, where $p=m+\zeta$, $m\geq0$
is an integer and $\zeta\in(0,1]$ is the Hölder exponent. We want
to show that $h^{2}\in\Lambda_{\tilde{R}}^{p}([0,1])$, where $\tilde{R}=R^{2}2^{m+1}$
. Recall that $\mathcal{D}$ is the differential operator. We note
that $||h||_{\infty}\leq R$ and thus $\sup_{x}|\mathcal{D}^{\omega}h(x)|\leq R$
for all $\omega\leq m$. By Leibniz's formula, we have 
\begin{align*}
\left|\mathcal{D}^{\omega}h^{2}(x)\right| & =\left|\sum_{\iota\leq\omega}\begin{pmatrix}\omega\\
\iota
\end{pmatrix}\mathcal{D}^{\iota}h\mathcal{D}^{\omega-\iota}h\right|\leq R^{2}\sum_{\iota\leq\omega}\begin{pmatrix}\omega\\
\iota
\end{pmatrix}=R^{2}2^{\omega}\leq K^{2}2^{m}<\infty
\end{align*}
 for all $\omega\leq m$. Observe that, by Leibniz's formula, for
any $x,y\in[0,1]$ with $x\neq y$, 
\begin{align*}
\left|\mathcal{D}^{m}h^{2}(x)-\mathcal{D}^{m}h^{2}(y)\right| & =\left|\sum_{\omega\leq m}\begin{pmatrix}m\\
\omega
\end{pmatrix}\mathcal{D}^{\omega}h(x)\mathcal{D}^{m-\omega}h(x)-\sum_{\omega\leq m}\begin{pmatrix}m\\
\omega
\end{pmatrix}\mathcal{D}^{\omega}h(y)\mathcal{D}^{m-\omega}h(y)\right|\\
 & \leq\left|\sum_{\omega\leq m}\begin{pmatrix}m\\
\omega
\end{pmatrix}\mathcal{D}^{\omega}h(x)\mathcal{D}^{m-\omega}h(x)-\sum_{\omega\leq m}\begin{pmatrix}m\\
\omega
\end{pmatrix}\mathcal{D}^{\omega}h(y)\mathcal{D}^{m-\omega}h(x)\right|\\
 & \ \ +\left|\sum_{\omega\leq m}\begin{pmatrix}m\\
\omega
\end{pmatrix}\mathcal{D}^{\omega}h(y)\mathcal{D}^{m-\omega}h(x)-\sum_{\omega\leq m}\begin{pmatrix}m\\
\omega
\end{pmatrix}\mathcal{D}^{\omega}h(y)\mathcal{D}^{m-\omega}h(y)\right|\\
 & \leq2\times\{\sup_{\omega\leq m}\sup_{x}\left|\mathcal{D}^{\omega}h(x)\right|\}\times\left|\sum_{\omega\leq m}\begin{pmatrix}m\\
\omega
\end{pmatrix}\{\mathcal{D}^{\omega}h(x)-\mathcal{D}^{\omega}h(y)\}\right|\\
 & \leq2R\sum_{\omega\leq m}\begin{pmatrix}m\\
\omega
\end{pmatrix}\left|\mathcal{D}^{\omega}h(x)-\mathcal{D}^{\omega}h(y)\right|.
\end{align*}
 We also have that, for all $\omega<m$, 
\begin{align*}
\frac{|\mathcal{D}^{\omega}h(x)-D^{\omega}h(y)|}{|x-y|^{\zeta}} & =\frac{|\mathcal{D}^{\omega}h(x)-\mathcal{D}^{\omega}h(y)|}{|x-y|}|x-y|^{1-\zeta}=|\mathcal{D}^{\omega+1}h(\tilde{x})||x-y|^{1-\zeta}\leq R,
\end{align*}
 where $\tilde{x}$ is between $x$ and $y$. Note that $\zeta\in(0,1]$
and thus $|x-y|^{1-\zeta}\leq1$ for all $x,y\in[0,1]$. Since $h\in\Lambda_{R}^{p}([0,1])$,
we have $\frac{|\mathcal{D}^{m}h(x)-\mathcal{D}^{m}h(y)|}{|x-y|^{\zeta}}\leq R$.
Hence, 
\begin{align*}
\frac{|\mathcal{D}^{m}h^{2}(x)-\mathcal{D}^{m}h^{2}(y)|}{|x-y|^{\zeta}} & \leq2R\sum_{\omega\leq m}\begin{pmatrix}m\\
\omega
\end{pmatrix}\frac{|\mathcal{D}^{\omega}h(x)-\mathcal{D}^{\omega}h(y)|}{|x-y|^{\zeta}}\leq2R^{2}\sum_{\omega\leq m}\begin{pmatrix}m\\
\omega
\end{pmatrix}=R^{2}2^{m+1}<\infty,
\end{align*}
 and this implies that $h^{2}\in\Lambda_{\tilde{R}}^{p}([0,1])$ with
$\tilde{R}=R^{2}2^{m+1}$.

\section{\label{sec:add_simulation}Additional Simulation Results }

\subsection{\label{subsec:larger_sample}A Larger Sample Size}

Tables \ref{tab:correct_1000} and \ref{tab:mar_1000} show the simulation
results with a larger sample size ($n=1000$). We can see that the
main findings in the main text remain the same even with this larger
sample size. 

\subsection{\label{subsec:copula_mis}Copula and Marginal Misspecification }

We consider the simulation results when both the copula and the marginal
distributions are misspecified, reported in Tables \ref{tab:cop1_mar_500}\textendash \ref{tab:cop4_mar_500}
and \ref{tab:cop1_mar_1000}\textendash \ref{tab:cop4_mar_1000}.
If both the copula and the marginal distributions are misspecified,
the performance of the parametric ML estimators are comparable to,
or slightly worse than that under marginal misspecification. Consider,
for example, the case where the true copula function is the Frank
copula and the sample size is 500. The estimators of $\psi$ under
both the copula and marginal misspecification (Table \ref{tab:cop2_mar_500})
have slightly larger root mean squared errors (RMSEs) than the corresponding
estimators under the marginal misspecification (Table \ref{tab:mar_500}).
On the other hand, the performance of the estimators of the ATE varies
across copula specifications. In particular, when the true data generating
process (DGP) is based on the Gumbel copula, the copula and marginal
misspecification has a significant effect on the performance of the
parametric estimators of the ATE. The RMSEs of the estimators of the
ATE under the copula and marginal misspecification (Table \ref{tab:cop4_mar_500})
are larger than those under the marginal misspecification (Table \ref{tab:mar_500}).
Specifically, the RMSE of the parametric estimator of the ATE under
the marginal misspecification is 0.1637 (Table \ref{tab:mar_500}),
whereas the RMSEs of the corresponding estimators under both the copula
and marginal misspecification are 0.1835, 0.2178, and 0.2732 when
the Gaussian, Frank, and Clayton copulas are used, respectively (Table
\ref{tab:cop4_mar_500}). On the other hand, there is no clear evidence
that the performance of the sieve ML estimators under both the copula
and marginal misspecification is worse than that under misspecification
of the marginal distributions. For example, when the true copula belongs
to the Frank family but the copula is specified as the Gaussian or
Gumbel copula, we can see that the RMSEs of the sieve ML estimators
of the finite-dimensional parameters other than $\gamma$ and the
ATE under the copula and marginal misspecification (Table \ref{tab:cop2_mar_500})
are lower than those under the marginal misspecification (Table \ref{tab:mar_500}).
In contrast, we can see from the same tables that the Clayton copula
specification draws the opposite conclusion when the true copula is
the Frank. In general, no matter whether the copula is misspecified,
we find that the sieve ML estimators outperform the parametric estimators
in terms of the RMSE when the marginal distributions are misspecified. 

\subsection{\label{subsec:fat_tail}Unknown Marginal Density Functions with Fat
Tails}

We examine the finite sample performance of the sieve ML estimator
of $\theta_{0}$ when the unknown marginal density functions $f_{\epsilon0}$
and $f_{\nu0}$ have fat tails. We consider the $t$ distribution
with 3 degree of freedom as the true marginal distributions. While
the marginal distributions in the parametric models are specified
by normal distributions, we consider two specifications for the semiparametric
models. These specifications differ in the choice of $G$: we choose
the standard normal distribution and the distribution function of
$t(3)$ for $G$ in the first and second specifications, respectively.
All simulation results are obtained with 500 observations and 2000
simulation iterations. 

Table \ref{tab:t3} presents simulation results. While the parametric
estimates have larger standard deviations, the biases of the semiparametric
estimates are larger than those of the parametric estimates. However,
the resulting RMSEs of the semiparametric estimates are slightly larger
than those of the parametric estimates. This is because the semiparametric
specification does not satisfy the assumptions required for the asymptotic
theory. 

Table \ref{tab:t3_G} shows simulation results where $G$ is the distribution
function of $t(3)$.  The performance of semiparametric estimator
is comparable to that of parametric estimator in terms of the RMSE.
The biases of semiparametric estimates in Table \ref{tab:t3_G} are
much smaller than those in Table \ref{tab:t3}, and the standard deviations
of semiparametric estimates are very similar to those of parametric
estimates. 

The simulation results in Tables \ref{tab:t3} and \ref{tab:t3_G}
suggest that if a researcher has a prior belief about the tail behavior
of the unknown marginal density functions, it should be reflected
in the choice of $G$ for semiparametric models. If it is believed
that the marginal density functions have fat tails, one may choose
a distribution function with fat tails for $G$, such as the distribution
function of $t(3)$. 

\subsection{\label{subsec:diff_dependece}Different Degrees of Dependence }

Tables \ref{tab:correct_02} through \ref{tab:mis_neg} provide simulation
results across various degrees of dependence between $\epsilon$ and
$\nu$. The dependence measure is unified into the Spearman's $\rho$,
and we consider cases of $\rho_{sp}\in\{-0.5,0.2,0.7\}$.\footnote{Note that we only consider the Gaussian and Frank copulas for $\rho_{sp}=-0.5$
as the Clayton or the Gumbel copula does not allow for negative dependence.} We find that regardless of degrees of dependence, the results in
our main paper remains the same: (i) the performance of the semiparametric
estimator is comparable to that of the parametric estimator under
correct specification, (ii) the semiparametric estimators outperform
the parametric estimators under misspecification of the marginals.

\subsection{\label{subsec:cp_boot}Coverage Probabilities of Bootstrap Confidence
Intervals}

We conduct simulations to investigate coverage probabilities of bootstrap
confidence intervals (CIs). We consider the following design: 
\begin{align*}
Y_{i} & =\mathbf{1}\{-X_{1i}+X_{2i}\beta+D_{i}\delta\geq\epsilon_{i}\},\quad D_{i}=\mathbf{1}\{-X_{1i}+X_{2i}\alpha+Z_{i}\gamma\geq\nu_{i}\},
\end{align*}
 where $(\alpha,\gamma,\beta,\delta)=(0.5,0.8,0.8,1.1)$ and $(\epsilon,\nu)$
are generated from the Gaussian copula and normal marginals with $\rho_{sp}=0.5$.
$(X_{1i},X_{2i},Z_{i})$ is drawn from a multivariate normal distribution.
Note that the coefficients on $X_{1i}$ are fixed for scale normalization.
The sample size, number of bootstrap iterations, and number of simulations
are 500, 200, and 200, respectively. We consider two types of CIs:
(i) CIs using the normal approximation, (ii) the percentile bootstrap
CIs. 

Table \ref{tab:cp_boot} presents the coverage probabilities of both
CIs. We find that the bootstrap percentile CIs performs better than
the CIs based on the normal approximation and that their coverage
probabilities are close to the nominal level (95\%).

\newpage{}

\begin{table}[H]
\caption{\label{tab:correct_1000}Correct Specification ($n=1,000$) (True
marginal: normal)}

\bigskip{}

\centering{}{\small{}}%
\begin{tabular}{|c|c|c|c|c||c|c|c|c|c|}
\hline 
\multicolumn{5}{|c||}{Parametric Estimation, Gaussian Copula} & \multicolumn{5}{c|}{Semiparametric Estimation, Gaussian Copula}\tabularnewline
\hline 
 & {\small{}$\gamma$} & {\small{}$\delta_{1}$} & {\small{}$\rho_{sp}$} & {\small{}$ATE$} &  & {\small{}$\gamma$} & {\small{}$\delta_{1}$} & {\small{}$\rho_{sp}$} & {\small{}$ATE$}\tabularnewline
\hline 
{\small{}True Values} & {\small{}0.8000} & {\small{}1.1000} & {\small{}0.5000} & {\small{}0.3643} & {\small{}True Values} & {\small{}0.8000} & {\small{}1.1000} & {\small{}0.5000} & {\small{}0.3643}\tabularnewline
\hline 
{\small{}Estimate} & {\small{}0.8025} & {\small{}1.1165} & {\small{}0.4996} & {\small{}0.3632} & {\small{}Estimate} & {\small{}0.8026} & {\small{}1.1205} & {\small{}0.5031} & {\small{}0.3596}\tabularnewline
\hline 
{\small{}S.D} & {\small{}0.0654} & {\small{}0.2737} & {\small{}0.1081} & {\small{}0.0656} & {\small{}S.D} & {\small{}0.0655} & {\small{}0.2939} & {\small{}0.1092} & {\small{}0.0668}\tabularnewline
\hline 
{\small{}Bias} & {\small{}0.0025} & {\small{}0.0165} & {\small{}-0.0004} & {\small{}-0.0011} & {\small{}Bias} & {\small{}0.0026} & {\small{}0.0205} & {\small{}0.0031} & {\small{}-0.0048}\tabularnewline
\hline 
{\small{}RMSE} & {\small{}0.0655} & {\small{}0.2742} & {\small{}0.1081} & {\small{}0.0656} & {\small{}RMSE} & {\small{}0.0655} & {\small{}0.2946} & {\small{}0.1092} & {\small{}0.0670}\tabularnewline
\hline 
\multicolumn{5}{|c||}{Parametric Estimation, Frank Copula} & \multicolumn{5}{c|}{Semiparametric Estimation, Frank Copula}\tabularnewline
\hline 
\end{tabular}{\small\par}
\end{table}
\begin{table}[H]
\caption{\label{tab:mar_1000}Misspecification of Marginals ($n=1,000$) (True
marginal: mixture of normals)}

\bigskip{}

\centering{}{\small{}}%
\begin{tabular}{|c|c|c|c|c||c|c|c|c|c|}
\hline 
\multicolumn{5}{|c||}{Parametric Estimation, Gaussian Copula} & \multicolumn{5}{c|}{Semiparametric Estimation, Gaussian Copula}\tabularnewline
\hline 
 & {\small{}$\gamma$} & {\small{}$\delta_{1}$} & {\small{}$\rho_{sp}$} & {\small{}$ATE$} &  & {\small{}$\gamma$} & {\small{}$\delta_{1}$} & {\small{}$\rho_{sp}$} & {\small{}$ATE$}\tabularnewline
\hline 
{\small{}True Values} & {\small{}0.8000} & {\small{}1.1000} & {\small{}0.5000} & {\small{}0.1066} & {\small{}True Values} & {\small{}0.8000} & {\small{}1.1000} & {\small{}0.5000} & {\small{}0.1066}\tabularnewline
\hline 
{\small{}Estimate} & {\small{}0.7941} & {\small{}1.0549} & {\small{}0.4496} & {\small{}0.2447} & {\small{}Estimate} & {\small{}0.8641} & {\small{}1.3030} & {\small{}0.4778} & {\small{}0.1262}\tabularnewline
\hline 
{\small{}S.D} & {\small{}0.0911} & {\small{}0.4256} & {\small{}0.1156} & {\small{}0.0807} & {\small{}S.D} & {\small{}0.0778} & {\small{}0.2576} & {\small{}0.0721} & {\small{}0.0463}\tabularnewline
\hline 
{\small{}Bias} & {\small{}-0.0059} & {\small{}-0.0451} & {\small{}-0.0504} & {\small{}0.1381} & {\small{}Bias} & {\small{}0.0641} & {\small{}0.2030} & {\small{}-0.0222} & {\small{}0.0195}\tabularnewline
\hline 
{\small{}RMSE} & {\small{}0.0913} & {\small{}0.4279} & {\small{}0.1261} & {\small{}0.1599} & {\small{}RMSE} & {\small{}0.1008} & {\small{}0.3279} & {\small{}0.0755} & {\small{}0.0502}\tabularnewline
\hline 
\multicolumn{5}{|c||}{Parametric Estimation, Frank Copula} & \multicolumn{5}{c|}{Semiparametric Estimation, Frank Copula}\tabularnewline
\hline 
\end{tabular}{\small\par}
\end{table}
\begin{table}[H]
\caption{\label{tab:cop1_mar_500}Copula and Marginals Misspecification 1 ($n=500$)
(True copula: Gaussian, true marginal: mixture of normals)}
\bigskip{}

\centering{}{\small{}}%
\begin{tabular}{|c|c|c|c|c||c|c|c|c|c|}
\hline 
\multicolumn{5}{|c||}{Parametric Estimation, Frank Copula} & \multicolumn{5}{c|}{Semiparametric Estimation, Frank Copula}\tabularnewline
\hline 
 & {\small{}$\gamma$} & {\small{}$\delta_{1}$} & {\small{}$\rho_{sp}$} & {\small{}$ATE$} &  & {\small{}$\gamma$} & {\small{}$\delta_{1}$} & {\small{}$\rho_{sp}$} & {\small{}$ATE$}\tabularnewline
\hline 
{\small{}True Values} & {\small{}0.8000} & {\small{}1.1000} & {\small{}0.5000} & {\small{}0.1066} & {\small{}True Values} & {\small{}0.8000} & {\small{}1.1000} & {\small{}0.5000} & {\small{}0.1066}\tabularnewline
\hline 
{\small{}Estimate} & {\small{}0.8140} & {\small{}1.3080} & {\small{}0.3775} & {\small{}0.2916} & {\small{}Estimate} & {\small{}0.8463} & {\small{}1.3514} & {\small{}0.4499} & {\small{}0.1351}\tabularnewline
\hline 
{\small{}S.D} & {\small{}0.1257} & {\small{}0.4899} & {\small{}0.1202} & {\small{}0.0862} & {\small{}S.D} & {\small{}0.1137} & {\small{}0.3502} & {\small{}0.0964} & {\small{}0.0686}\tabularnewline
\hline 
{\small{}Bias} & {\small{}0.0140} & {\small{}0.2080} & {\small{}-0.1225} & {\small{}0.1849} & {\small{}Bias} & {\small{}0.0463} & {\small{}0.2514} & {\small{}-0.0501} & {\small{}0.0285}\tabularnewline
\hline 
{\small{}RMSE} & {\small{}0.1265} & {\small{}0.5322} & {\small{}0.1716} & {\small{}0.2040} & {\small{}RMSE} & {\small{}0.1227} & {\small{}0.4311} & {\small{}0.1087} & {\small{}0.0743}\tabularnewline
\hline 
\multicolumn{5}{|c||}{Parametric Estimation, Clayton Copula} & \multicolumn{5}{c|}{Semiparametric Estimation, Clayton Copula}\tabularnewline
\hline 
\end{tabular}{\small\par}
\end{table}
\begin{table}[H]
\caption{\label{tab:cop2_mar_500}Copula and Marginals Misspecification 2 ($n=500$)
(True copula: Frank, true marginal: mixture of normals)}
\bigskip{}

\centering{}{\small{}}%
\begin{tabular}{|c|c|c|c|c|c|c|c|c|c|}
\hline 
\multicolumn{5}{|c|}{Parametric Estimation, Gaussian Copula} & \multicolumn{5}{c|}{Semiparametric Estimation, Gaussian Copula}\tabularnewline
\hline 
 & {\small{}$\gamma$} & {\small{}$\delta_{1}$} & {\small{}$\rho_{sp}$} & {\small{}$ATE$} &  & {\small{}$\gamma$} & {\small{}$\delta_{1}$} & {\small{}$\rho_{sp}$} & {\small{}$ATE$}\tabularnewline
\hline 
{\small{}True Values} & {\small{}0.8000} & {\small{}1.1000} & {\small{}0.5000} & {\small{}0.1066} & {\small{}True Values} & {\small{}0.8000} & {\small{}1.1000} & {\small{}0.5000} & {\small{}0.1066}\tabularnewline
\hline 
{\small{}Estimate} & {\small{}0.7992} & {\small{}1.1673} & {\small{}0.4517} & {\small{}0.2527} & {\small{}Estimate} & {\small{}0.8500} & {\small{}1.1788} & {\small{}0.5173} & {\small{}0.1192}\tabularnewline
\hline 
{\small{}S.D} & {\small{}0.1342} & {\small{}0.6901} & {\small{}0.1680} & {\small{}0.1179} & {\small{}S.D} & {\small{}0.1158} & {\small{}0.3602} & {\small{}0.1000} & {\small{}0.0652}\tabularnewline
\hline 
{\small{}Bias} & {\small{}-0.0008} & {\small{}0.0673} & {\small{}-0.0483} & {\small{}0.1461} & {\small{}Bias} & {\small{}0.0500} & {\small{}0.0788} & {\small{}0.0173} & {\small{}0.0126}\tabularnewline
\hline 
{\small{}RMSE} & {\small{}0.1342} & {\small{}0.6934} & {\small{}0.1748} & {\small{}0.1877} & {\small{}RMSE} & {\small{}0.1262} & {\small{}0.3687} & {\small{}0.1015} & {\small{}0.0664}\tabularnewline
\hline 
\multicolumn{5}{|c|}{Parametric Estimation, Clayton Copula} & \multicolumn{5}{c|}{Semiparametric Estimation, Clayton Copula}\tabularnewline
\hline 
\end{tabular}{\small\par}
\end{table}
\begin{table}[H]
\caption{\label{tab:cop3_mar_500}Copula and Marginals Misspecification 3 ($n=500$)
(True copula: Clayton, true marginal: mixture of normals)}
\bigskip{}

\centering{}{\small{}}%
\begin{tabular}{|c|c|c|c|c||c|c|c|c|c|}
\hline 
\multicolumn{5}{|c||}{Parametric Estimation, Gaussian Copula} & \multicolumn{5}{c|}{Semiparametric Estimation, Gaussian Copula}\tabularnewline
\hline 
 & {\small{}$\gamma$} & {\small{}$\delta_{1}$} & {\small{}$\rho_{sp}$} & {\small{}$ATE$} &  & {\small{}$\gamma$} & {\small{}$\delta_{1}$} & {\small{}$\rho_{sp}$} & {\small{}$ATE$}\tabularnewline
\hline 
{\small{}True Values} & {\small{}0.8000} & {\small{}1.1000} & {\small{}0.5000} & {\small{}0.1066} & {\small{}True Values} & {\small{}0.8000} & {\small{}1.1000} & {\small{}0.5000} & {\small{}0.1066}\tabularnewline
\hline 
{\small{}Estimate} & {\small{}0.7986} & {\small{}1.0471} & {\small{}0.4017} & {\small{}0.2392} & {\small{}Estimate} & {\small{}0.8533} & {\small{}1.1780} & {\small{}0.4493} & {\small{}0.1076}\tabularnewline
\hline 
{\small{}S.D} & {\small{}0.1346} & {\small{}0.6366} & {\small{}0.1731} & {\small{}0.1181} & {\small{}S.D} & {\small{}0.1164} & {\small{}0.3438} & {\small{}0.1033} & {\small{}0.0569}\tabularnewline
\hline 
{\small{}Bias} & {\small{}-0.0014} & {\small{}-0.0529} & {\small{}-0.0983} & {\small{}0.1325} & {\small{}Bias} & {\small{}0.0533} & {\small{}0.0780} & {\small{}-0.0508} & {\small{}0.0009}\tabularnewline
\hline 
{\small{}RMSE} & {\small{}0.1346} & {\small{}0.6388} & {\small{}0.1991} & {\small{}0.1775} & {\small{}RMSE} & {\small{}0.1281} & {\small{}0.3525} & {\small{}0.1151} & {\small{}0.0569}\tabularnewline
\hline 
\multicolumn{5}{|c||}{Parametric Estimation, Frank Copula} & \multicolumn{5}{c|}{Semiparametric Estimation, Frank Copula}\tabularnewline
\hline 
\end{tabular}{\small\par}
\end{table}
\begin{table}[H]
\caption{\label{tab:cop4_mar_500}Copula and Marginals Misspecification 4 ($n=500$)
(True copula: Gumbel, true marginal: mixture of normals)}

\bigskip{}

\centering{}{\small{}}%
\begin{tabular}{|c|c|c|c|c||c|c|c|c|c|}
\hline 
\multicolumn{5}{|c||}{Parametric Estimation, Gaussian Copula} & \multicolumn{5}{c|}{Semiparametric Estimation, Gaussian Copula}\tabularnewline
\hline 
 & {\small{}$\gamma$} & {\small{}$\delta_{1}$} & {\small{}$\rho_{sp}$} & {\small{}$ATE$} &  & {\small{}$\gamma$} & {\small{}$\delta_{1}$} & {\small{}$\rho_{sp}$} & {\small{}$ATE$}\tabularnewline
\hline 
{\small{}True Values} & {\small{}0.8000} & {\small{}1.1000} & {\small{}0.5000} & {\small{}0.1066} & {\small{}True Values} & {\small{}0.8000} & {\small{}1.1000} & {\small{}0.5000} & {\small{}0.1066}\tabularnewline
\hline 
{\small{}Estimate} & {\small{}0.7978} & {\small{}1.1488} & {\small{}0.4658} & {\small{}0.2523} & {\small{}Estimate} & {\small{}0.8609} & {\small{}1.3801} & {\small{}0.4957} & {\small{}0.1460}\tabularnewline
\hline 
{\small{}S.D} & {\small{}0.1304} & {\small{}0.6489} & {\small{}0.1598} & {\small{}0.1117} & {\small{}S.D} & {\small{}0.1132} & {\small{}0.3749} & {\small{}0.1052} & {\small{}0.0730}\tabularnewline
\hline 
{\small{}Bias} & {\small{}-0.0022} & {\small{}0.0488} & {\small{}-0.0342} & {\small{}0.1456} & {\small{}Bias} & {\small{}0.0609} & {\small{}0.2801} & {\small{}-0.0042} & {\small{}0.0393}\tabularnewline
\hline 
{\small{}RMSE} & {\small{}0.1304} & {\small{}0.6508} & {\small{}0.1634} & {\small{}0.1835} & {\small{}RMSE} & {\small{}0.1286} & {\small{}0.4679} & {\small{}0.1053} & {\small{}0.0829}\tabularnewline
\hline 
\multicolumn{5}{|c||}{Parametric Estimation, Frank Copula} & \multicolumn{5}{c|}{Semiparametric Estimation, Frank Copula}\tabularnewline
\hline 
\end{tabular}{\small\par}
\end{table}
\begin{table}[H]
\caption{\label{tab:cop1_mar_1000}Copula and Marginals Misspecification 1
($n=1,000$) (True copula: Gaussian, true marginal: mixture of normals)}

\bigskip{}

\centering{}{\small{}}%
\begin{tabular}{|c|c|c|c|c||c|c|c|c|c|}
\hline 
\multicolumn{5}{|c||}{Parametric Estimation, Frank Copula} & \multicolumn{5}{c|}{Semiparametric Estimation, Frank Copula}\tabularnewline
\hline 
 & {\small{}$\gamma$} & {\small{}$\delta_{1}$} & {\small{}$\rho_{sp}$} & {\small{}$ATE$} &  & {\small{}$\gamma$} & {\small{}$\delta_{1}$} & {\small{}$\rho_{sp}$} & {\small{}$ATE$}\tabularnewline
\hline 
{\small{}True Values} & {\small{}0.8000} & {\small{}1.1000} & {\small{}0.5000} & {\small{}0.1066} & {\small{}True Values} & {\small{}0.8000} & {\small{}1.1000} & {\small{}0.5000} & {\small{}0.1066}\tabularnewline
\hline 
{\small{}Estimate} & {\small{}0.8086} & {\small{}1.3159} & {\small{}0.3652} & {\small{}0.2975} & {\small{}Estimate} & {\small{}0.8549} & {\small{}1.3936} & {\small{}0.4376} & {\small{}0.1371}\tabularnewline
\hline 
{\small{}S.D} & {\small{}0.0897} & {\small{}0.3636} & {\small{}0.0927} & {\small{}0.0650} & {\small{}S.D} & {\small{}0.0830} & {\small{}0.2548} & {\small{}0.0689} & {\small{}0.0506}\tabularnewline
\hline 
{\small{}Bias} & {\small{}0.0086} & {\small{}0.2159} & {\small{}-0.1347} & {\small{}0.1909} & {\small{}Bias} & {\small{}0.0549} & {\small{}0.2936} & {\small{}-0.0623} & {\small{}0.0305}\tabularnewline
\hline 
{\small{}RMSE} & {\small{}0.0901} & {\small{}0.4229} & {\small{}0.1636} & {\small{}0.2017} & {\small{}RMSE} & {\small{}0.0995} & {\small{}0.3887} & {\small{}0.0929} & {\small{}0.0591}\tabularnewline
\hline 
\multicolumn{5}{|c||}{Parametric Estimation, Clayton Copula} & \multicolumn{5}{c|}{Semiparametric Estimation, Clayton Copula}\tabularnewline
\hline 
\end{tabular}{\small\par}
\end{table}
\begin{table}[H]
\caption{\label{tab:cop2_mar_1000}Copula and Marginals Misspecification 2
($n=1,000$) (True copula: Frank, true marginal: mixture of normals)}

\bigskip{}

\centering{}{\small{}}%
\begin{tabular}{|c|c|c|c|c||c|c|c|c|c|}
\hline 
\multicolumn{5}{|c||}{Parametric Estimation, Gaussian Copula} & \multicolumn{5}{c|}{Semiparametric Estimation, Gaussian Copula}\tabularnewline
\hline 
 & {\small{}$\gamma$} & {\small{}$\delta_{1}$} & {\small{}$\rho_{sp}$} & {\small{}$ATE$} &  & {\small{}$\gamma$} & {\small{}$\delta_{1}$} & {\small{}$\rho_{sp}$} & {\small{}$ATE$}\tabularnewline
\hline 
{\small{}True Values} & {\small{}0.8000} & {\small{}1.1000} & {\small{}0.5000} & {\small{}0.1066} & {\small{}True Values} & {\small{}0.8000} & {\small{}1.1000} & {\small{}0.5000} & {\small{}0.1066}\tabularnewline
\hline 
{\small{}Estimate} & {\small{}0.7935} & {\small{}1.0825} & {\small{}0.4653} & {\small{}0.2465} & {\small{}Estimate} & {\small{}0.8601} & {\small{}1.1832} & {\small{}0.5145} & {\small{}0.1196}\tabularnewline
\hline 
{\small{}S.D} & {\small{}0.0926} & {\small{}0.4333} & {\small{}0.1152} & {\small{}0.0803} & {\small{}S.D} & {\small{}0.0768} & {\small{}0.2641} & {\small{}0.0723} & {\small{}0.0450}\tabularnewline
\hline 
{\small{}Bias} & {\small{}-0.0065} & {\small{}-0.0175} & {\small{}-0.0347} & {\small{}0.1399} & {\small{}Bias} & {\small{}0.0601} & {\small{}0.0832} & {\small{}0.0145} & {\small{}0.0130}\tabularnewline
\hline 
{\small{}RMSE} & {\small{}0.0929} & {\small{}0.4336} & {\small{}0.1203} & {\small{}0.1613} & {\small{}RMSE} & {\small{}0.0976} & {\small{}0.2769} & {\small{}0.0738} & {\small{}0.0468}\tabularnewline
\hline 
\multicolumn{5}{|c||}{Parametric Estimation, Clayton Copula} & \multicolumn{5}{c|}{Semiparametric Estimation, Clayton Copula}\tabularnewline
\hline 
\end{tabular}{\small\par}
\end{table}
\begin{table}[H]
\caption{\label{tab:cop3_mar_1000}Copula and Marginals Misspecification 3
($n=1,000$) (True copula: Clayton, true marginal: mixture of normals)}

\bigskip{}

\centering{}{\small{}}%
\begin{tabular}{|c|c|c|c|c||c|c|c|c|c|}
\hline 
\multicolumn{5}{|c||}{Parametric Estimation, Gaussian Copula} & \multicolumn{5}{c|}{Semiparametric Estimation, Gaussian Copula}\tabularnewline
\hline 
 & {\small{}$\gamma$} & {\small{}$\delta_{1}$} & {\small{}$\rho_{sp}$} & {\small{}$ATE$} &  & {\small{}$\gamma$} & {\small{}$\delta_{1}$} & {\small{}$\rho_{sp}$} & {\small{}$ATE$}\tabularnewline
\hline 
{\small{}True Values} & {\small{}0.8000} & {\small{}1.1000} & {\small{}0.5000} & {\small{}0.1066} & {\small{}True Values} & {\small{}0.8000} & {\small{}1.1000} & {\small{}0.5000} & {\small{}0.1066}\tabularnewline
\hline 
{\small{}Estimate} & {\small{}0.7928} & {\small{}0.9952} & {\small{}0.4102} & {\small{}0.2370} & {\small{}Estimate} & {\small{}0.8618} & {\small{}1.2015} & {\small{}0.4441} & {\small{}0.1097}\tabularnewline
\hline 
{\small{}S.D} & {\small{}0.0929} & {\small{}0.4262} & {\small{}0.1233} & {\small{}0.0837} & {\small{}S.D} & {\small{}0.0764} & {\small{}0.2527} & {\small{}0.0737} & {\small{}0.0411}\tabularnewline
\hline 
{\small{}Bias} & {\small{}-0.0072} & {\small{}-0.1048} & {\small{}-0.0898} & {\small{}0.1303} & {\small{}Bias} & {\small{}0.0618} & {\small{}0.1015} & {\small{}-0.0559} & {\small{}0.0030}\tabularnewline
\hline 
{\small{}RMSE} & {\small{}0.0932} & {\small{}0.4389} & {\small{}0.1525} & {\small{}0.1549} & {\small{}RMSE} & {\small{}0.0983} & {\small{}0.2723} & {\small{}0.0925} & {\small{}0.0412}\tabularnewline
\hline 
\multicolumn{5}{|c||}{Parametric Estimation, Frank Copula} & \multicolumn{5}{c|}{Semiparametric Estimation, Frank Copula}\tabularnewline
\hline 
\end{tabular}{\small\par}
\end{table}
\begin{table}[H]
\caption{\label{tab:cop4_mar_1000}Copula and Marginals Misspecification 4
($n=1,000$) (True copula: Gumbel, true DGP marginal: mixture of normals)}

\bigskip{}

\centering{}{\small{}}%
\begin{tabular}{|c|c|c|c|c||c|c|c|c|c|}
\hline 
\multicolumn{5}{|c||}{Parametric Estimation, Gaussian Copula} & \multicolumn{5}{c|}{Semiparametric Estimation, Gaussian Copula}\tabularnewline
\hline 
 & {\small{}$\gamma$} & {\small{}$\delta_{1}$} & {\small{}$\rho_{sp}$} & {\small{}$ATE$} &  & {\small{}$\gamma$} & {\small{}$\delta_{1}$} & {\small{}$\rho_{sp}$} & {\small{}$ATE$}\tabularnewline
\hline 
{\small{}True Values} & {\small{}0.8000} & {\small{}1.1000} & {\small{}0.5000} & {\small{}0.1066} & {\small{}True Values} & {\small{}0.8000} & {\small{}1.1000} & {\small{}0.5000} & {\small{}0.1066}\tabularnewline
\hline 
{\small{}Estimate} & {\small{}0.7905} & {\small{}1.1059} & {\small{}0.4669} & {\small{}0.2520} & {\small{}Estimate} & {\small{}0.8660} & {\small{}1.4046} & {\small{}0.4893} & {\small{}0.1428}\tabularnewline
\hline 
{\small{}S.D} & {\small{}0.0896} & {\small{}0.4412} & {\small{}0.1167} & {\small{}0.0815} & {\small{}S.D} & {\small{}0.0775} & {\small{}0.2644} & {\small{}0.0723} & {\small{}0.0508}\tabularnewline
\hline 
{\small{}Bias} & {\small{}-0.0095} & {\small{}0.0059} & {\small{}-0.0330} & {\small{}0.1454} & {\small{}Bias} & {\small{}0.0660} & {\small{}0.3046} & {\small{}-0.0107} & {\small{}0.0362}\tabularnewline
\hline 
{\small{}RMSE} & {\small{}0.0901} & {\small{}0.4412} & {\small{}0.1213} & {\small{}0.1667} & {\small{}RMSE} & {\small{}0.1018} & {\small{}0.4034} & {\small{}0.0730} & {\small{}0.0624}\tabularnewline
\hline 
\multicolumn{5}{|c||}{Parametric Estimation, Frank Copula} & \multicolumn{5}{c|}{Semiparametric Estimation, Frank Copula}\tabularnewline
\hline 
\end{tabular}{\small\par}
\end{table}
\begin{table}[H]
\begin{centering}
\caption{\label{tab:t3} Misspecification of Marginals ($n=500$) (True Marginal:
$t(3)$)}
\bigskip{}
\par\end{centering}
\begin{centering}
{\small{}}%
\begin{tabular}{|c|c|c|c|c||c|c|c|c|c|}
\hline 
\multicolumn{5}{|c||}{Parametric Estimation, Gaussian Copula} & \multicolumn{5}{c|}{Semiparametric Estimation$\dagger$, Gaussian Copula}\tabularnewline
\hline 
 & {\small{}$\gamma$} & {\small{}$\delta_{1}$} & {\small{}$\rho_{sp}$} & {\small{}$ATE$} &  & {\small{}$\gamma$} & {\small{}$\delta_{1}$} & {\small{}$\rho_{sp}$} & {\small{}$ATE$}\tabularnewline
\hline 
{\small{}True Values} & {\small{}0.8000} & {\small{}1.1000} & {\small{}0.5000} & {\small{}0.3242} & {\small{}True Values} & {\small{}0.8000} & {\small{}1.1000} & {\small{}0.5000} & {\small{}0.3242}\tabularnewline
\hline 
{\small{}Estimate} & {\small{}0.8060} & {\small{}1.1937} & {\small{}0.4938} & {\small{}0.3098} & {\small{}Estimate} & {\small{}0.7288} & {\small{}0.8480} & {\small{}0.5762} & {\small{}0.2499}\tabularnewline
\hline 
{\small{}S.D} & {\small{}0.1119} & {\small{}0.5749} & {\small{}0.1647} & {\small{}0.1068} & {\small{}S.D} & {\small{}0.1037} & {\small{}0.3832} & {\small{}0.1339} & {\small{}0.1085}\tabularnewline
\hline 
{\small{}Bias} & {\small{}0.0060} & {\small{}0.0937} & {\small{}-0.0062} & {\small{}-0.0143} & {\small{}Bias} & {\small{}-0.0712} & {\small{}-0.2520} & {\small{}0.0763} & {\small{}-0.0742}\tabularnewline
\hline 
{\small{}RMSE} & {\small{}0.0125} & {\small{}0.3306} & {\small{}0.0271} & {\small{}0.0116} & {\small{}RMSE} & {\small{}0.0108} & {\small{}0.1468} & {\small{}0.0179} & {\small{}0.0173}\tabularnewline
\hline 
\multicolumn{5}{|c||}{Parametric Estimation, Frank Copula} & \multicolumn{5}{c|}{Semiparametric Estimation$\dagger$, Frank Copula}\tabularnewline
\hline 
\end{tabular}{\small\par}
\par\end{centering}
\begin{centering}
\bigskip{}
\par\end{centering}
$\dagger$: The semiparametric models are specified with $G=\Phi$,
where $\Phi(\cdot)$ is the standard normal distribution function. 
\end{table}
\begin{table}[H]
\begin{centering}
\caption{\label{tab:t3_G} Misspecification of Marginals ($n=500$) (True marginal:
$t(3)$)}
\bigskip{}
\par\end{centering}
\begin{centering}
{\small{}}%
\begin{tabular}{|c|c|c|c|c||c|c|c|c|c|}
\hline 
\multicolumn{5}{|c||}{Parametric Estimation, Gaussian Copula} & \multicolumn{5}{c|}{Semiparametric Estimation$\dagger$, Gaussian Copula}\tabularnewline
\hline 
 & {\small{}$\gamma$} & {\small{}$\delta_{1}$} & {\small{}$\rho_{sp}$} & {\small{}$ATE$} &  & {\small{}$\gamma$} & {\small{}$\delta_{1}$} & {\small{}$\rho_{sp}$} & {\small{}$ATE$}\tabularnewline
\hline 
{\small{}True Values} & {\small{}0.8000} & {\small{}1.1000} & {\small{}0.5000} & {\small{}0.3242} & {\small{}True Values} & {\small{}0.8000} & {\small{}1.1000} & {\small{}0.5000} & {\small{}0.3242}\tabularnewline
\hline 
{\small{}Estimate} & {\small{}0.8124} & {\small{}1.1948} & {\small{}0.4913} & {\small{}0.3098} & {\small{}Estimate} & {\small{}0.8098} & {\small{}1.1957} & {\small{}0.4930} & {\small{}0.3252}\tabularnewline
\hline 
{\small{}S.D} & {\small{}0.1149} & {\small{}0.5540} & {\small{}0.1626} & {\small{}0.1068} & {\small{}S.D} & {\small{}0.1146} & {\small{}0.5905} & {\small{}0.1639} & {\small{}0.1107}\tabularnewline
\hline 
{\small{}Bias} & {\small{}0.0124} & {\small{}0.0948} & {\small{}-0.0086} & {\small{}-0.0143} & {\small{}Bias} & {\small{}0.0098} & {\small{}0.0957} & {\small{}-0.0069} & {\small{}0.0010}\tabularnewline
\hline 
{\small{}RMSE} & {\small{}0.0132} & {\small{}0.3069} & {\small{}0.0264} & {\small{}0.0116} & {\small{}RMSE} & {\small{}0.0131} & {\small{}0.3487} & {\small{}0.0269} & {\small{}0.0123}\tabularnewline
\hline 
\multicolumn{5}{|c||}{Parametric Estimation, Frank Copula} & \multicolumn{5}{c|}{Semiparametric Estimation$\dagger$, Frank Copula}\tabularnewline
\hline 
\end{tabular}{\small\par}
\par\end{centering}
\begin{centering}
\bigskip{}
\par\end{centering}
$\dagger$: The semiparametric models are specified with $G=F_{t_{3}}$,
where $F_{t_{3}}$ is the distribution function of $t(3)$. 
\end{table}
\begin{table}[H]
\caption{\label{tab:correct_02}Correct Specification ($n=500,\ \rho_{sp}=0.2$)
(True marginal: normal) }

\begin{centering}
\bigskip{}
\par\end{centering}
\centering{}{\small{}}%
\begin{tabular}{|c|c|c|c|c||c|c|c|c|c|}
\hline 
\multicolumn{5}{|c||}{Parametric Estimation, Gaussian Copula} & \multicolumn{5}{c|}{Semiparametric Estimation, Gaussian Copula}\tabularnewline
\hline 
 & {\small{}$\gamma$} & {\small{}$\delta_{1}$} & {\small{}$\rho_{sp}$} & {\small{}$ATE$} &  & {\small{}$\gamma$} & {\small{}$\delta_{1}$} & {\small{}$\rho_{sp}$} & {\small{}$ATE$}\tabularnewline
\hline 
{\small{}True Values} & {\small{}0.8000} & {\small{}1.1000} & {\small{}0.2000} & {\small{}0.3643} & {\small{}True Values} & {\small{}0.8000} & {\small{}1.1000} & {\small{}0.2000} & {\small{}0.3643}\tabularnewline
\hline 
{\small{}Estimate} & {\small{}0.8026} & {\small{}1.1342} & {\small{}0.2093} & {\small{}0.3643} & {\small{}Estimate} & {\small{}0.8026} & {\small{}1.1342} & {\small{}0.2093} & {\small{}0.3526}\tabularnewline
\hline 
{\small{}S.D} & {\small{}0.0945} & {\small{}0.4199} & {\small{}0.1840} & {\small{}0.0963} & {\small{}S.D} & {\small{}0.0945} & {\small{}0.4199} & {\small{}0.1840} & {\small{}0.0952}\tabularnewline
\hline 
{\small{}Bias} & {\small{}0.0026} & {\small{}0.0342} & {\small{}0.0093} & {\small{}0.0000} & {\small{}Bias} & {\small{}0.0026} & {\small{}0.0342} & {\small{}0.0093} & {\small{}-0.0117}\tabularnewline
\hline 
{\small{}RMSE} & {\small{}0.0089} & {\small{}0.1763} & {\small{}0.0339} & {\small{}0.0093} & {\small{}RMSE} & {\small{}0.0089} & {\small{}0.1763} & {\small{}0.0339} & {\small{}0.0092}\tabularnewline
\hline 
\multicolumn{5}{|c||}{Parametric Estimation, Frank Copula} & \multicolumn{5}{c|}{Semiparametric Estimation, Frank Copula}\tabularnewline
\hline 
\end{tabular}{\small\par}
\end{table}
\begin{table}[H]
\caption{\label{tab:mis_02}Misspecification of Marginals ($n=500,\ \rho_{sp}=0.2$)
(True marginal: mixture of normals)}

\begin{centering}
\bigskip{}
\par\end{centering}
\centering{}{\small{}}%
\begin{tabular}{|c|c|c|c|c||c|c|c|c|c|}
\hline 
\multicolumn{5}{|c||}{Parametric Estimation, Gaussian Copula} & \multicolumn{5}{c|}{Semiparametric Estimation, Gaussian Copula}\tabularnewline
\hline 
 & {\small{}$\gamma$} & {\small{}$\delta_{1}$} & {\small{}$\rho_{sp}$} & {\small{}$ATE$} &  & {\small{}$\gamma$} & {\small{}$\delta_{1}$} & {\small{}$\rho_{sp}$} & {\small{}$ATE$}\tabularnewline
\hline 
{\small{}True Values} & {\small{}0.8000} & {\small{}1.1000} & {\small{}0.2000} & {\small{}0.1066} & {\small{}True Values} & {\small{}0.8000} & {\small{}1.1000} & {\small{}0.2000} & {\small{}0.1066}\tabularnewline
\hline 
{\small{}Estimate} & {\small{}0.8038} & {\small{}0.9013} & {\small{}0.2088} & {\small{}0.2108} & {\small{}Estimate} & {\small{}0.8544} & {\small{}1.2755} & {\small{}0.1821} & {\small{}0.1256}\tabularnewline
\hline 
{\small{}S.D} & {\small{}0.1308} & {\small{}0.5666} & {\small{}0.1823} & {\small{}0.1137} & {\small{}S.D} & {\small{}0.1166} & {\small{}0.3865} & {\small{}0.1271} & {\small{}0.0638}\tabularnewline
\hline 
{\small{}Bias} & {\small{}0.0038} & {\small{}-0.1987} & {\small{}0.0088} & {\small{}0.1041} & {\small{}Bias} & {\small{}0.0544} & {\small{}0.1755} & {\small{}-0.0179} & {\small{}0.0190}\tabularnewline
\hline 
{\small{}RMSE} & {\small{}0.0171} & {\small{}0.3210} & {\small{}0.0332} & {\small{}0.0238} & {\small{}RMSE} & {\small{}0.0136} & {\small{}0.1494} & {\small{}0.0161} & {\small{}0.0044}\tabularnewline
\hline 
\multicolumn{5}{|c||}{Parametric Estimation, Frank Copula} & \multicolumn{5}{c|}{Semiparametric Estimation, Frank Copula}\tabularnewline
\hline 
\end{tabular}{\small\par}
\end{table}
\begin{table}[H]
\caption{\label{tab:correct_07}Correct Specification ($n=500,\ \rho_{sp}=0.7$)(True
marginal: normal) }

\begin{centering}
\bigskip{}
\par\end{centering}
\centering{}{\small{}}%
\begin{tabular}{|c|c|c|c|c||c|c|c|c|c|}
\hline 
\multicolumn{5}{|c||}{Parametric Estimation, Gaussian Copula} & \multicolumn{5}{c|}{Semiparametric Estimation, Gaussian Copula}\tabularnewline
\hline 
 & {\small{}$\gamma$} & {\small{}$\delta_{1}$} & {\small{}$\rho_{sp}$} & {\small{}$ATE$} &  & {\small{}$\gamma$} & {\small{}$\delta_{1}$} & {\small{}$\rho_{sp}$} & {\small{}$ATE$}\tabularnewline
\hline 
{\small{}True Values} & {\small{}0.8000} & {\small{}1.1000} & {\small{}0.7000} & {\small{}0.3643} & {\small{}True Values} & {\small{}0.8000} & {\small{}1.1000} & {\small{}0.7000} & {\small{}0.3643}\tabularnewline
\hline 
{\small{}Estimate} & {\small{}0.8032} & {\small{}1.1403} & {\small{}0.6979} & {\small{}0.3660} & {\small{}Estimate} & {\small{}0.8038} & {\small{}1.1475} & {\small{}0.7059} & {\small{}0.3615}\tabularnewline
\hline 
{\small{}S.D} & {\small{}0.0932} & {\small{}0.3663} & {\small{}0.1161} & {\small{}0.0860} & {\small{}S.D} & {\small{}0.0942} & {\small{}0.3909} & {\small{}0.1167} & {\small{}0.0928}\tabularnewline
\hline 
{\small{}Bias} & {\small{}0.0032} & {\small{}0.0403} & {\small{}-0.0020} & {\small{}0.0016} & {\small{}Bias} & {\small{}0.0038} & {\small{}0.0475} & {\small{}0.0060} & {\small{}-0.0028}\tabularnewline
\hline 
{\small{}RMSE} & {\small{}0.0087} & {\small{}0.1342} & {\small{}0.0135} & {\small{}0.0074} & {\small{}RMSE} & {\small{}0.0089} & {\small{}0.1528} & {\small{}0.0136} & {\small{}0.0086}\tabularnewline
\hline 
\multicolumn{5}{|c||}{Parametric Estimation, Frank Copula} & \multicolumn{5}{c|}{Semiparametric Estimation, Frank Copula}\tabularnewline
\hline 
\end{tabular}{\small\par}
\end{table}
\begin{table}[H]
\caption{\label{tab:mis_07}Misspecification of Marginals ($n=500,\ \rho_{sp}=0.7$)
(True marginal: mixture of normals)}

\begin{centering}
\bigskip{}
\par\end{centering}
\centering{}{\small{}}%
\begin{tabular}{|c|c|c|c|c||c|c|c|c|c|}
\hline 
\multicolumn{5}{|c||}{Parametric Estimation, Gaussian Copula} & \multicolumn{5}{c|}{Semiparametric Estimation, Gaussian Copula}\tabularnewline
\hline 
 & {\small{}$\gamma$} & {\small{}$\delta_{1}$} & {\small{}$\rho_{sp}$} & {\small{}$ATE$} &  & {\small{}$\gamma$} & {\small{}$\delta_{1}$} & {\small{}$\rho_{sp}$} & {\small{}$ATE$}\tabularnewline
\hline 
{\small{}True Values} & {\small{}0.8000} & {\small{}1.1000} & {\small{}0.7000} & {\small{}0.1066} & {\small{}True Values} & {\small{}0.8000} & {\small{}1.1000} & {\small{}0.7000} & {\small{}0.1066}\tabularnewline
\hline 
{\small{}Estimate} & {\small{}0.7942} & {\small{}1.1740} & {\small{}0.6323} & {\small{}0.2582} & {\small{}Estimate} & {\small{}0.8565} & {\small{}1.2619} & {\small{}0.6932} & {\small{}0.1252}\tabularnewline
\hline 
{\small{}S.D} & {\small{}0.1276} & {\small{}0.6180} & {\small{}0.1331} & {\small{}0.1080} & {\small{}S.D} & {\small{}0.1084} & {\small{}0.3714} & {\small{}0.0835} & {\small{}0.0661}\tabularnewline
\hline 
{\small{}Bias} & {\small{}-0.0058} & {\small{}0.0740} & {\small{}-0.0676} & {\small{}0.1515} & {\small{}Bias} & {\small{}0.0565} & {\small{}0.1619} & {\small{}-0.0068} & {\small{}0.0186}\tabularnewline
\hline 
{\small{}RMSE} & {\small{}0.0163} & {\small{}0.3820} & {\small{}0.0177} & {\small{}0.0346} & {\small{}RMSE} & {\small{}0.0118} & {\small{}0.1379} & {\small{}0.0070} & {\small{}0.0047}\tabularnewline
\hline 
\multicolumn{5}{|c||}{Parametric Estimation, Frank Copula} & \multicolumn{5}{c|}{Semiparametric Estimation, Frank Copula}\tabularnewline
\hline 
\end{tabular}{\small\par}
\end{table}
\begin{table}[H]
\caption{\label{tab:corr_neg}Correct Specification ($n=500,\ \rho_{sp}=-0.5$)(True
marginal: normal) }

\begin{centering}
\bigskip{}
\par\end{centering}
\centering{}{\small{}}%
\begin{tabular}{|c|c|c|c|c||c|c|c|c|c|}
\hline 
\multicolumn{5}{|c||}{Parametric Estimation, Gaussian Copula} & \multicolumn{5}{c|}{Semiparametric Estimation, Gaussian Copula}\tabularnewline
\hline 
 & {\small{}$\gamma$} & {\small{}$\delta_{1}$} & {\small{}$\rho_{sp}$} & {\small{}$ATE$} &  & {\small{}$\gamma$} & {\small{}$\delta_{1}$} & {\small{}$\rho_{sp}$} & {\small{}$ATE$}\tabularnewline
\hline 
{\small{}True Values} & {\small{}0.8000} & {\small{}1.1000} & {\small{}-0.5000} & {\small{}0.3643} & {\small{}True Values} & {\small{}0.8000} & {\small{}1.1000} & {\small{}-0.5000} & {\small{}0.3643}\tabularnewline
\hline 
{\small{}Estimate} & {\small{}0.8090} & {\small{}1.1134} & {\small{}-0.4912} & {\small{}0.3560} & {\small{}Estimate} & {\small{}0.8101} & {\small{}1.1164} & {\small{}-0.4822} & {\small{}0.3448}\tabularnewline
\hline 
{\small{}S.D} & {\small{}0.0970} & {\small{}0.4097} & {\small{}0.1727} & {\small{}0.0871} & {\small{}S.D} & {\small{}0.0974} & {\small{}0.4248} & {\small{}0.1708} & {\small{}0.0840}\tabularnewline
\hline 
{\small{}Bias} & {\small{}0.0090} & {\small{}0.0134} & {\small{}0.0088} & {\small{}-0.0084} & {\small{}Bias} & {\small{}0.0101} & {\small{}0.0164} & {\small{}0.0177} & {\small{}-0.0196}\tabularnewline
\hline 
{\small{}RMSE} & {\small{}0.0094} & {\small{}0.1678} & {\small{}0.0298} & {\small{}0.0077} & {\small{}RMSE} & {\small{}0.0095} & {\small{}0.1805} & {\small{}0.0292} & {\small{}0.0074}\tabularnewline
\hline 
\multicolumn{5}{|c||}{Parametric Estimation, Frank Copula} & \multicolumn{5}{c|}{Semiparametric Estimation, Frank Copula}\tabularnewline
\hline 
\end{tabular}{\small\par}
\end{table}
\begin{table}[H]
\caption{\label{tab:mis_neg}Misspecification of Marginals ($n=500,\ \rho_{sp}=-0.5$)
(True marginal: mixture of normals)}

\begin{centering}
\bigskip{}
\par\end{centering}
\centering{}{\small{}}%
\begin{tabular}{|c|c|c|c|c||c|c|c|c|c|}
\hline 
\multicolumn{5}{|c||}{Parametric Estimation, Gaussian Copula} & \multicolumn{5}{c|}{Semiparametric Estimation, Gaussian Copula}\tabularnewline
\hline 
 & {\small{}$\gamma$} & {\small{}$\delta_{1}$} & {\small{}$\rho_{sp}$} & {\small{}$ATE$} &  & {\small{}$\gamma$} & {\small{}$\delta_{1}$} & {\small{}$\rho_{sp}$} & {\small{}$ATE$}\tabularnewline
\hline 
{\small{}True Values} & {\small{}0.8000} & {\small{}1.1000} & {\small{}-0.5000} & {\small{}0.1066} & {\small{}True Values} & {\small{}0.8000} & {\small{}1.1000} & {\small{}-0.5000} & {\small{}0.1066}\tabularnewline
\hline 
{\small{}Estimate} & {\small{}0.8385} & {\small{}0.3931} & {\small{}-0.3171} & {\small{}0.1028} & {\small{}Estimate} & {\small{}0.8389} & {\small{}1.2832} & {\small{}-0.5237} & {\small{}0.1168}\tabularnewline
\hline 
{\small{}S.D} & {\small{}0.1301} & {\small{}0.4122} & {\small{}0.1665} & {\small{}0.1052} & {\small{}S.D} & {\small{}0.1123} & {\small{}0.4094} & {\small{}0.1232} & {\small{}0.0627}\tabularnewline
\hline 
{\small{}Bias} & {\small{}0.0385} & {\small{}-0.7069} & {\small{}0.1829} & {\small{}-0.0038} & {\small{}Bias} & {\small{}0.0389} & {\small{}0.1832} & {\small{}-0.0238} & {\small{}0.0101}\tabularnewline
\hline 
{\small{}RMSE} & {\small{}0.0169} & {\small{}0.1699} & {\small{}0.0277} & {\small{}0.0111} & {\small{}RMSE} & {\small{}0.0126} & {\small{}0.1676} & {\small{}0.0152} & {\small{}0.0040}\tabularnewline
\hline 
\multicolumn{5}{|c||}{Parametric Estimation, Frank Copula} & \multicolumn{5}{c|}{Semiparametric Estimation, Frank Copula}\tabularnewline
\hline 
\end{tabular}{\small\par}
\end{table}
\newpage{}

\begin{table}[H]
\caption{\label{tab:cp_boot}Coverage Probabilities of Bootstrap Confidence
Intervals (Nominal Level = 0.95)}

\begin{centering}
\bigskip{}
\par\end{centering}
\centering{}%
\begin{tabular}{ccc}
\hline 
 & Normal Approximation & Bootstrap Percentile\tabularnewline
\hline 
$ATE$ & 0.9050 & 0.9300\tabularnewline
$\alpha$ & 0.9700 & 0.9500\tabularnewline
$\gamma$ & 0.9600 & 0.9250\tabularnewline
$\beta$ & 0.9400 & 0.9300\tabularnewline
$\delta$ & 0.8750 & 0.9200\tabularnewline
$\rho$ & 0.9000 & 0.9500\tabularnewline
\hline 
\end{tabular}
\end{table}
\end{appendix}
\end{document}